\documentclass[bj,authoryear, noshowframe]{imsart}

\usepackage{graphicx} 
\usepackage{float}
\usepackage{amssymb,bbm}
\usepackage{amsfonts}
\usepackage{amsmath}
\usepackage{amsthm}
\usepackage{bm}
\usepackage{xcolor}
\usepackage{algorithm}
\usepackage{algpseudocode}
\usepackage{pdflscape}
\usepackage{xcolor}
\usepackage{amssymb}
\usepackage{float}
\usepackage{bm}
\usepackage{rotating}
\usepackage{physics}
\usepackage{subcaption}
\usepackage{caption}
\usepackage{mathrsfs}
\usepackage[capitalise]{cleveref} 
\usepackage{enumitem}
\usepackage{natbib}

\DeclareMathOperator{\SP}{SP}
\DeclareMathOperator{\AC}{AC}
\DeclareMathOperator{\HS}{HS}

\DeclareMathOperator{\essinf}{ess\,inf}

\DeclareMathOperator{\ESJD}{ESJD}
\DeclareMathOperator{\Var}{Var}
\DeclareMathOperator{\Gap}{Gap}
\DeclareMathOperator{\MTM}{MTM}
\DeclareMathOperator{\SMTM}{SMTM}

\newcommand{\dee }{ \text{d}}

\startlocaldefs

\theoremstyle{plain}

\newtheorem{theorem}{Theorem}[section]
\newtheorem{lemma}[theorem]{Lemma}
\newtheorem{proposition}[theorem]{Proposition}
\theoremstyle{definition}

\newtheorem{remark}[theorem]{Remark}

\endlocaldefs

\makeatletter
\renewcommand{\journal@name}{} 
\makeatother

\begin{document}

\begin{frontmatter}
\title{Stereographic Multiple-Try Metropolis}
\runtitle{Stereographic Multiple-Try Metropolis}

\begin{aug}
\author{\inits{F.}\fnms{Zhihao}~\snm{Wang}\ead[label=e1]{zw@math.ku.dk}\orcid{0009-0006-2794-6869}}
\author{\inits{S.}\fnms{Jun}~\snm{Yang}\ead[label=e2]{jy@math.ku.dk}}

\address{Department of Mathematical Sciences,
University of Copenhagen, Denmark
\printead[presep={,\ }]{e1,e2}}
\end{aug}

\begin{abstract}
Multiple-proposal MCMC algorithms have recently gained attention for their potential to improve performance, especially through parallel implementation on modern hardware. We introduce Stereographic Multiple-Try Metropolis (SMTM), a novel family of gradient-free algorithms designed for sampling high-dimensional distributions. By integrating multiple-try Metropolis (MTM) with the stereographic MCMC framework, SMTM overcomes the traditional limitations of MTM, particularly its pathological convergence behavior often observed in high dimensions. For both light-tailed and heavy-tailed targets, SMTM not only outperforms classical MTM and the existing stereographic random-walk Metropolis but also demonstrates strong robustness to tuning. These advantages are supported by high-dimensional scaling analysis and validated through extensive simulation studies.
\end{abstract}


\end{frontmatter}

\section{Introduction}
Markov chain Monte Carlo (MCMC) methods are essential tools for sampling from complex probability distributions, particularly in Bayesian inference and statistical physics. Their widespread use stems from their ability to approximate expectations under intractable distributions efficiently. In this paper, we focus on the multiple-try Metropolis (MTM) algorithms, first proposed by \cite{liu2000mtm}. Given a target distribution $\pi$ on $\mathbb{R}^d$, we want to construct a Markov chain $X^d(t)$ such that it has $\pi$ as its invariant distribution. Unlike the random-walk Metropolis (RWM) algorithm, which considers a single proposal per iteration, MTM generates multiple, say $N$, candidate proposals and selects one based on a weight function $\omega(\cdot,\cdot)$. See \cref{alg:MTM} for one version of MTM with a general weight function \citep{Pandolfi2010}. MTM has recently gained attention for its potential to facilitate parallel computation (see applications in \cite{frenkel2004,Calderhead2014,Holbrook2023}), enhancing computational efficiency. It does not require gradient information, making it applicable to a broader class of target distributions.

\begin{algorithm}
\caption{Multiple-try Metropolis (MTM)}
\label{alg:MTM}
\begin{algorithmic}
    \State \textbf{Input: }Current state $X^d(t)=x$, target $\pi$, proposal density $q(\cdot,\cdot)$ and weight function $\omega(\cdot,\cdot)$.
    \begin{itemize}
        \item Draw $y_1,\dots,y_N$ independently from $q(x,\cdot)$;
        \item Select one proposal from $y_1,\dots,y_N$, say $y_j$ with probability proportional to  $\omega(x,y_j)$;
        \item Draw $z_1,\dots,z_{N-1}$ independently from $q(y_j,\cdot)$;
        \item $X^d(t+1)=y_j$ with probability
        \begin{equation}
            \alpha(x,y_j)=1\wedge\frac{\pi(y_j)q(y_j,x)\omega(y_j,x)/(\sum_{i=1}^{N-1}\omega(y_j,z_i)+\omega(y_j,x))}{\pi(x)q(x,y_j)\omega(x,y_j)/(\sum_{i=1}^N\omega(x,y_i))};
        \end{equation}
        otherwise $X^d(t+1)=x$.
    \end{itemize}
\end{algorithmic}
\end{algorithm}

For the weight function, one popular choice is $\omega(x,y)=\pi(y)/\pi(x)$ where $x\in\mathbb{R}^d$ denotes the current state and $y\in\mathbb{R}^d$ denotes the proposal, and the corresponding MTM will be referred to as globally-balanced (GB) MTM as in \cite{Gagnon2023}. 
Recently, \cite{Gagnon2023} proposed the locally-balanced (LB) MTM where the weight function is chosen as $\omega(x,y)=\sqrt{\pi(y)/\pi(x)}$. For other versions of MTM: \cite{craiu2007} propose a modification of MTM allowing the use of correlated proposals; \cite{casarin2013interacting} use different independent distributions to sample proposals and introduce an interacting MTM mechanism; \cite{pandolfi2014generalized} generalize the reversible jump algorithm through the multiple-try method; \cite{Luo2018} propose an MTM algorithm with tailored proposal distributions; \cite{yang2019adaptive} introduce a component-wise multiple-try Metropolis algorithm whose computational efficiency is increased using an adaptation rule; \cite{fontaine2022} propose an adaptive MTM algorithm for target distributions with complex geometry such as multiple-modality; \cite{li2024} propose a rejection-free modification of MTM. We refer to \cite{martino2018} for a review of MTM algorithms. 

The properties of MTM have also been extensively studied: \cite{martino2017} describe the cases where the increase of the number of proposals does not produce a corresponding enhancement of the
performance and then introduce possible solutions for these issues; recently, \cite{pozza2024fundamental} investigated the tradeoff between the number of candidates and the computational efficiency of MTM. \cite{Gagnon2023} also study the convergence behavior of GB-MTM and LB-MTM when the number of proposals $N\to\infty$ and find that GB-MTM has a pathological behavior when the initial state is far away; the convergence rate of the MTM independent sampler is derived in \cite{Yang2023MTMIS}; the mixing time of MTM for Bayesian model selection is studied in \cite{chang2022} and the spectral gap for light-tailed targets is studied in \cite{pozza2024fundamental,caprio2025}.

For analyzing the asymptotic behavior of MCMC methods in high-dimensional settings, the optimal scaling framework, e.g., \cite{roberts97,roberts98,roberts01}, is a widely used and effective approach. The results of optimal scaling provide mathematically grounded insights that aid in optimizing MCMC performance by guiding the tuning of proposal distribution parameters in Metropolis–Hastings (MH) algorithms. A well-known example is the classical result recommending an acceptance probability of $0.234$ for RWM \citep{roberts97}. Optimal scaling for MTM has been studied and shown to outperform RWM in \cite{bedard2012scaling,Gagnon2023}. The optimal acceptance rate of MTM for fixed $N$ is known in \citet[Table 1]{bedard2012scaling}, for example, $0.32$ when $N=2$ and $0.37$ when $N=3$ for GB-MTM. In the case of $N\to\infty$, the optimal acceptance rate is proven in \citet[Corollary 1]{Gagnon2023} as 0.234 for GB-MTM and 0.574 for LB-MTM.

However, for heavy-tailed target distributions, it is well established that many popular Metropolis--Hastings algorithms lack geometric ergodicity \citep{jarner2000geometric, jarner2003necessary, jarner2007convergence, roberts1996geometric}. For instance, RWM is known to fail geometric ergodicity in \cite{jarner2000geometric}. Furthermore, as we later prove in \cref{prop:mtm non uniform}, MTM does not exhibit geometric ergodicity either. Recently, \cite{yang2024} introduced a new framework for gaining uniform ergodicity, a property stronger than geometric ergodicity, even for heavy-tailed distributions, by leveraging the concept of stereographic projection. One example they introduce is a variant of RWM, the Stereographic Projection Sampler (SPS) \cite[Algorithm 1]{yang2024}, which we refer to as the stereographic random-walk Metropolis (SRWM) in this paper. The established improvements of both MTM and SRWM over RWM motivate combining MTM with SRWM to develop an efficient sampling algorithm for both light-tailed and heavy-tailed targets.

In this paper, we propose stereographic multiple-try Metropolis (SMTM), a novel family of gradient-free MCMC algorithms that combine MTM with the stereographic MCMC framework, and establish several key theoretical results that highlight its advantages. We prove that, like SRWM, SMTM achieves uniform ergodicity for both light-tailed and heavy-tailed distributions, whereas MTM is not geometrically ergodic for any heavy-tailed distribution. Furthermore, a detailed high-dimensional scaling analysis shows that SMTM outperforms both MTM and SRWM. Moreover, we also prove that SMTM avoids the pathological convergence behavior often observed in MTM. Finally, numerical experiments demonstrate the strong robustness of SMTM to tuning. 

This paper is organized as follows. In \cref{Sec:SMTM}, we introduce SMTM and compare its ergodic properties with those of the MTM algorithm. \cref{Sec:optimal scaling} investigates the high-dimensional scaling limit of SMTM in the fixed 
$N$ regime, focusing on the limiting acceptance rate and expected squared jumping distance (ESJD) as the dimension $d\to\infty$. \cref{Sec:n to infty} turns to the large $N$ regime. We first provide partial results on the optimal acceptance rate in the limit $N\to\infty$. Then we demonstrate that SMTM avoids the pathological behavior known to affect classical MTM in the large $N$ regime. Finally, we discuss the open problem on how the spectral gap of SMTM scales with $N$ in a heavy-tailed setting. \cref{Sec:numerical} presents numerical experiments illustrating SMTM’s practical performance.

\section{Stereographic Multiple-Try Metropolis}
\label{Sec:SMTM}

\subsection{Stereographic MCMC}

The idea of stereographic MCMC \citep{yang2024} is to map the distribution from Euclidean space onto the sphere, construct a Markov chain on the sphere, and then project the chain back to Euclidean space. Stereographic projection plays a key role in this process. It is a classical mapping technique that projects points from a sphere onto a plane. Let $\mathbb{S}^d$ denote the unit sphere in $\mathbb{R}^{d+1}$ with $\mathscr{N}=(0,\dots,0,1)^T$ as its north pole. Then we can define the stereographic projection $\SP:\mathbb{S}^d\backslash \mathscr{N}\to\mathbb{R}^d$ as 
\begin{equation}\label{def_SP}
    x=\SP(z):=\left(R\frac{z_1}{1-z_{d+1}},\dots,R\frac{z_d}{1-z_{d+1}}\right)^T,
\end{equation}
where $R>0$ is a tuning parameter. Intuitively, $R$ acts as a scaling factor for the stereographic coordinates in $\mathbb{R}^d$ after projecting from the sphere. Increasing $R$ stretches the projected space, pushing points farther from the origin, while decreasing $R$ compresses the projection, bringing points closer to the origin. Although $R$ does not represent the radius of the sphere, which is always 1, we refer to $R$ as the radius parameter in this paper. The inverse of the stereographic projection $\SP^{-1}:\mathbb{R}^d\to\mathbb{S}^d\backslash \mathscr{N}$ can be written as 
\begin{equation}\label{def_invSP}
    z=\SP^{-1}(x)=\left(\frac{2Rx_1}{\|x\|^2+R^2},\dots,\frac{2Rx_d}{\|x\|^2+R^2},\frac{\|x\|^2-R^2}{\|x\|^2+R^2}\right)^T.    
\end{equation}

\begin{figure}[]
    \centering
    \includegraphics[width=0.6\linewidth]{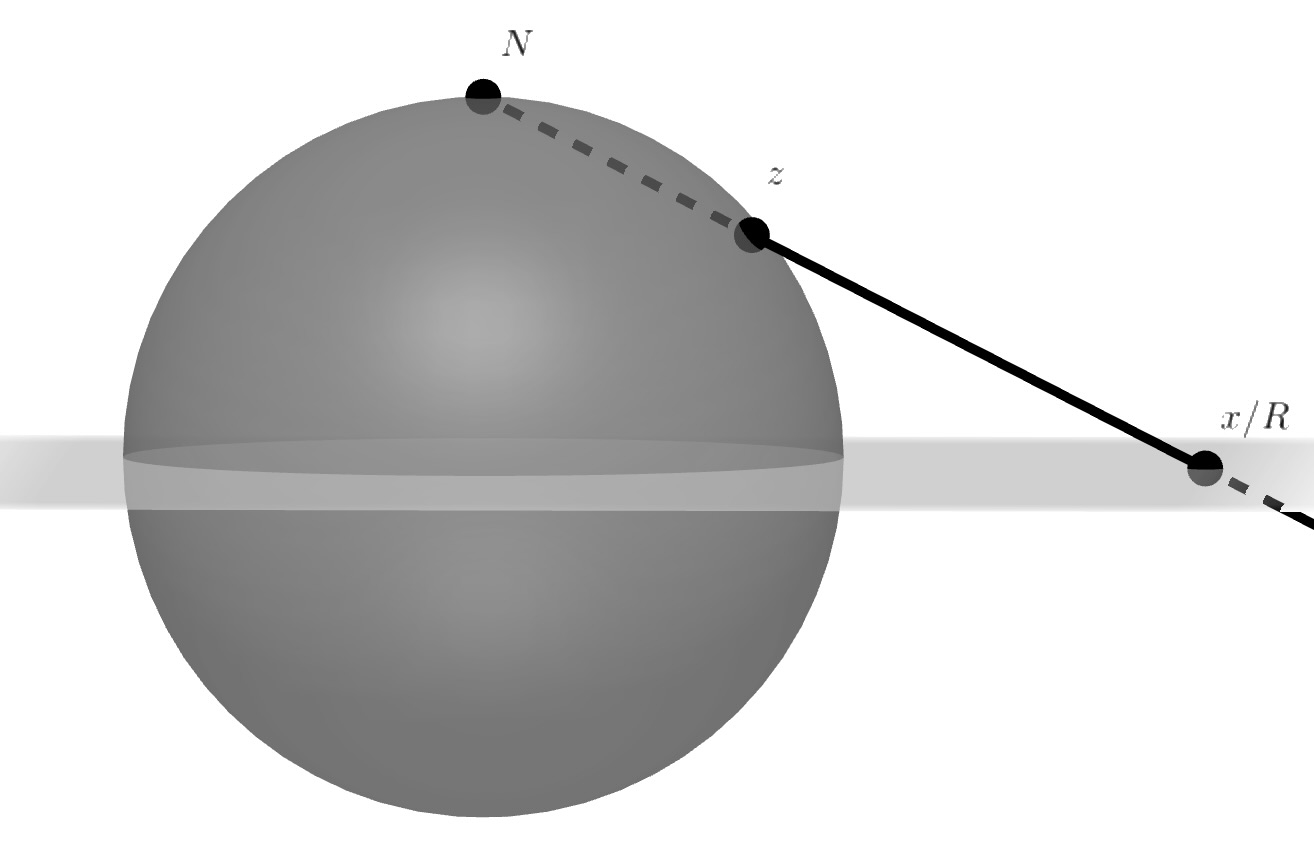}
    \caption{Stereographic Projection $\SP:\mathbb{S}^2\backslash \mathscr{N}\to\mathbb{R}^2$}
    \label{fig:stereographic projection}
\end{figure}

By \citet[Eq.(2)]{yang2024}, the Jacobian determinant at $x\in\mathbb{R}^d$ of $\SP$ satisfies
\begin{equation}
    J_{\SP}(x)\propto(R^2+\|x\|^2)^d.
\end{equation}
Given a target distribution $\pi(x)$ where $x\in\mathbb{R}^d$, we can transfer the target to a distribution on the sphere $\mathbb{S}^d$ using the Jacobian determinant, which is
\begin{equation}
    \pi_S(z)\propto\pi(x)(R^2+\|x\|^2)^d,
\end{equation}
where $z=\SP^{-1}(x)$. Therefore, one can construct an MCMC on the sphere targeting $\pi_S$ and then project the chain back to $\mathbb{R}^d$ \citep{yang2024}.

One such MCMC algorithm is SRWM; see \cref{alg:SPS}. Starting from $z\in\mathbb{S}^d$, SRWM first draws an isotropic Gaussian perturbation $\dee\tilde{z}$ in $\mathbb{R}^{d+1}$. The perturbation is then projected onto the tangent space of $z$, producing the tangent vector $\dee z$. Finally, $z+\dee z$ is normalized to obtain the proposal $\hat{z}\in\mathbb{S}^d$. SRWM demonstrates significantly better ergodic properties than RWM and has been proven to be uniformly ergodic for a large class of distributions whose tails are heavier than exponential (but no heavier than multivariate Student's $t$ with degrees of freedom $d$). In contrast, RWM fails to achieve even geometric ergodicity for any heavy-tailed distribution. Furthermore, SRWM consistently outperforms RWM for a stylized family of targets in optimal scaling theory with a larger ESJD \cite[Corollary 5.1]{yang2024} and is more robust to the choice of tuning parameters \citep{yang2024,bell2024adaptive}.

\begin{remark}
\label{remark:SP_extend}
    The stereographic projection can be trivially extended by introducing additional parameters, such as a ``location'' parameter and a preconditioning matrix. By \citet[Section 4]{yang2024}, given a vector $\mu\in\mathbb{R}^d$ and a positive definite matrix $\Sigma\in\mathbb{R}^{d\times d}$, a natural extension of $\SP$ can be defined by
    \begin{equation}
        \SP_{\mu,\Sigma}(z)=\Sigma^{1/2}\left(\frac{z_1}{1-z_{d+1}},\dots,\frac{z_d}{1-z_{d+1}}\right)^T+\mu.
    \end{equation}
    \begin{equation}
        \SP_{\mu,\Sigma}^{-1}(x)=\left(\frac{2\Sigma^{-1/2}(x_1-\mu_1)}{\|\Sigma^{-1/2}(x-\mu)\|^2+1},\dots,\frac{2\Sigma^{-1/2}(x_d-\mu_d)}{\|\Sigma^{-1/2}(x-\mu)\|^2+1},\frac{\|\Sigma^{-1/2}(x-\mu)\|^2-1}{\|\Sigma^{-1/2}(x-\mu)\|^2+1}\right)^T.
    \end{equation}
 
 This extension of stereographic projection integrates naturally into the adaptive MCMC framework; see \citet{bell2024adaptive} for principled strategies for adaptive parameter tuning.
\end{remark}

\begin{algorithm}[]
\caption{Stereographic Random-walk Metropolis (SRWM)}
\label{alg:SPS}
\begin{algorithmic}
    \State \textbf{Input: }Current state $X^d(t)=x$ and target distribution $\pi$.
    \begin{itemize}
        \item Let $z=\SP^{-1}(x)$; \hfill Project the current state onto $\mathbb{S}^d$
        \item Sample independently $\dee\tilde{z}\sim\mathcal{N}(0,h^2I_{d+1})$; \hfill Compute the proposal on $\mathbb{S}^d$
        \item Let $\dee z=\dee\tilde{z}-\frac{(z^T \dee\tilde{z})z}{\|z\|^2}$ and $\hat{z}=\frac{z+\dee z}{\|z+\dee z\|}$; 
        \item Define the proposal as $\hat{x}=\SP(\hat{z})$;\hfill Project the proposal back to $\mathbb{R}^d$
        \item $X^d(t+1)=\hat{x}$ with probability $1\wedge\frac{\pi(\hat{x})(R^2+\|\hat{x}\|^2)^d}{\pi(x)(R^2+\|x\|^2)^d}$; otherwise $X^d(t+1)=x$.
    \end{itemize}
\end{algorithmic}
\end{algorithm}

\subsection{Stereographic MTM}
The SMTM algorithm integrates the principles of MTM and SRWM. Let $Q_S(\cdot, \cdot)$ denote the proposal distribution on $\mathbb{S}^d$ of SRWM (steps 2--3 of \cref{alg:SPS}), i.e., $Q_S(z, \cdot)$ is the distribution of $\hat{z}$ in the third step of \cref{alg:SPS} when $z$ is given. Here we give details of SMTM, which is in \cref{alg:SMTM}. \cref{alg:SMTM} is stated for $N>1$ because its multiple-try construction includes $N-1$ auxiliary proposals. Nevertheless, if the construction is formally extended to $N=1$, there is only one candidate; it is selected with probability one, and the auxiliary set is empty. The resulting acceptance probability reduces to that of SRWM in \cref{alg:SPS}. In this precise sense, we regard SRWM as the $N=1$ special case of SMTM, while reserving the name SMTM for $N>1$.
Similar to \citet[Section 1.2]{Gagnon2023}, we call the weight function $\omega(z,\hat{z})=\pi_S(\hat{z})/\pi_S(z)$ globally-balanced and $\omega(z,\hat{z})=\sqrt{\pi_S(\hat{z})/\pi_S(z)}$ locally-balanced. Note that the weight functions defined here differ from those in \cite{Gagnon2023}, as ours are defined on the sphere. More generally, a locally balanced weight can be written as $\omega_g(z,\hat{z})=g\left(\pi_S(\hat{z})/\pi_S(z)\right)$, where $g(t)=tg(1/t)$ \citep[Theorem 1]{zanella2020informed}. Examples include $g(t)=\sqrt{t}$, $g(t)=t/(1+t)$ and $g(t)=1\wedge t$. \cref{alg:SMTM} can be implemented with any suitable positive weight function. However, studies of Euclidean MTM indicate that $g(t)=\sqrt{t}$ has an advantage over $g(t)=t/(1+t)$ in terms of convergence speed-ups \citep{Gagnon2023,zhou2022rapid,chang2022}, while in stationarity, the performance of the two functions is similar. Thus, we follow the recommendation of \citet[Section 2.1]{Gagnon2023} and specialize to $g(t)=\sqrt{t}$ throughout this paper.

Because the stereographic projection and its Jacobian determinant admit closed-form expressions, they do not incur additional computational cost. Therefore, with the same $N$, the computational complexity of SMTM and MTM is similar in each iteration. Next, we show SMTM is a valid MCMC algorithm with invariant distribution $\pi$.

\begin{algorithm}[]
\caption{Stereographic Multiple-try Metropolis (SMTM)}
\label{alg:SMTM}
\begin{algorithmic}
    \State \textbf{Input:} Current state $X^d(t)=x$, $N>1$, target $\pi$ on $\mathbb{R}^d$ (or $\pi_S$ on $\mathbb{S}^d$) and weight function $\omega(\cdot,\cdot)$.
    \begin{itemize}
        \item Let $z=\SP^{-1}(x)$;\hfill Project the current state onto $\mathbb{S}^d$
        \item Draw $\hat{z}_1, \dots, \hat{z}_N$ independently from $Q_S(z, \cdot)$;\hfill Compute the candidates on $\mathbb{S}^d$
        \item Select $\hat{z} = \hat{z}_j$ with probability proportional to $\omega(z, \hat{z}_j)$;\hfill Select one proposal
        \item Draw $z_1^{\ast}, \dots, z_{N-1}^{\ast}$ independently from $Q_S(\hat{z}_j, \cdot)$;\hfill Compute the auxiliary variables
        \item Define the proposal as $\hat{x}_j=\SP(\hat{z}_j)$;\hfill Project the proposal back to $\mathbb{R}^d$
        \item $X^d(t+1)=\hat{x}_j$ with probability
        \begin{equation}\label{eq:SMTM_alpha1}
            \alpha(z,\hat{z}_j)=1\wedge\frac{\pi_S(\hat{z}_j)\omega(\hat{z}_j,z)/(\sum_{i=1}^{N-1}\omega(\hat{z}_j,z_i^{\ast})+\omega(\hat{z}_j,z))}{\pi_S(z)\omega(z,\hat{z}_j)/(\sum_{i=1}^N\omega(z,\hat{z}_i))};
        \end{equation}
        otherwise $X^d(t+1)=x$.
    \end{itemize}
\end{algorithmic}
\end{algorithm}

\begin{proposition}
\label{prop:detailed balance}
    Let $\pi$ be positive and continuous on $\mathbb{R}^d$, $\sup_{x\in\mathbb{R}^d}\pi(x)(R^2+\|x\|^2)^d<\infty$, then SMTM with fixed $N$ and globally or locally-balanced weight function induces an ergodic Markov chain with $\pi$ as its invariant distribution.
\end{proposition}
\begin{proof}
    See \cref{proof:prop:detailed balance}.
\end{proof}

Next, we compare the ergodic properties of SMTM and MTM for heavy-tailed distributions. We show that MTM lacks geometric ergodicity for any heavy-tailed distribution, whereas SMTM can achieve uniform ergodicity under certain conditions. Recall that a Markov chain with transition kernel $P$ is said to be geometrically ergodic if for any $x$ and $n\in\mathbb{N}$ we have $\|P^n(x,\cdot)-\pi(\cdot)\|_{\text{TV}}\leq C(x)r^n$ with constant $r\in(0,1)$, where $\|\cdot\|_{\text{TV}}$ represents the total variation. A Markov chain is said to be uniformly ergodic if for any $\epsilon>0$, there exists $N\in\mathbb{N}$ such that $\|P^N(x,\cdot)-\pi(\cdot)\|_{\text{TV}}\leq\epsilon,\forall x$. 
The geometric and uniform ergodicity of RWM has been extensively studied. For heavy-tailed target distributions, it is well known that RWM is not geometrically ergodic \citep{jarner2000geometric, jarner2003necessary, jarner2007convergence, roberts1996geometric}. Although one might intuitively expect that MTM could improve convergence over RWM, it is not expected to achieve geometric ergodicity under such conditions. However, to the best of our knowledge, no rigorous result confirming this has been established in the existing literature. In the following proposition, we provide a formal proof, following similar arguments as in \cite{jarner2000geometric}, showing that random-walk-based MTM fails to be geometrically ergodic when the target distribution is heavy-tailed.
\begin{proposition}
\label{prop:mtm non uniform}
    Suppose $q$ is a random-walk-based proposal distribution, i.e., $q(x,y) = q(\|x-y\|)$, satisfying $\int_{\mathbb{R}^d} \|x\| q(\|x\|)\dee x < \infty$. Suppose an MTM algorithm with a fixed number of candidates $N$, proposal distribution $q$, target distribution $\pi$, and weight function $\omega$, is ergodic. If it is geometrically ergodic, then there exists $s>0$ such that
    \begin{equation}
    \label{Eq:finite exponential moments}
        \int_{\mathbb{R}^d}e^{s\|x\|}\pi(x)\dee x<\infty.
    \end{equation}
\end{proposition}
\begin{proof}
    See \cref{proof:prop:mtm non uniform}.
\end{proof}

When the target distribution $\pi$ is heavy-tailed, such as multivariate Student's $t$ distributions with any degrees of freedom, \cref{Eq:finite exponential moments} does not hold anymore, which implies that MTM is not geometrically ergodic. Meanwhile, due to the compactness of the sphere, it is proven in \citet[Theorem 2.1]{yang2024} that SRWM can be uniformly ergodic for certain heavy-tailed targets. We next show that SMTM retains uniform ergodicity in such cases.
\begin{theorem}
\label{Thm:uniform ergodic}
    If $\pi(x)$ is continuous and positive in $\mathbb{R}^d$, $\sup_{x\in\mathbb{R}^d}\pi(x)(R^2+\|x\|^2)^d<\infty$, then SMTM, with fixed $N$ and globally or locally-balanced weight function, is uniformly ergodic.
\end{theorem}
\begin{proof}
    See \cref{proof:Thm:uniform ergodic}.
\end{proof}
\begin{remark}
    \cref{Thm:uniform ergodic} shows that SMTM is uniformly ergodic if $\pi$ is a multivariate Student's $t$ distribution with degrees of freedom no less than $d$. To achieve uniform ergodicity for target distributions with heavier tails, two approaches are available. First, one could apply a composition of the variable transformations proposed by \cite{johnson2012variable}, which transfers the target $\pi$ to a distribution that satisfies the conditions of \cref{Thm:uniform ergodic}, and the stereographic projection. Alternatively, one could leverage recent extensions of the stereographic projection by \citet{grazzi2026sub}, which directly establish uniform ergodicity for \emph{sub-Cauchy} distributions. In this paper, we focus exclusively on the traditional stereographic projection and leave a detailed investigation of these alternative approaches to future work.
\end{remark}

\subsection{Practical Implementation and Tuning}
\label{Sec:tuning}
For the general stereographic projection introduced in \cref{remark:SP_extend}, SMTM requires the number of candidates $N$, the step size of the spherical proposal $h$, the stereographic location $\mu$ and the positive-definite preconditioning matrix $\Sigma$. We let $N$ be fixed as the design choice before a run (\cref{app:choose N} provides some instructions on how to choose $N$), while $(\mu,\Sigma,h)$ may be tuned during a discarded warm-up period. The theoretical results in \cref{Sec:SMTM,Sec:optimal scaling,Sec:n to infty} concern fixed tuning parameters. In the workflow below, adaptation is stopped before samples are retained.

The Adapting Increasingly Rarely (AIR) framework, introduced by \citet{chimisov2018air}, provides a structurally simplified approach to adaptive MCMC. In the AIR framework, the underlying Markov kernel is updated based on the entire available chain output, but these adaptations only occur at specific time points separated by an increasing number of iterations. Within the framework of AIR, $(\mu,\Sigma,h)$ can be tuned adaptively to improve the geometric alignment between the target distribution and the sphere \cite[Section 3.2]{bell2024adaptive}.

At an adaptation, given warm-up states $X_1,\cdots, X_m$, let
\begin{equation}
    \hat{\mu}=\bar{X},\quad\hat{C}=\frac{1}{m-1}\sum_{i=1}^m(X_i-\bar{X})(X_i-\bar{X})^T.
\end{equation}
To ensure that the empirical covariance matrix is positive definite, we use the shrinkage estimator
\begin{equation}
    \hat{C}_{\epsilon}=(1-\epsilon)\hat{C}+\epsilon\frac{\text{tr}(\hat{C})}{d}I_d,\quad 0<\epsilon<1,
\end{equation}
where $\epsilon$ is a small fixed regularization constant. The identity component ensures positive definiteness, while the factor
$\text{tr}(\hat{C})/d$ preserves the average marginal scale. We then set
\begin{equation}
    \hat{\Sigma}=d\hat{C}_{\epsilon}.
\end{equation}
The step size $h$ can be updated based on the optimal acceptance rate (see \cref{table:optimal_parameters}) for example in a log scale way
\begin{equation}
    \log h_{k+1}=\log h_k+\eta_k(\hat{a}_k-a^\ast),\quad \eta_k>0,
\end{equation}
where $\hat{a}_k$ is the empirical acceptance rate in the most recent epoch and $a^{\ast}$ is the optimal acceptance rate.

The practical tuning of SMTM is given in \cref{alg:SMTM_tuned}. The discussion following \cref{Thm:maxesjd} provides a theoretical interpretation of these updates. 
\begin{algorithm}[H]
\caption{Practical Tuning of SMTM}
\label{alg:SMTM_tuned}
\begin{algorithmic}
    \State \textbf{Input:} Initial state $x_0$, initial $(\mu_0,\Sigma_0,h_0)$, $N>1$, target $\pi$ on $\mathbb{R}^d$ (or $\pi_S$ on $\mathbb{S}^d$), optimal acceptance rate $a^\ast$, increasing epoch lengths $L_k$ with $k\leq B$ and weight function $\omega(\cdot,\cdot)$.
    \begin{itemize}
        \item For epochs $k=1,2,\dots$,
        \begin{itemize}
            \item run SMTM for $L_k$ iterations with the current $(\mu_k,\Sigma_k,h_k)$;
            \item update $\mu_{k+1}=\hat{\mu}_k$ and $\Sigma_{k+1}=\hat{\Sigma}_k$ from the warm-up window;
            \item update $h_{k+1}$ based on the current acceptance rate towards $a^\ast$;
            \item update the spherical representation of the current Euclidean state using the new $(\mu_{k+1},\Sigma_{k+1})$.
        \end{itemize}
        \item At iteration $B$, freeze $(\mu,\Sigma,h)$ and rerun the SMTM with the fixed kernel.
    \end{itemize}
\end{algorithmic}
\end{algorithm}

\section{Optimal Scaling for SMTM}
\label{Sec:optimal scaling}

In the context of MCMC methods, optimal scaling \citep{roberts97,roberts01} addresses the problem of choosing a proposal distribution that results in efficient sampling from the target distribution. Practical implementations of the MH algorithm often face challenges related to proposal step size: when the proposed jumps are too small, the Markov chain explores the state space very slowly, leading to inefficient sampling; if the proposed jumps are too large, the chain frequently proposes moves into low-probability regions, resulting in high rejection rates and long periods of the chain staying in the same state. Optimal scaling aims to find a balance: the optimal choice of proposal step size that maximizes the efficiency of the algorithm for exploring the state space.

A powerful theoretical tool used to study optimal scaling of MCMC is the analysis of diffusion limits \citep{roberts97}. This approach involves considering a sequence of target distributions in increasing dimension and analyzing the limiting behavior of the corresponding Markov chains. Another approach to finding the optimal scaling is to maximize ESJD of the chain. Given a Markov chain $\{X^d(t)\}$ with a stationary distribution $\pi$, the ESJD is defined as
\begin{equation}
\label{Eq:original esjd}
    \ESJD:=\mathbb{E}_{X^d(t)\sim\pi}\mathbb{E}_{X^d(t+1)\,|\,X^d(t)}\left[\|X^d(t+1)-X^d(t)\|^2\right].
\end{equation}

The justification for using ESJD as an optimization criterion comes from its connection to diffusion limits. When weak convergence to a diffusion limit is established, the ESJD of the discrete-time Markov chain converges to the quadratic variation of the limiting diffusion \citep{roberts01,Yang2020}. A larger ESJD generally indicates that the Markov chain is exploring the state space more efficiently, avoiding both overly conservative moves and excessive rejection. For instance, as demonstrated in \citet[Section 5.4]{yang2024}, when the same optimal acceptance rate is maintained, the maximum ESJD achieved by SRWM exceeds that of RWM. This suggests that SRWM offers improved sampling efficiency compared to RWM.

In most existing literature of optimal scaling, the standard assumption for the target distribution $\pi$ is to have a product i.i.d.\,form \citep{roberts97}, that is
\begin{equation}\label{eq:iid}
    \pi(x)=\prod_{i=1}^d f(x_i),
\end{equation}
where $f$ is a one-dimensional probability density function. Recent work has shown that the product i.i.d.\,assumption can be relaxed \citep{Yang2020, Ning2024}. However, in this paper, we will restrict our attention to the case of product i.i.d.\,target distributions for ease of analysis.

In this section, we present our main results on the optimal scaling of SMTM with fixed $N$. We analyze the limiting acceptance rate and ESJD in \cref{Thm:accept rate} and \cref{Thm:esjd}, respectively. 
Other than \cref{eq:iid}, we adopt the same additional assumptions on the target density as outlined in \citet[Section 5.1]{yang2024}. Specifically, we assume $X\sim f$ satisfies
\begin{equation}\label{eq_finite_moment}
    \mathbb{E}_f\left[X^2\right]=\int x^2f(x)\dee x=1,\quad\mathbb{E}_f\left[X^6\right]<\infty.
\end{equation}
Furthermore, we assume that $(\log f)'$ and $(\log f)''$ are Lipschitz continuous, $\lim_{x\to\pm\infty}xf'(x)=0$, and 
\begin{equation}
    \mathbb{E}_f\left[\left(\frac{f'(X)}{f(X)}\right)^8\right]<\infty,\quad\mathbb{E}_f\left[\left(\frac{f''(X)}{f(X)}\right)^4\right]<\infty,\quad\mathbb{E}_f\left[\left(\frac{Xf'(X)}{f(X)}\right)^4\right]<\infty.
\end{equation}
Notably, product i.i.d.\,distributions under these assumptions include a range of heavy-tailed distributions. For example, $f$ can be the density of a Student’s $t$ distribution with degrees of freedom greater than $6$. As discussed in \citet[Section 5.2]{yang2024}, under the assumptions on $f$ in \cref{eq_finite_moment}, effective application of the SRWM algorithm requires scaling $R$ to be $O(d^{1/2})$, otherwise the transformed target will finally be concentrated at the north or south poles of the sphere. Thus, we assume that $R=\sqrt{\lambda d}$ for a fixed $\lambda>0$.

To study the limiting acceptance probability of SMTM, we focus on two related quantities. First, conditional on a candidate $\hat{z}_j$ having already been selected from the proposal pool, we consider the probability that this candidate $\hat{z}_j$ is accepted:
\begin{equation}
    \alpha_1^j:=\mathbb{P}(\text{accept }\hat{z}_j\mid \text{choose }\hat{z}_j, z).
\end{equation}
Clearly, $\alpha_1^j$ measures how likely the algorithm is to accept $\hat{z}_j$ once it has been chosen. Second, closely related to $\alpha_1^j$, we consider the probability that $\hat{z}_{j}$ is first selected from the proposal pool and subsequently accepted:
\begin{equation}
    \alpha_2^j:=\mathbb{P}(\text{choose }\hat{z}_j, \text{ accept }\hat{z}_j\mid z).
\end{equation}
Note that $\sum_j \alpha_2^j$ represents the probability that a proposal from the pool is accepted, which will be referred to as the (total) acceptance rate of SMTM.
Note that both $\alpha_1^j$ and $\alpha_2^j$ depend on the weight function. Therefore, we give their explicit expressions under both globally-balanced and locally-balanced weight functions in the following.
\begin{lemma}
\label{lemma:alpha computation}
    In the globally-balanced case, we have
    \begin{equation}
        \alpha_1^j=1\wedge\frac{\sum_{i=1}^N\frac{\pi_S(\hat{z}_i)}{\pi_S(z)}}{\sum_{i=1}^{N-1}\frac{\pi_S(z_i^{\ast})}{\pi_S(\hat{z}_j)}\frac{\pi_S(\hat{z}_j)}{\pi_S(z)}+1},\quad\alpha_2^j=\frac{\frac{\pi_S(\hat{z}_j)}{\pi_S(z)}}{\sum_{i=1}^N\frac{\pi_S(\hat{z}_i)}{\pi_S(z)}}\wedge\frac{\frac{\pi_S(\hat{z}_j)}{\pi_S(z)}}{\sum_{i=1}^{N-1}\frac{\pi_S(z_i^{\ast})}{\pi_S(\hat{z}_j)}\frac{\pi_S(\hat{z}_j)}{\pi_S(z)}+1};
    \end{equation}
    And in the locally-balanced case, we have
    \begin{equation}
    \label{Eq:LB alpha}
        \alpha_1^j=1\wedge\frac{\sum_{i=1}^N\sqrt{\frac{\pi_S(\hat{z}_i)}{\pi_S(z)}\frac{\pi_S(\hat{z}_j)}{\pi_S(z)}}}{\sum_{i=1}^{N-1}\sqrt{\frac{\pi_S(z_i^{\ast})}{\pi_S(\hat{z}_j)}\frac{\pi_S(\hat{z}_j)}{\pi_S(z)}}+1},\quad\alpha_2^j=\frac{\frac{\pi_S(\hat{z}_j)}{\pi_S(z)}}{\sum_{i=1}^N\sqrt{\frac{\pi_S(\hat{z}_i)}{\pi_S(z)}\frac{\pi_S(\hat{z}_j)}{\pi_S(z)}}}\wedge\frac{\frac{\pi_S(\hat{z}_j)}{\pi_S(z)}}{\sum_{i=1}^{N-1}\sqrt{\frac{\pi_S(z_i^{\ast})}{\pi_S(\hat{z}_j)}\frac{\pi_S(\hat{z}_j)}{\pi_S(z)}}+1}.
    \end{equation}
\end{lemma}
\begin{proof}
    See \cref{proof:lamma:alpha computation}.
\end{proof}

To express the scaling limits of acceptance rate and ESJD more concisely, we define the following two functions which will later be used to approximate $\alpha_1^j$ and $\alpha_2^j$, respectively: for each $j\in\{1,\dots,N\}$, define $\phi_1^j,\phi_2^j:\mathbb{R}^{2N-1}\to\mathbb{R}$, where $N$ is the number of candidates in each step, by
\begin{equation}
\label{eq:phi 1}
    \phi_1^j\left((x_i)_{i=1}^{N},(y_i)_{i=1}^{N-1}\right):=  
    \left\{
        \begin{aligned}
            &1\wedge\frac{\sum_{i=1}^{N}e^{x_i}}{\sum_{i=1}^{N-1}e^{x_j}e^{y_i}+1}, & \text{globally-balanced},\\
            &1\wedge\frac{\sum_{i=1}^N\sqrt{e^{x_i}e^{x_j}}}{\sum_{i=1}^{N-1}\sqrt{e^{x_j}e^{y_i}}+1}, & \text{locally-balanced}.
        \end{aligned}
        \right.
\end{equation}
\begin{equation}
\label{eq:phi 2}
    \phi_2^j\left((x_i)_{i=1}^{N},(y_i)_{i=1}^{N-1}\right):=\left\{
        \begin{aligned}
            &\frac{e^{x_j}}{\sum_{i=1}^{N}e^{x_i}}\wedge\frac{e^{x_j}}{\sum_{i=1}^{N-1}e^{x_j}e^{y_i}+1}, & \text{globally-balanced},\\
            &\frac{e^{x_j}}{\sum_{i=1}^N\sqrt{e^{x_i}e^{x_j}}}\wedge\frac{
        e^{x_j}}{\sum_{i=1}^{N-1}\sqrt{e^{x_j}e^{y_i}}+1}, & \text{locally-balanced}.
        \end{aligned}
        \right.
\end{equation}
The exponential terms in $\phi_1^j$ and $\phi_2^j$ arise from later approximations of the logarithm of the acceptance probabilities. Note that when $N=1$, both $\phi_1^j$ and $\phi_2^j$ reduce to the function $1 \wedge e^x$, which plays a central role in optimal scaling results for RWM \citep{roberts97}. In this sense, $\phi_1^j$ and $\phi_2^j$ can be viewed as natural extensions of this function to the multiple-proposal setting.
Next, we establish a few properties of these functions, such as boundedness and Lipschitz continuity.
\begin{lemma}
\label{Lemma:lipschitz}
    For any $j\in\{1,\dots,N\}$, $\phi_1^j,\phi_2^j$ are all bounded and Lipschitz continuous.
\end{lemma}
\begin{proof}
    See \cref{proof:Lemma:lipschitz}.
\end{proof}

Next, we establish the high-dimensional scaling limits of both $\alpha_1^j$ and $\alpha_2^j$ using functions $\phi_1^j$ and $\phi_2^j$, respectively. The results will be used to approximate the limiting ESJD of SMTM. Note that the convergence we establish is in the $L^1$ norm, which is stronger than the weak convergence typically used in existing literature on optimal scaling analysis.
\begin{theorem}
\label{Thm:accept rate}
    Under the assumptions on $\pi$ at the beginning of \cref{Sec:optimal scaling}, suppose the current state of SMTM on $\mathbb{S}^d$ is $z=\SP^{-1}(X)$ and the radius parameter of the SMTM  is $R=\sqrt{\lambda d}$, where $\lambda>0$ is a fixed constant. Then, if either $\lambda\neq1$ or $f$ is not the standard Gaussian density, there exists a sequence of sets $\{H_d\}$ in $\mathbb{R}^d$ such that $\pi(H_d)\to1$ and for any $j\in\{1,\dots,N\}$,
    \begin{equation}
        \begin{split}
            \sup_{\SP(z)\in H_d}\mathbb{E}\left[\left|\alpha_1^j-\phi_1^j((W_i)_{i=1}^{N},(V_i)_{i=1}^{N-1})\right|\middle| z\right]=o(d^{-1/4}\log(d)),\\
            \sup_{\SP(z)\in H_d}\mathbb{E}\left[\left|\alpha_2^j-\phi_2^j((W_i)_{i=1}^{N},(V_i)_{i=1}^{N-1})\right|\middle| z\right]=o(d^{-1/4}\log(d)),
        \end{split}
    \end{equation}
    where $(W_i)_{i=1}^{N}$ and $(V_i)_{i=1}^{N-1}$ are conditionally i.i.d.\,random variables with distribution $\mathcal{N}(\mu,\sigma^2)$,
    \begin{equation}
    \label{Eq:normal distribution}
        \mu=\frac{\ell^2}{2}\left(\frac{4\lambda}{(1+\lambda)^2}-\mathbb{E}[((\log f)')^2]\right),\quad\sigma^2=\ell^2\left(\mathbb{E}[((\log f)')^2]-\frac{4\lambda}{(1+\lambda)^2}\right),
    \end{equation}
    and $\ell$ is a re-parametrization of $h$ satisfying
    \begin{equation}
    \label{Eq:reparametrization}
        \frac{1}{\sqrt{1+h^2(d-1)}}=1-\frac{\ell^2}{2d}\frac{4\lambda}{(1+\lambda)^2}.
    \end{equation}
\end{theorem}
\begin{proof}
    See \cref{proof:Thm:accept rate}.
\end{proof}

\begin{remark} Note that \cref{Thm:accept rate} generalizes \citet[Lemma 5.1]{yang2024}.
    When $N=1$, $\phi_1^1(x)=\phi_2^1(x)=1\wedge e^x$, SMTM in \cref{alg:SMTM} reduces to SRWM in \cref{alg:SPS} and \cref{Thm:accept rate} reduces to \citet[Lemma 5.1]{yang2024}.
\end{remark}

Establishing diffusion limits for stereographic MCMC algorithms presents significant technical challenges. For instance, to establish such a limit in the original stereographic MCMC paper, \citet[Theorem 5.2]{yang2024} had to introduce a revised version of SRWM. However, they also demonstrated that the limiting ESJD can be readily established, providing a valuable diagnostic for understanding algorithmic performance and guiding parameter tuning. Motivated by this precedent, we adopt the same perspective in this paper, utilizing ESJD as a key analytical tool in our study of SMTM. Since the proposals in SMTM are conditionally i.i.d., the original ESJD in \cref{Eq:original esjd} can be rewritten as
\begin{equation}
    \ESJD_{\SMTM}=\mathbb{E}_{X\sim\pi}\left[\sum_{j=1}^N\mathbb{E}_{\hat{X}_j\mid X}\left[\|\hat{X}_j-X\|^2\alpha_2^j\middle| X\right]\right]=N\mathbb{E}_{X\sim\pi}\left[\mathbb{E}_{\hat{X}_1\mid X}\left[\|\hat{X}_1-X\|^2\alpha_2^1\middle| X\right]\right].
\end{equation}
Using $\mathbb{E}[\|\hat{X}_1-X\|^2]\to\ell^2$ and \cref{Thm:accept rate}, we can obtain a widely used approximation of the limiting ESJD of SMTM:
\begin{equation}
\label{Eq:approx esjd}
    \begin{aligned}
        \ESJD_{\SMTM}&\approx N\mathbb{E}\left[\|\hat{X}_1-X\|^2\right]\cdot\mathbb{E}\left[\alpha_2^1\right]\\
        &\xrightarrow{d\to\infty}N\ell^2\mathbb{E}_{X\sim\pi}\left[\mathbb{E}\left[\phi_2^1((W_i)_{i=1}^N,(V_i)_{i=1}^{N-1})\middle| X\right]\right]=:\widetilde{\ESJD}.
    \end{aligned}
\end{equation}

Before using this approximation to compare SMTM with MTM, we establish that it coincides with the limiting ESJD under the well-scaled choice $\lambda=1$, or equivalently $R=\sqrt d$.
\begin{theorem}
\label{Thm:esjd}
    Under the assumptions on $\pi$ in \cref{Sec:optimal scaling}, suppose that the current state $X$ is in the stationary phase, the radius parameter of SMTM is chosen as $R=\sqrt{d}$ and $f$ is not the probability density function of the standard Gaussian distribution, then 
    \begin{equation}
        \lim_{d\to\infty}\ESJD_{\SMTM}(\ell)=\widetilde{\ESJD}(\ell),
    \end{equation}
    for any given $\ell$, where $(W_i)_{i=1}^{N}$ and $(V_i)_{i=1}^{N-1}$ are conditionally normal i.i.d.\,random variables with distribution in \cref{Eq:normal distribution}.
\end{theorem}
\begin{proof}
    See \cref{proof:thm:esjd}.
\end{proof}

Although explicit expressions for the expectations in \cref{Eq:approx esjd} in the general case are not available, we can still compare the maximum (approximate) ESJD of SMTM and MTM. In particular, we show that the maximum approximate ESJD of SMTM is always no smaller than that of MTM in \cref{Thm:maxesjd}. Moreover, they can also be computed numerically, allowing us to study their properties through simulations. Later in \cref{Sec:numerical}, we present a numerical study to illustrate the behavior of the ESJD under specific scenarios.

\begin{theorem}
\label{Thm:maxesjd}
    Under the same assumptions as in \cref{Thm:accept rate} and denote the mean value of $f$ as $\mathbb{E}_f[X]=m$, then we have
    \begin{equation}
        \max_\ell\widetilde{\ESJD}=\frac{1}{1-\alpha\cdot\beta\cdot\gamma}\cdot\left[\max_\ell\lim_{d\to\infty}\ESJD_{\MTM}\right],\quad\alpha,\beta,\gamma\in[0,1],
    \end{equation}
    where $\alpha=\frac{4\lambda}{(1+\lambda)^2}, \beta=1-m^2, \gamma=\frac{1}{(1-m^2)\mathbb{E}\left[((\log f)')^2\right]}$.
\end{theorem}
\begin{proof}
    See \cref{proof:thm:maxesjd}.
\end{proof}

\cref{Thm:maxesjd} highlights the theoretical robustness and efficiency gains of SMTM over MTM. Specifically, we establish that $\max_\ell\widetilde{\ESJD}\ge \max_\ell\lim_{d\to\infty}\ESJD_{\MTM}$, meaning SMTM dominates MTM in terms of ESJD. Crucially, the performance gap between the two algorithms is governed by three underlying geometric and distributional factors: $\alpha$, $\beta$, and $\gamma$. Each factor takes values in $[0, 1]$ and acts as a penalty for a specific type of misspecification. The equality $\max_\ell\widetilde{\ESJD} = \max_\ell\lim_{d\to\infty}\ESJD_{\MTM}$ holds if and only if the product $\alpha \cdot \beta \cdot \gamma = 0$. We can interpret these three factors as follows:

    \begin{itemize}
        \item \emph{Radius misspecification ($\alpha$):} Recalling that $R=\sqrt{\lambda d}$, the term $\alpha=\frac{4\lambda}{(1+\lambda)^2}$ measures the penalty for sub-optimal radius parameter of SMTM. It attains its maximum $\alpha=1$ when $\lambda=1$ (meaning the radius parameter is perfectly scaled). Conversely, as $\lambda \to 0$ or $\lambda \to \infty$, we have $\alpha \to 0$, meaning that the benefit of the stereographic construction vanishes and SMTM naturally degenerates to the Euclidean MTM.
        
        \item \emph{Location misspecification ($\beta$):} Recall that $\mathbb{E}_f[X]=m$, so the term $\beta=1-m^2$ measures the penalty for misaligning the sphere's center at $(0,\dots,0)^T\in\mathbb{R}^d$ relative to the mean of the target distribution at $(m,\dots,m)^T\in\mathbb{R}^d$. Since we assume $\mathbb{E}_f\left[X^2\right]=1$, $\beta$ represents the variance of the marginal target $f$. This penalty is minimized (i.e., $\beta$ is maximized at $1$) when $m=0$, which occurs when the center of the stereographic projection exactly matches the mean of the target distribution.
        
        \item \emph{Distribution misspecification ($\gamma$):} The term $\gamma=\frac{1}{(1-m^2)\mathbb{E}\left[((\log f)')^2\right]}$ measures the penalty associated with the target distribution's departure from isotropy (in the product i.i.d.\,case, Gaussianity). By the Cauchy-Schwarz inequality and integration by parts, we can bound $1/\gamma$ from below:
        \begin{equation*}
        \frac{1}{\gamma}=\left[\int_{\mathbb{R}} \left(\frac{f'(x)}{f(x)}\right)^2f(x)\dee x\right]\cdot\left[ \int_{\mathbb{R}} (x-m)^2 f(x)\dee x\right]\geq \left[\int_{\mathbb{R}} (x-m)f'(x)\dee x\right]^2=\left[\int_{\mathbb{R}} f(x)\dee x\right]^2=1.
        \end{equation*}
        This proves that $\gamma \leq 1$. The Cauchy--Schwarz inequality is tight if and only if $(\log f)' \propto (x-m)$. Because $f$ has full support on $\mathbb{R}$ with $\mathbb{E}_f\left[X^2\right]=1$, $(\log f)' \propto (x-m)$ holds if and only if $f$ is Gaussian with mean $m$ and variance $1-m^2$. 
    \end{itemize}

This factorization also clarifies the role of the practical tuning procedure in \cref{Sec:tuning}. In the product-target setting of \cref{Thm:maxesjd}, updating the location parameter $\mu$ toward the empirical mean centers the transformed target, moving $m$ toward zero and hence $\beta=1-m^2$ toward one. Updating the preconditioning matrix $\Sigma$ adjusts both scale and orientation. In the isotropic setting considered here, this scale adaptation corresponds to moving $\lambda$ toward one, at which $\alpha=4\lambda/(1+\lambda)^2$ is maximized. Thus, adaptation can reduce the avoidable location and scale misspecification and thereby increase $\alpha\beta\gamma$, leading to a larger predicted ESJD advantage of SMTM over MTM. By contrast, $\gamma$ is determined by the standardized shape of the marginal target and is not generally altered by location or scale adaptation. Finally, tuning $h$ does not change $\alpha$, $\beta$, or $\gamma$, but helps the finite-dimensional sampler approach the maximization over $\ell$ appearing in \cref{Thm:maxesjd}.

\section{Large $N$ Regime}
\label{Sec:n to infty}

Exploring the behavior of SMTM in the regime of large $N$ is of both theoretical and practical interest. For example, \cite{Gagnon2023} compare the performance of MTM under different weight functions and address its pathological behavior in the large $N$ regime. In this section, we perform a similar study. First, when $N\to\infty$, the limiting acceptance rates of SMTM for both globally-balanced and locally-balanced cases are established in \cref{prop:accept rate converge}. We find that, for large enough $N$, using the same step size, the acceptance rate of SMTM in the locally-balanced case is strictly higher than that in the globally-balanced case. Moreover, in \cref{sec:pathological} we prove that SMTM avoids some of the pitfalls of Euclidean MTM, such as pathological behavior in the transient phase discovered by \citet{Gagnon2023}. We then apply the comparison result of \citet{pozza2024fundamental} to derive an upper bound on the spectral gap of SMTM for heavy-tailed targets.

\subsection{Limiting Acceptance Rates}
The acceptance of a proposal in multiple-try MCMC can be viewed as choosing one proposal from all candidates and accepting it. The acceptance rate of SMTM can be computed as $\sum_{j=1}^N\alpha_2^j$. By Theorem \ref{Thm:accept rate}, this acceptance rate can be approximated by \begin{equation}
    \sum_{j=1}^N\mathbb{E}\left[\phi_2^j((W_i)_{i=1}^N,(V_i)_{i=1}^{N-1})\middle| X\right].
\end{equation}
It is quite challenging to explicitly express $\mathbb{E}[\phi_2^j((W_i)_{i=1}^N,(V_i)_{i=1}^{N-1})\mid X]$. However, rather than attempting to derive an exact formula, we can focus on understanding its asymptotic behavior as $N$ approaches infinity. Interestingly, in the globally-balanced case, the optimal acceptance rate of SMTM, when the ESJD is maximized, ultimately converges to the same limiting value as that of SRWM, which is approximately $0.234$ \cite[Corollary 5.2]{yang2024}. In contrast, in the locally-balanced case, the acceptance rate converges to $1$.

\begin{proposition}
\label{prop:accept rate converge}
    Assume $h$ is chosen as in \cref{Eq:reparametrization} with $\ell$ fixed. Then, as $N\to\infty$, we have in the globally-balanced case, 
    \begin{equation}
        \sum_{j=1}^{N}\mathbb{E}\left[\phi_2^j((W_i)_{i=1}^N,(V_i)_{i=1}^{N-1})\middle| X\right]\to\mathbb{E}\left[1\wedge\exp(W_1)\middle| X\right];
    \end{equation}
    in the locally-balanced case
    \begin{equation}
    \label{eq:accept rate large N}
        \sum_{j=1}^{N}\mathbb{E}\left[\phi_2^j((W_i)_{i=1}^N,(V_i)_{i=1}^{N-1})\middle| X\right]\to 1, 
    \end{equation}
    where $W_i$ and $V_i$ are the same as in \cref{Thm:accept rate}.
\end{proposition}
\begin{proof}
    See \cref{proof:prop:accept rate converge}.
\end{proof}

In both globally and locally-balanced cases, we use the re-parametrization of $h$ in \cref{Eq:reparametrization}. When $\ell$ is fixed as a constant, then by \citet[Eq.(15)]{yang2024}, $h\approx\frac{\ell}{d}\sqrt{4\lambda/(1+\lambda)^2}$, implying that $h$ is scaled as $O(d^{-1})$. \cref{eq:accept rate large N} suggests that, in the locally-balanced case, $h$ should follow a scaling different from $O(d^{-1})$ to obtain a non-degenerate limit for the acceptance rate. We conjecture that the appropriate scaling is $O(d^{-1/3})$, which is the same as the scaling order of common gradient-based MCMC \citep{roberts98}. However, establishing this rigorously is challenging. It requires a higher-order expansion, which is beyond the Gaussian approximation used in \cref{Thm:accept rate}. We leave a full theoretical analysis for future work.

\subsection{Pathological Performance and Spectral Gap}
\label{sec:pathological}
It is shown in \citet[Proposition 3]{Gagnon2023} that when the weight function is chosen as globally-balanced, MTM exhibits pathological behavior during its convergence phase. Specifically, when the initial state is in the tail, the expected acceptance probability approaches 0 in the large $N$ regime, causing the chain to get stuck. While for the locally-balanced weight function, \cite{Gagnon2023} present some numerical results suggesting that as the initial state moves further into the tail, the acceptance rate does not converge to 0. A rigorous lower bound of the acceptance rate has recently been established for Gaussian targets by \cite{caprio2025}.

In this section, we show that the acceptance rate of SMTM will always have a positive lower bound. Therefore, SMTM will not exhibit the pathological performance mentioned in \cite{Gagnon2023} in either globally or locally balanced cases. In the following, we consider two cases of tails: $\pi(x)\propto\|x\|^{-\alpha},\alpha>2d$ and $\pi(x)\propto\exp(-\|x\|^\beta),\beta>0$. The pathological behavior established in \citet[Proposition 3]{Gagnon2023} arises in the setting where the target distribution is standard normal (i.e., $\beta=2$).

Firstly, we study the weak convergence of the proposal distribution in SMTM. According to \cite{liu2000mtm, Gagnon2023}, the proposal distribution of MTM weakly converges to a limiting distribution when $N\to\infty$. We show that the weak convergence still holds for SMTM. To study the behavior of SMTM as $N\to\infty$, we define the ideal scheme (\cref{alg:ideal}) as a random-walk-type algorithm similar to SRWM but using a limiting proposal distribution $Q_\infty$ 
\begin{equation}
    Q_{\infty}(z,\hat{z}):=\frac{\omega(z,\hat{z})Q_S(z,\hat{z})}{\int_{\mathbb{S}^d}\omega(z,\hat{z})Q_S(z,\hat{z})\dee\hat{z}},
\end{equation}
where $Q_S(\cdot,\cdot)$ is the proposal distribution of the SRWM algorithm, assuming that the denominator exists and is finite. The proposal distribution $Q_\infty$ is the limiting distribution of the proposal selected by SMTM. Therefore, studying the ideal scheme allows us to analyze the limiting acceptance rate of SMTM.
\begin{algorithm}[]
\caption{Ideal Scheme}
\label{alg:ideal}
\begin{algorithmic}
    \State \textbf{Input: }Current state $X^d(t)=x$, target distribution $\pi$ on $\mathbb{R}^d$ (or $\pi_S$ on $\mathbb{S}^d$).
    \begin{itemize}
        \item Let $z=\SP^{-1}(x)$;
        \item Sample $\hat{z}$ with $Q_\infty(z,\hat{z})$;
        \item Define the proposal as $\hat{x}=\SP(\hat{z})$;
        \item $X^d(t+1)=\hat{x}$ with probability
        \begin{equation}\label{eq:alpha_inf}
            \alpha_\infty(z,\hat{z})=1\wedge\frac{\pi_S(\hat{z})Q_\infty(\hat{z},z)}{\pi_S(z)Q_\infty(z,\hat{z})};
        \end{equation}
        otherwise $X^d(t+1)=x$.
    \end{itemize}
\end{algorithmic}
\end{algorithm}

Write the Markov chain simulated by SMTM with fixed $N$ as $\{X_N(m):m\in\mathbb{N}\}$ and the Markov chain simulated by the ideal scheme in \cref{alg:ideal} as $\{X_{\infty}(m):m\in\mathbb{N}\}$. Following a similar proof as in \citet[Proof of Theorem 1]{Gagnon2023}, we establish in the following the weak convergence result for our ideal scheme.

\begin{theorem}
\label{thm:infinity weak convergence}
    Assume that $\mathbb{E}[\omega(Z,\hat{Z})^4]<\infty$ and $\mathbb{E}[\omega(Z,\hat{Z})^{-4}]<\infty$ where $Z\sim\pi_S$ and $\hat{Z}\mid Z\sim Q_S(Z,\cdot)$, then as $N\to\infty$, $\{X_N(m)\}$ converges weakly to $\{X_{\infty}(m)\}$ provided that $X_N(0)\sim\pi$ and $X_{\infty}(0)\sim\pi$.
\end{theorem}
\begin{proof}
    See \cref{proof:thm:infinity weak convergence}.
\end{proof}
Furthermore, we can show that our ideal scheme preserves uniform ergodicity, which is unlikely to hold for the Euclidean version in \citet[Theorem 1]{Gagnon2023}.
\begin{theorem}
\label{thm: ideal uniform ergodic}
If $\pi(x)$ is positive and continuous on $\mathbb{R}^d$ satisfying $\sup_{x\in\mathbb{R}^d}\pi(x)(R^2+\|x\|^2)^d<\infty$, then the ideal scheme in \cref{alg:ideal} is uniformly ergodic.
\end{theorem}
\begin{proof}
    See \cref{proof:thm: ideal uniform ergodic}.
\end{proof}

Now we prove that in the tail region, the acceptance rate $\alpha_\infty$ in \cref{eq:alpha_inf} of the ideal scheme has a positive lower bound in both globally and locally balanced cases under certain conditions on the target distribution and the step size.
\begin{proposition}
\label{prop:infinity lower bound}
    Suppose the target distribution $\pi$ satisfies either $\pi(x)\propto\|x\|^{-\alpha}$ or $\pi(x)\propto\exp(-\|x\|^\beta)$ where $\alpha>2d,\beta>0$ and the step size $h=O(d^{-1/2})$, then there exists a constant $\xi>0$ such that when $\frac{1}{\|x\|}=o(h)$, $z=\SP^{-1}(x)$ and $\hat{z}\sim Q_{\infty}(z,\cdot)$,
    \begin{itemize}
        \item in the globally-balanced case, we have
        \begin{equation}
            \alpha_{\infty}(z,\hat{z})\geq \xi,
        \end{equation}
        \item in the locally-balanced case, for any $\varepsilon>0$ there exists $\mathscr{D}_\varepsilon\in\mathbb{N}$ such that when $d>\mathscr{D}_\varepsilon$, with probability no smaller than $1-\varepsilon$, we have
        \begin{equation}
            \alpha_{\infty}(z,\hat{z})\geq \xi.
        \end{equation}        
    \end{itemize}
\end{proposition}
\begin{proof}
    See \cref{proof:prop:infinity lower bound}.
\end{proof}

Finally, we discuss some recent progress on the spectral gap of multiple-proposal algorithms. The spectral gap provides a measure of the convergence rate of a (reversible) Markov chain, with a larger gap indicating faster convergence. Intuitively, the improvement of MTM using $N$ candidates should be about $N$ times that of RWM. But for log-concave distributions, \cite{pozza2024fundamental} have proven that such improvement is at most $\log(N)$, which implies that if we want to achieve the computational efficiency of $N$ RWM algorithms, we may need the number of candidates in each step to be exponential in $N$. However, their result does not hold for heavy-tailed distributions. Therefore, we consider a similar spectral gap analysis on the role of $N$ for the improvement of SMTM compared to SRWM when the target distribution is heavy-tailed.

Recall that the spectral gap of a $\pi$-reversible Markov kernel $P$ is defined as
\begin{equation}
    \Gap(P)=\inf_{f\in L^2(\pi)}\frac{\int_{\mathbb{R}^d}\int_{\mathbb{R}^d}\left(f(y)-f(x)\right)^2\pi(\dee x)P(x,\dee y)}{2\Var_{\pi}(f)},
\end{equation}
where 
\begin{equation}
    L^2(\pi)=\left\{f:\mathbb{R}^d\to\mathbb{R}\mid\Var_{\pi}(f):=\int_{\mathbb{R}^d}f^2(x)\pi(\dee x)-\left(\int_{\mathbb{R}^d}f(x)\pi(\dee x)\right)^2<\infty\right\}.
\end{equation}

The following result is from \cite{pozza2024fundamental}.
\begin{proposition}
\label{prop:log gap}\cite[Proposition 2]{pozza2024fundamental}
    Let $P^{(N)}$ be the Markov kernel of MTM in \cref{alg:MTM} and $\tilde{P}$ be the kernel of RWM with the same proposal distribution $q$ as in \cref{alg:MTM} and denote $\pi$ as the target distribution. Let $X\sim\pi$, $Y_i\mid X\sim q(X,\cdot)$ and $Y_i$'s are conditionally independent. Assume that $\mathbb{E}[\exp(s(\nu^T(Y_i-X))^2)]\leq M^{\nu}(s)<\infty$ for some $s>0$ and all $i\in\{1,\dots,N\}$, then we have
    \begin{equation}
        \Gap(P^{(N)})\leq\min\left(2N\mathop{\essinf}\limits_{x \in \mathbb{R}^d}\tilde{P}(x,\mathbb{R}^d\backslash\{x\}),\inf_{\nu\in\mathbb{R}^d,\|\nu\|=1}\frac{1}{2s}\frac{\log(N)+\log(M^{\nu}(s))}{\Var(\nu^TX)}\right).
    \end{equation}
\end{proposition}
Notably, this proposition does not hold for heavy-tailed distributions, as such distributions do not satisfy the condition $\mathbb{E}[\exp(s(\nu^T(Y_i-X))^2)]\leq M^{\nu}(s)<\infty$ in Proposition \ref{prop:log gap}. Next, we show that for multivariate Student's $t$ distribution with degrees of freedom no smaller than $d$, following a similar proof technique, we can establish an upper bound with a factor of $N^{2/p}$ instead of $\log(N)$, where $p$ is a positive integer smaller than $d$.

\begin{proposition}
\label{prop:almost linear gap}
    Let $P^{(N)}$ be the Markov kernel of SMTM in Algorithm \ref{alg:SMTM} with proposal distribution $\widetilde{Q_S}(\cdot,\cdot)$ on $\mathbb{R}^d$ and $\tilde{P}$ be the kernel of SRWM in Algorithm \ref{alg:SPS}. Suppose the target $\pi$ is a Student's $t$ distribution with mean $0$ and degrees of freedom no smaller than $d$. Let $X\sim\pi$, $Y_i\mid X\sim\widetilde{Q_S}(X,\cdot)$ and $Y_i$'s are conditionally independent, then $\mathbb{E}[(\nu^T(Y_i-X))^p]\leq M^{\nu}(p)<\infty$ for any positive integer $p<d$ and all $i\in\{1,\dots,N\}$, and we also have
    \begin{equation}\label{eq:gap_poly}
        \Gap(P^{(N)})\leq\min\left(2N\mathop{\essinf}\limits_{x \in \mathbb{R}^d}\tilde{P}(x,\mathbb{R}^d\backslash\{x\}),\inf_{\nu\in\mathbb{R}^{d},\|\nu\|=1}\frac{N^{\frac{2}{p}}M^{\nu}(p)^{\frac{2}{p}}}{2\Var(\nu^TX)}\right).
    \end{equation}
\end{proposition}
\begin{proof}
    See \cref{proof:prop:almost linear gap}.
\end{proof}

In a serial computing environment, according to \cref{eq:gap_poly} (see also \citet[Corollary 1]{pozza2024fundamental}), the improvement of the spectral gap by using a multiple-try method with $N$ candidates is at-most-linear in $N$. Thus running SMTM with $N>1$ for $1$ iteration is less efficient than running SRWM for $N$ iterations. This limitation may be mitigated in a parallel implementation, where the $N$ candidates can be evaluated simultaneously (see \citet[Discussion after Corollary 1]{pozza2024fundamental}). However, \cref{prop:log gap} establishes a stronger limitation for log-concave distributions: the factor $\log(N)$ has negative implications for the number of parallel workers. By contrast, the factor $N^{2/p}$ in \cref{prop:almost linear gap} suggests that SMTM might be more suited to a parallel computing context. Compared to light-tailed (log-concave) settings, increasing the number of proposals $N$ in heavy-tailed settings could yield significant performance gains and the number of parallel workers can be polynomial in $N$. Nevertheless, $N^{2/p}$ in \cref{prop:almost linear gap} represents only an upper bound. To rigorously confirm the benefit of increasing $N$ for SMTM in heavy-tailed settings, a (matching) lower bound is necessary. Establishing such a bound remains an open problem and is left for future work.

\section{Numerical Experiments}\label{Sec:numerical}

In \cref{subsec_simulation}, we present numerical experiments illustrating the performance of SMTM relative to SRWM and MTM. In \cref{subsec_application}, we consider two realistic examples arising in Bayesian statistics. As noted in \cref{remark:SP_extend}, the stereographic projection admits a natural extension by introducing a location parameter $\mu$ and a covariance matrix $\Sigma$, enabling adaptive tuning. For simplicity, we fix $\Sigma = R^2 I_d$ and refer to $\mu$ and $R$ as the sphere parameters. These parameters are tuned adaptively following the framework of \cref{Sec:tuning}.

\subsection{Simulations}\label{subsec_simulation}

We begin by examining the convergence behavior of SMTM through comparisons with SRWM, MTM, and RWM under both light-tailed and heavy-tailed target distributions. We further show that, in the globally-balanced case, SMTM avoids the pathological behavior observed in MTM. We then focus on the robustness of SMTM to sphere parameters by varying sphere parameters $\mu$ and $R$ in both the burn-in phase and the stationary phase. 

Firstly, for the burn-in phase, we let the dimension of the state space be $d=100$ and the initial state be $(200,\dots,200)^T$. For the light-tailed target distribution, we consider the standard multivariate Gaussian. For the heavy-tailed distribution, we consider the product distribution: $\pi(x)=\prod_{i=1}^{100}f(x_i)$ where $f$ is the rescaled density function of one-dimensional standard Student's $t$ distribution with degrees of freedom $d+1$. Note that we rescale $f$ to have a unit second moment, which satisfies the assumption in \cref{eq_finite_moment}.

Unlike the multivariate Student’s $t$ distribution, the product of (rescaled) i.i.d. univariate Student’s $t$ distributions is not isotropic (rotation-invariant). Our optimal scaling results in \cref{Sec:optimal scaling} hold for these targets.
For MTM and SMTM, the number of candidates is set to $N=5$. For SRWM and SMTM, we fix the radius parameter at $R=\sqrt{d}=10$. In both examples, the acceptance rates of SRWM and SMTM remain above their respective optimal values, even for relatively large step sizes. To ensure a fair comparison, we therefore fix the step size at $h=0.01$ for both SMTM and SRWM, and compare their performance against the optimally tuned MTM and RWM. Specifically, the step sizes for RWM and MTM (with $N=5$) are calibrated to achieve acceptance rates of $23\%$ and $41\%$, respectively, which are (approximately) optimal according to \citet[Table~1]{bedard2012scaling}.
\cref{fig:heavy vs light} shows the burn-in phase of these algorithms. For both light-tailed and heavy-tailed target distributions, we observe that SRWM and SMTM converge significantly faster than the Euclidean algorithms. 
\begin{figure}[h]
    \centering
    \begin{minipage}[b]{0.48\linewidth}
        \centering
        \includegraphics[width=\linewidth]{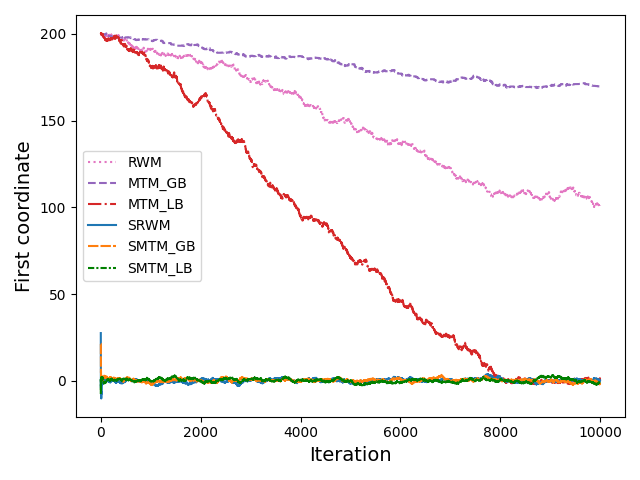}
        \subcaption{light-tailed distribution}
    \end{minipage}
    \hspace{-0.01\linewidth}
    \begin{minipage}[b]{0.48\linewidth}
        \centering
        \includegraphics[width=\linewidth]{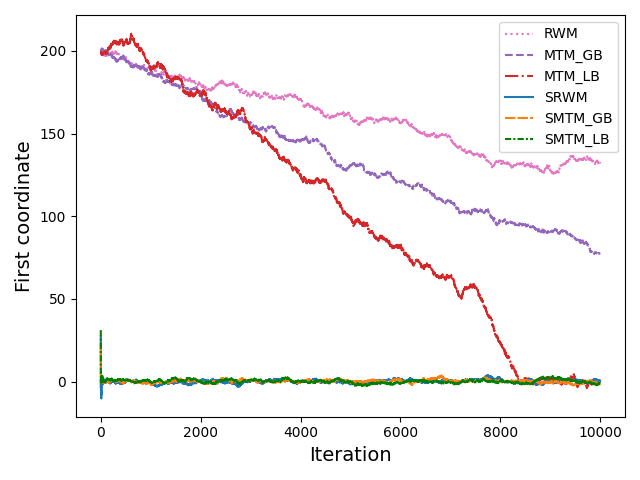}
        \subcaption{heavy-tailed distribution}
    \end{minipage}
    \caption{Burn-in Phase for different target distributions}
    \label{fig:heavy vs light}
\end{figure}

Next, we consider the large $N$ regime and demonstrate that GB-SMTM does not exhibit the pathological behavior of GB-MTM described in \cref{sec:pathological}. We set the dimension of state space to $d=100$ and choose the radius parameter $R=\sqrt{d}$. For GB-SMTM, the step size is fixed at $0.01$, whereas for GB-MTM, it is tuned to achieve the optimal acceptance rate. For the standard Gaussian target, we examine trace plots initialized at $(200,\dots,200)^T$. As shown in \cref{fig:pathological}, increasing the number of candidates $N$ significantly prolongs the burn-in phase of GB-MTM, making it more prone to becoming trapped in distant regions. In contrast, GB-SMTM does not suffer from this issue and consistently exhibits rapid convergence.
\begin{figure}[H]
    \centering
    \includegraphics[width=1\linewidth]{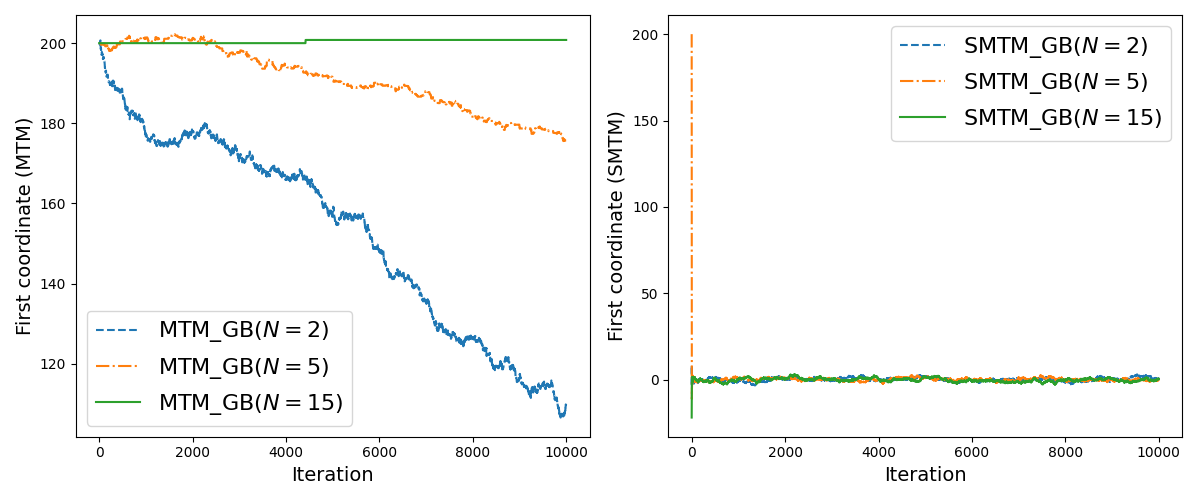}
    \caption{Burn-in Phase of GB-MTM and GB-SMTM}
    \label{fig:pathological}
\end{figure}

Next, we investigate the robustness of SMTM and SRWM to parameter choices. For simplicity, we restrict attention to the locally-balanced weight function for both MTM and SMTM. By default, the sphere is centered at the origin, whereas the target distribution may have a nonzero mean. This mismatch is equivalent to shifting the sphere away from the origin while keeping the target centered at zero. We therefore refer to any setting in which the target distribution has a nonzero mean as a \emph{mis-located sphere}. Intuitively, the larger the deviation of the target mean from the origin, the more severe the mis-location. We set the dimension to $d=50$ and consider the target distribution $\pi(x)=\prod_{i=1}^{50} f(x_i)$, where $f$ is a rescaled and shifted one-dimensional Student's $t$ density with mean $m=100$, variance $s^2=1$, and degrees of freedom $d+1$. The sphere parameter is chosen as $R=\sqrt{(s^2+m^2)d}$ (see \citet[Section~5.3]{yang2024}), and the chain is initialized at $(0,\dots,0)^T$. The step sizes of both SRWM and SMTM are tuned to achieve their respective optimal acceptance rates.
\cref{fig:convergence} shows that, in this setting, SMTM converges faster than SRWM, with its convergence further improving as the number of candidates increases. This suggests that using multiple candidates can substantially accelerate burn-in, particularly when the sphere is misaligned with the target distribution; moreover, increasing the number of candidates leads to progressively faster convergence.

\begin{figure}[]
    \centering
    \includegraphics[width=0.6\linewidth]{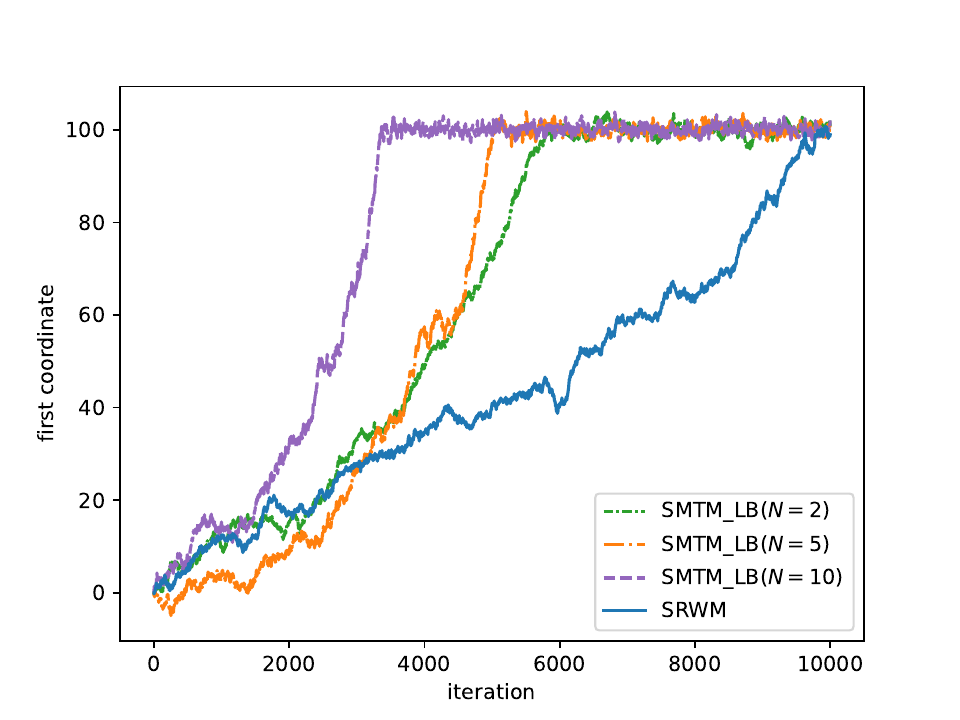}
    \caption{Burn-in Phase of SRWM and LB-SMTM for different $N$}
    \label{fig:convergence}
\end{figure}

Finally, we focus on the limiting ESJD in the stationary phase. Recall that under certain assumptions on the target distribution, we have proven in \cref{Sec:optimal scaling} that the (total) acceptance rate of SMTM $\sum_{j=1}^{N}\alpha_2^j$ converges to $\sum_{j=1}^{N}\phi_2^j((W_i)_{i=1}^{N},(V_i)_{i=1}^{N-1})$ and the ESJD of SMTM can be approximated by 
\begin{equation}
    N\ell^2\mathbb{E}_{X\sim\pi}\left[\mathbb{E}\left[\phi_2^j((W_i)_{i=1}^{N},(V_i)_{i=1}^{N-1})\middle| X\right]\right],
\end{equation} 
where $\phi_1^j, \phi_2^j$ are defined in \cref{eq:phi 1,eq:phi 2} and $(W_i)_{i=1}^N,(V_i)_{i=1}^N$ are defined in \cref{Eq:normal distribution}. If we take $f(x)$ as the probability density function of $\mathcal{N}(\mu,1-\mu^2)$ where $\mu\in(0,1)$, the assumptions on $f(x)$ stated at the beginning of \cref{Sec:optimal scaling} are satisfied:
\begin{equation}
    \label{eq:normal simulation}
    \mathbb{E}_f\left[X^2\right]=1,\quad\mathbb{E}_f\left[\left(\log f'\right)^2\right]=\frac{1}{1-\mu^2}.
\end{equation}
Recall that our target distribution is not isotropic (rotation-invariant) if $\mu\neq 0$, and $\mu$ indicates the level of mis-location of the sphere. 
\begin{figure}[]
    \centering
    \includegraphics[width=0.8\linewidth]{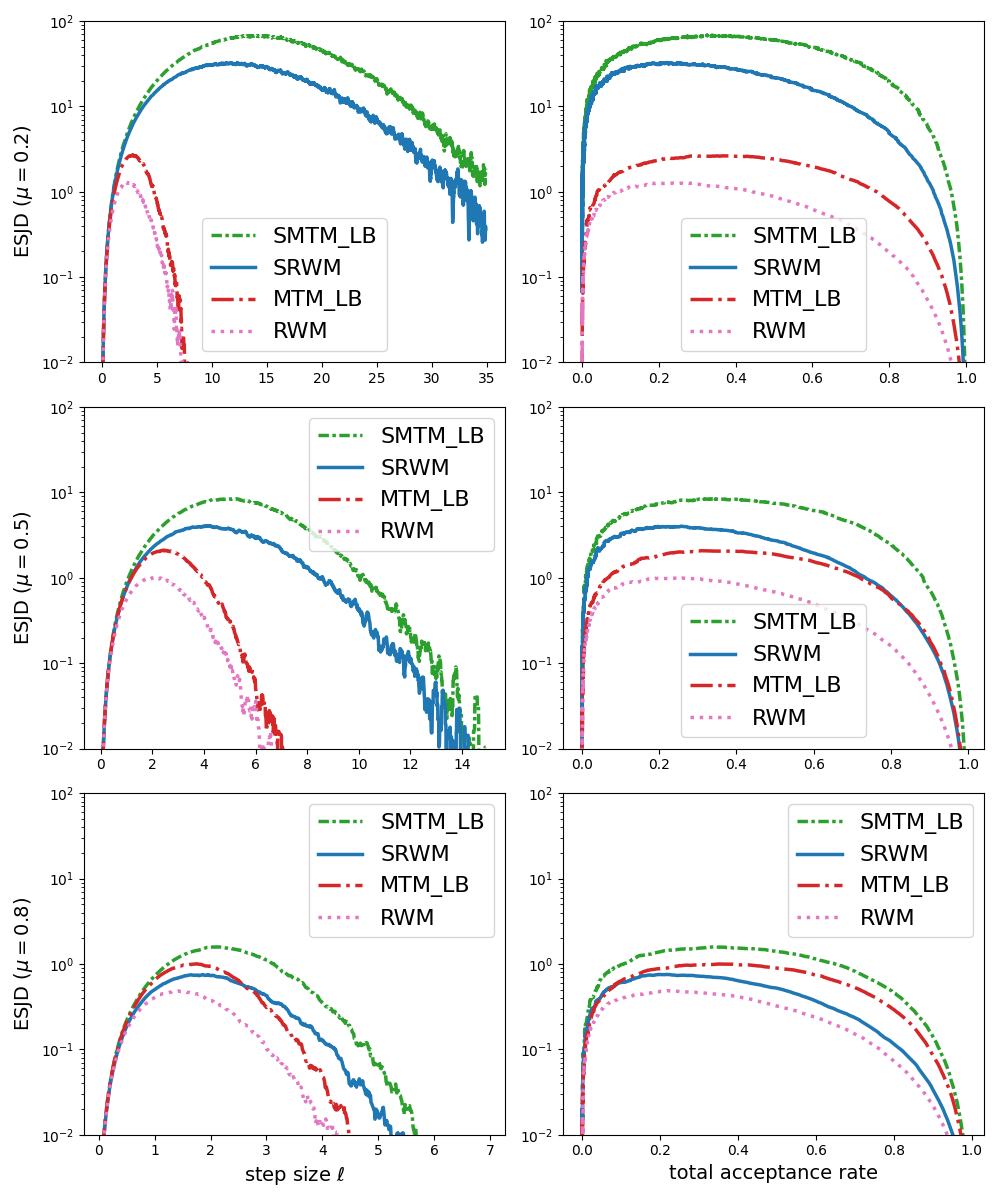}
    \caption{Simulation results with different locations when $N=3$}
    \label{fig:robustness center}
\end{figure}

\cref{fig:robustness center,fig:robustness radius} illustrate the performance of LB-SMTM, LB-MTM, SRWM, and RWM under varying sphere parameters $\mu$ and $R$. In \cref{fig:robustness center}, the radius is fixed at $R=\sqrt{d}$, while in \cref{fig:robustness radius}, the location parameter is fixed at $\mu=0.5$. As shown in \cref{fig:robustness center}, SMTM significantly outperforms the other three algorithms and exhibits greater robustness to the choice of sphere parameters. In particular, the performance of SRWM deteriorates as the sphere becomes increasingly mis-located and is eventually surpassed by MTM when $\mu=0.8$. To assess robustness with respect to the radius parameter, we fix $\mu=0.5$ and vary $R$. Writing $R=\sqrt{\lambda d}$ as in \cref{Thm:accept rate,Thm:esjd}, we consider $\lambda=0.1$, $1$, and $10$ in \cref{fig:robustness radius}. According to \cref{def_invSP}, when $\lambda=0.1$, most of the mass of the target distribution is projected near $z_{d+1}=9/11 \approx 0.82$, indicating a severe misspecification of the radius (similarly for $\lambda=10$). In contrast, $\lambda=1$ corresponds to a well-specified radius, with most of the mass projected near the equator. The conclusions from \cref{fig:robustness radius} are consistent with those from \cref{fig:robustness center}. Notably, while SRWM is dominated by MTM when the radius is either too large or too small, SMTM maintains strong performance across all settings, highlighting its robustness to the choice of radius.

\begin{figure}[]
    \centering
    \includegraphics[width=0.8\linewidth]{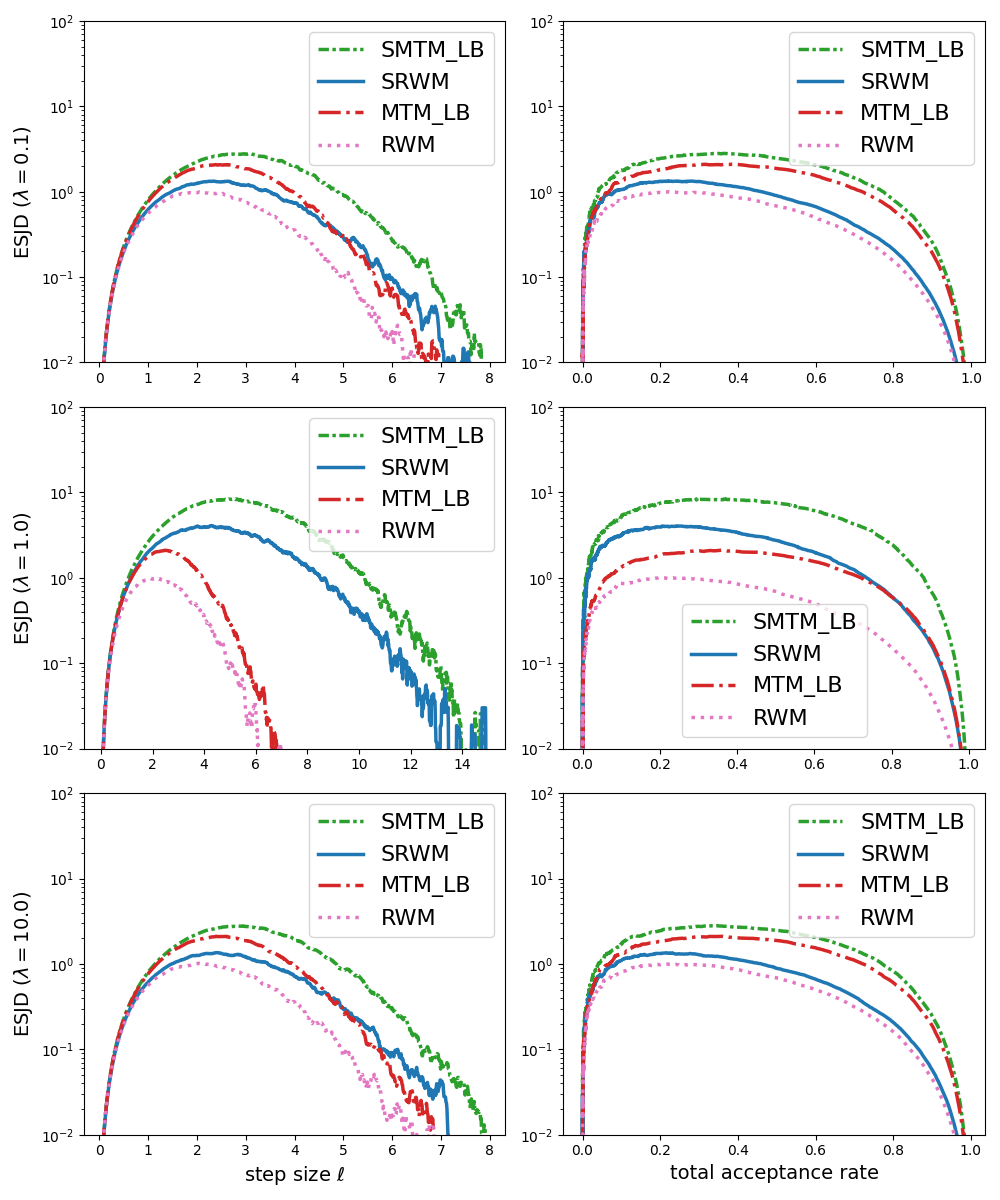}
    \caption{Simulation results with different radii when $N=3$}
    \label{fig:robustness radius}
\end{figure}

\subsection{Real Applications}\label{subsec_application}

\cref{sec:adaptive_tuning} records the implementation details of the practical procedure in \cref{Sec:tuning}. We then consider two realistic examples from Bayesian statistics in \cref{sec:bayesian_lasso,sec:bayesian_t}: (i) a light-tailed posterior arising from Bayesian Lasso, and (ii) a heavy-tailed posterior from Bayesian Student's $t$ regression. Our primary objective is to provide a clear comparison between SMTM, MTM, and SRWM. Accordingly, we do not pursue an exhaustive empirical evaluation against all competing methods, such as Gibbs samplers \citep{park2008bayesian,geweke1993}, which are typically tailored to specific model structures and prior choices.

\subsubsection{Adaptive Tuning}
\label{sec:adaptive_tuning}
In this section, we use the tuning framework introduced in \cref{Sec:tuning}. Specifically, we set the time points as $t_k:=1000\sum_{j=1}^kj^{0.8}$ for $k\in\mathbb{N}$ so that $t_k-t_{k-1}$ increases with $k$ and the adaptations occur less frequently. We fix $\Sigma=R^2I_d$ for simplicity, meaning that we adaptively tune the location $\mu$ and the radius $R$ (as well as $h$) under the AIR framework. During the adaptation steps, we use the empirical sample mean $\bar{X}$ to update $\mu$. For the radius parameter $R$, we update using $\bar{R}=\sqrt{\frac{1}{n-1}\sum\|X_i-\bar{X}\|^2}$, which is the isotropic version of $\Sigma=d\hat{C}$. We update $h$ toward the optimal acceptance rate in \cref{table:optimal_parameters} and the number of candidates $N$ is fixed within each run.

\subsubsection{Example 1: Bayesian Lasso}
\label{sec:bayesian_lasso}
Bayesian Lasso regression \citep{park2008bayesian} is a Bayesian analogue of the classical Lasso estimator. Given a linear model
\begin{equation}
y=X\beta+\varepsilon,\quad\varepsilon\sim\mathcal{N}(0,\sigma^2I_n),
\end{equation}
where $X\in\mathbb{R}^{n\times p}, y\in\mathbb{R}^{n\times 1}$. Bayesian Lasso replaces the $L_1$ penalty in classical Lasso with a Laplace prior so that
    $p(\beta\mid\sigma,\lambda)=\left(\frac{\lambda}{2\sigma}\right)^{p}\exp(-\frac{\lambda}{\sigma}\|\beta\|_1)$. 
    Let the prior of $\sigma^2$ be
    $p(\sigma^2)\propto\frac{1}{\sigma^2}$. The posterior distribution can be written as
\begin{equation}
    p(\beta,\sigma^2\mid y)\propto\sigma^{-(n+p+2)}\exp\left(-\frac{1}{2\sigma^2}\|y-X\beta\|_2^2-\frac{\lambda}{\sigma}\|\beta\|_1\right).
\end{equation}
Bayesian Lasso and its variants have been widely applied in high-dimensional regression problems such as genomics \citep{li2010genome} and signal recovery \citep{green2025complex}.

In this example, we demonstrate the performance of the proposed SMTM compared with SRWM and MTM on the Bayesian Lasso posterior. We adopt the standardized diabetes dataset used in \cite{efron2004least}. In this dataset, $X$ has size $442\times10$ containing the measurements of the explanatory variables (age, sex, body mass, etc.), and $y$ has size $442\times1$ corresponding to the response. Including $\sigma^2$, the dimension of the posterior is $d=11$. We first compare the convergence behavior during the burn-in phase using the adaptive tuning in \cref{sec:adaptive_tuning}. Under the setting that $\lambda=1$ and initial $\beta(0)=(100,\dots,100)^T$, the algorithms are tuned to achieve acceptance rates of approximately $24\%$ for RWM and SRWM, and $32\%$ for MTM and SMTM, which are nearly optimal according to \cref{table:optimal_parameters}. For the multiple-try algorithms, we use the locally-balanced weight function and $N=2$. For the stereographic algorithms, we let $\mu=(0,\dots,0)^T$ and  $R=\sqrt{d}=\sqrt{11}$ as the initial values. The traceplots of $\beta_3$ are shown in \cref{fig:bayeisanlasso_tuned}. We observe that the stereographic algorithms, SRWM and SMTM, converge much faster than RWM and MTM, with SMTM showing the fastest convergence of all the methods. 
\begin{figure}[h]
    \centering
    \includegraphics[width=0.6\linewidth]{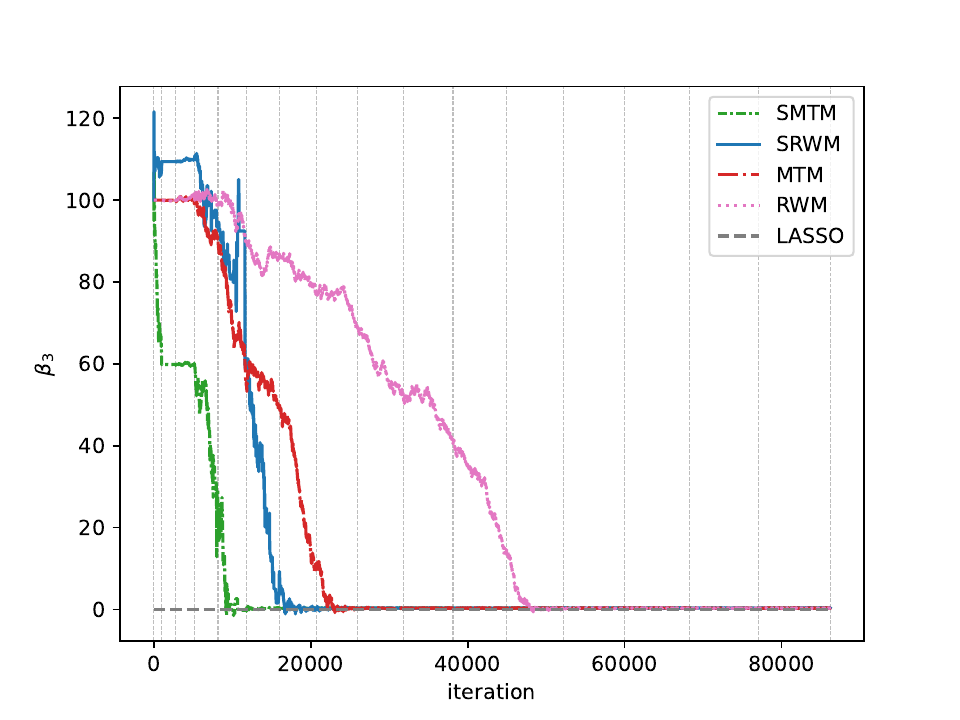}
    \caption{Burn-in Phases for Bayesian Lasso with Adaptive Tuning}
    \label{fig:bayeisanlasso_tuned}
\end{figure}

To compare performance in the stationary phase, we freeze the tuned parameters after adaptation and run each algorithm for an additional 10,000 iterations. For the coefficient $\beta_3$, we report the lag-1 autocorrelation (ACF), integrated autocorrelation time (IACT), effective sample size (ESS), and ESJD in \cref{tab:beta3_stationary}. The autocorrelation functions of $\beta_3$ over a range of lags are shown in \cref{fig:autocorrelation}. Overall, SMTM exhibits the best stationary performance, achieving the lowest lag-1 ACF and IACT together with the highest ESS and ESJD.

\begin{table}
    \centering
    \begin{tabular}{lllll}
        & Lag-1 ACF & IACT & ESS & ESJD \\ \hline
        SMTM & 0.9312 & 29.52 & 338.75 & $2.15 \times 10^{-4}$ \\
        SRWM & 0.9516 & 55.31 & 180.80 & $1.31 \times 10^{-4}$ \\
        MTM  & 0.9355 & 32.66 & 306.18 & $2.07 \times 10^{-4}$ \\
        RWM  & 0.9565 & 43.85 & 228.05 & $1.33 \times 10^{-4}$ \\ \hline
    \end{tabular}
    \caption{Stationary-phase sampling performance for the coefficient $\beta_3$, calculated from 10,000 iterations.}
    \label{tab:beta3_stationary}
\end{table}

\begin{figure}
    \centering
    \includegraphics[width=0.6\linewidth]{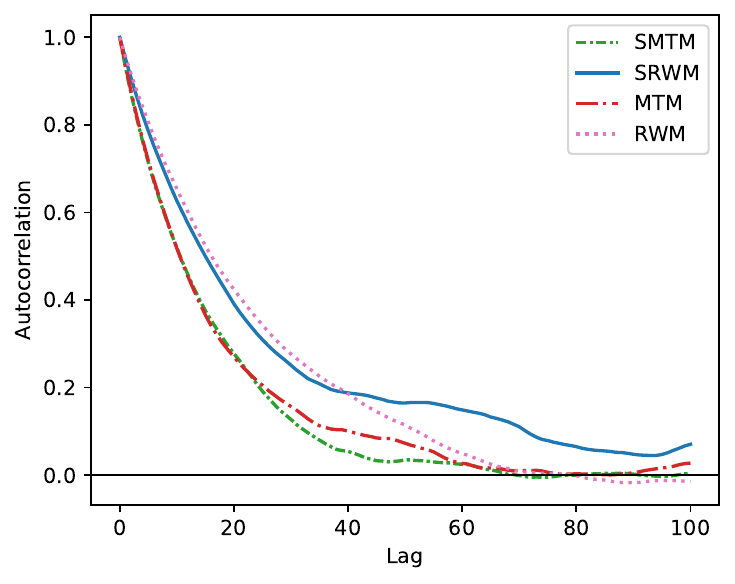}
    \caption{Autocorrelation functions of the sampled coefficient $\beta_3$ across different lags}
    \label{fig:autocorrelation}
\end{figure}

To assess the robustness of SMTM to the sphere parameters, we fix these parameters in the next example and only tune the step size $h$. We plot the burn-in phases using the same initial setting as before, while increasing the number of candidates to $N=10$. Here, $N=10$ is used as a representative larger number of candidates and is not the result of problem-specific optimization. As shown in \cref{fig:convergence}, the same qualitative behavior is observed for different values of $N$. The step size $h$ is tuned to achieve an acceptance rate of $30\%$ for all the algorithms. The traceplots of $\beta_3$ are shown in \cref{fig:bayeisanlasso}. In this setting, SMTM and MTM outperform SRWM and RWM. The slow convergence of SRWM is mainly due to an inappropriate choice of the sphere parameters. Notably, even under such a misspecified setting of sphere parameters, SMTM continues to perform well. This indicates that SMTM is more robust to poor choices of the sphere parameters than SRWM.
\begin{figure}[h]
    \centering
    \includegraphics[width=0.6\linewidth]{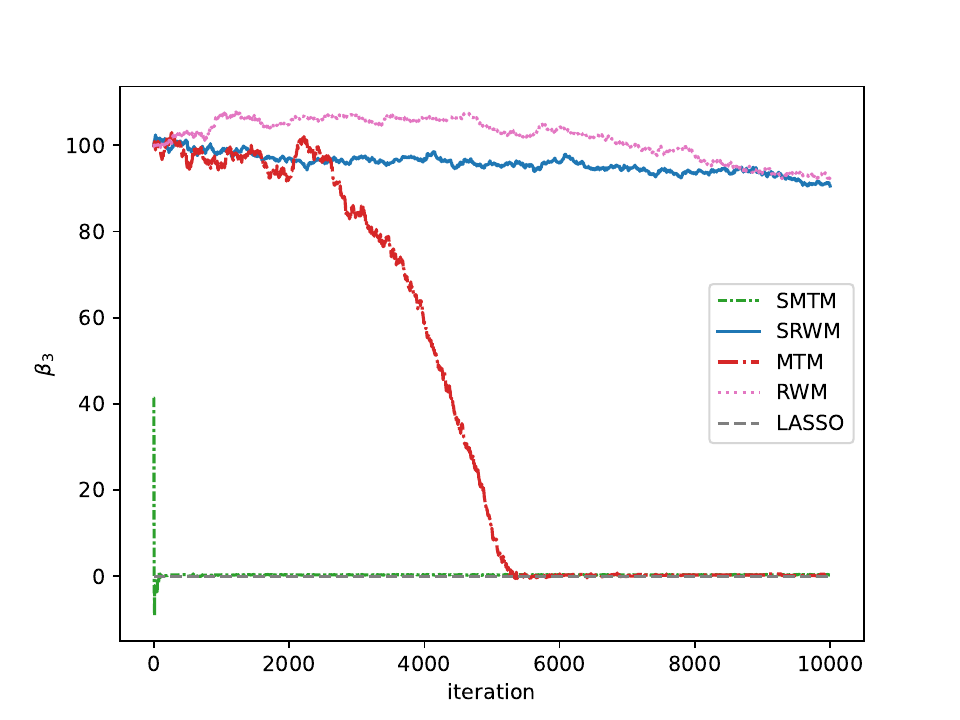}
    \caption{Burn-in Phases for Bayesian Lasso without tuning sphere parameters}
    \label{fig:bayeisanlasso}
\end{figure}

\subsubsection{Example 2: Bayesian Student's $t$ Regression}
\label{sec:bayesian_t}
Given the same linear equation $y= X\beta+\epsilon$, instead of assuming $\epsilon\sim \mathcal{N}(0,\sigma^2I_n)$, Bayesian Student's $t$ regression assumes $\{\epsilon_i\}$ to be independent with
\begin{equation}
    \epsilon_i\sim t_{\nu}(0,\sigma^2),
\end{equation}
where $t_{\nu}$ denotes the standard Student's $t$ distribution with degrees of freedom $\nu$. The Bayesian Student's $t$ regression model is well known for its robustness to outliers and has been widely applied in areas such as mobility trend analysis \citep{boonstra2021multilevel} and traffic crash modeling \citep{li2023mitigating}. Let the priors of $\sigma^2$ and $\beta$ be $p(\sigma^2)\propto\frac{1}{\sigma^2}$ and $p(\beta)\propto 1$. The posterior can be written as 
\begin{equation}
    p(\beta,\sigma^2\mid y)\propto\sigma^{-n-2}\prod_{i=1}^n\left(1+\frac{(y_i-X_i^T\beta)^2}{\nu\sigma^2}\right)^{-\frac{\nu+1}{2}}.
\end{equation}
Note that, for any fixed $\nu$, the posterior tail of $\beta$ is heavier than exponential. Consequently, RWM and MTM are not geometrically ergodic for this target. We therefore expect the stereographic MCMC algorithms, SRWM and SMTM, to outperform RWM and MTM in this example.

We use the standardized Boston Housing dataset in \cite{kuss2006}, which contains $506$ observations and $13$ input variables for predicting median house prices in the Boston metropolitan area. We first compare the burn-in phases with the adaptive tuning of sphere parameters in the AIR framework. For RWM and SRWM, the acceptance rates are tuned to approximately $24\%$. For MTM and SMTM (with the locally-balanced weight function and $N=2$), the acceptance rates are tuned to be approximately $32\%$. For the stereographic algorithms, the initial sphere parameters are $\mu=(0,\dots,0)^T$ and $R=\sqrt{d}=\sqrt{14}$. Let the degrees of freedom $\nu=4$ and the initial state at $(100,\dots,100)^T$. Under the reparameterization $u=\log\sigma$, the posterior has exponential tails in $u$. To characterize its tails in $\beta$, write $\beta=\|\beta\|v$ for a unit vector $v$ and fix the standardized matrix $X$. Let $k=\min\|Xv\|_0$ for any unit vector $v$, where $\|\cdot\|_0$ denotes the number of nonzero entries. Along any direction $v$, the posterior decays as $O(\|\beta\|^{-k(\nu+1)})$. A lower bound for $k$ can be obtained by partitioning $X$ into $t$ disjoint, full-rank row blocks: $X=(A_1,\dots, A_t)^T$. Since $A_i$ has full rank for any $i=1,\dots,t$, $A_iv\neq 0$ for any $v\neq 0$, which implies that there is at least one non-zero entry in $A_iv$ and thus $k\geq t$. In this case, $X$ can be partitioned into 26 such blocks. Consequently, for $\nu=4$, the posterior decays at least as $O(\|\beta\|^{-26(\nu+1)})=O(\|\beta\|^{-130})$ as $\|\beta\|\to\infty$ and therefore satisfies the tail condition in \cref{Thm:uniform ergodic}. We draw traceplots of $\beta_1$ in the burn-in phase for RWM, MTM, SRWM, and SMTM, respectively. According to \cref{fig:bayesianstudent_tuned}, SRWM performs well after sphere parameters are adaptively updated, while SMTM performs the best among all.
\begin{figure}
    \centering
    \includegraphics[width=0.6\linewidth]{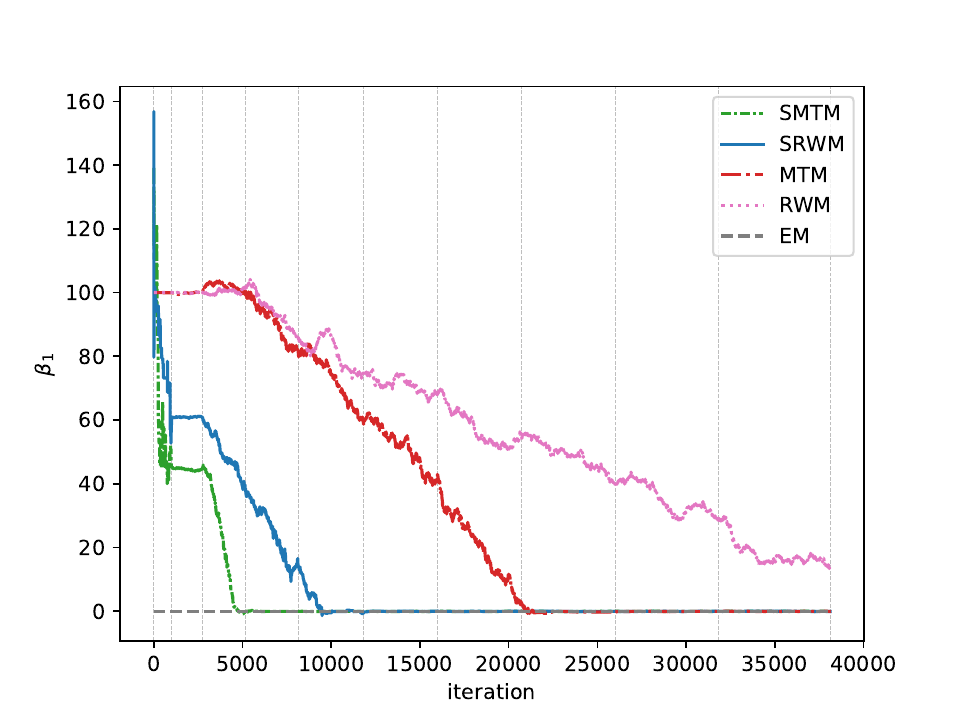}
    \caption{Burn-in Phases for Bayesian Student’s $t$ Regression with Adaptive Tuning}
    \label{fig:bayesianstudent_tuned}
\end{figure}

Next, we study how much the performance of SMTM and MTM can be improved by choosing a larger $N$. We increase the number of candidates to $N=10$ and fix the sphere parameters as their initial values, i.e., $\mu=(0,\dots,0)^T$ and $R=\sqrt{d}=\sqrt{14}$. The step sizes of the algorithms are tuned where the acceptance rates are around $30\%$. The traceplots of $\beta_1$ are shown in \cref{fig:bayesianstudent}. We observe that, using $10$ candidates in MTM still does not outperform SRWM. In contrast, SMTM shows a significant improvement with $N=10$ compared with $N=2$ and converges much faster than all the other algorithms.
\begin{figure}[H]
    \centering
    \includegraphics[width=0.6\linewidth]{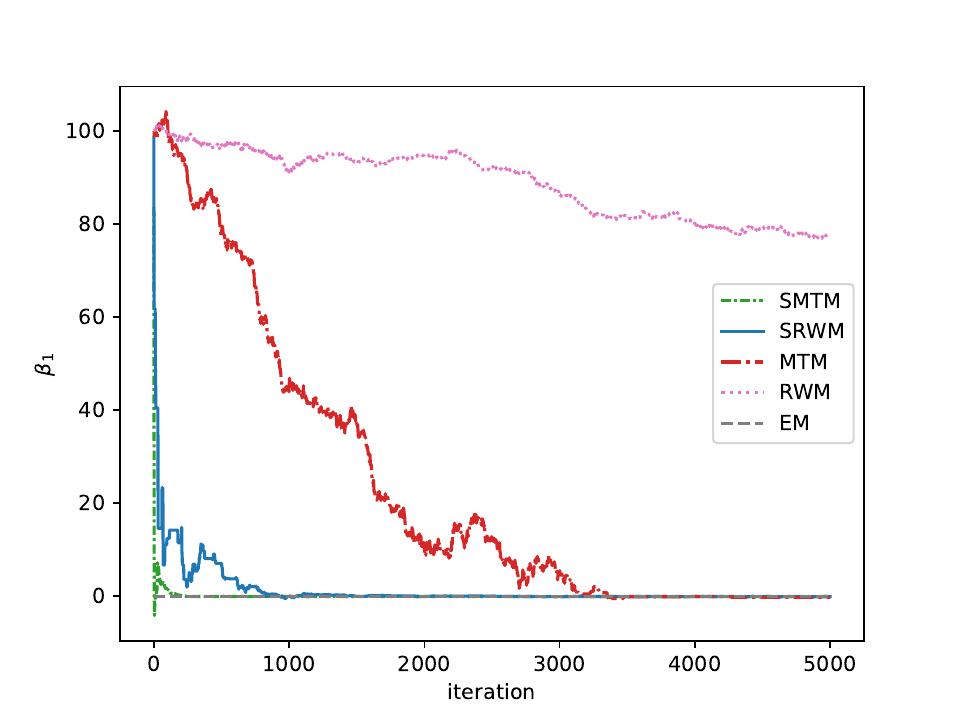}
    \caption{Burn-in Phases for Bayesian Student's $t$ Regression}
    \label{fig:bayesianstudent}
\end{figure}

\section{Conclusion}
We introduced SMTM, a gradient-free sampler that combines multiple-try proposals with stereographic MCMC, together with a practical warm-up tuning procedure. We established its favorable ergodic properties for both light- and heavy-tailed targets, derived high-dimensional results for acceptance rate and ESJD, and showed that it avoids the pathological large-$N$ behavior of Euclidean MTM. The numerical experiments demonstrate faster convergence, improved stationary performance, and robustness to misspecified sphere parameters. Developing matching spectral-gap bounds and establishing the large-$N$ scaling of locally balanced SMTM are natural directions for future work.

\section*{Acknowledgement}
The authors are grateful to Roberto Casarin for helpful suggestions regarding the relevant literature. JY acknowledges support from the Independent Research Fund Denmark (DFF) through the Sapere Aude Starting Grant (No.~5251-00032B).



\setcounter{section}{0}
\setcounter{subsection}{0}
\setcounter{subsubsection}{0}

\numberwithin{equation}{section}
\setcounter{equation}{0}

\setcounter{theorem}{0}

\setcounter{figure}{0}
\setcounter{table}{0}

\begin{appendix}
\section{Proofs in \cref{Sec:SMTM}}
\subsection{Proof of \cref{prop:detailed balance}}\label{proof:prop:detailed balance}
\begin{proof}
    Denote by $\pi_S$ the pushforward of $\pi$ under $\SP^{-1}$ onto $\mathbb{S}^d$. For simplicity, we write its density with respect to the surface measure on $\mathbb{S}^d$ as $\pi_S(\cdot)$. Let $Q_M(z,\hat{z})$ and $Q_S(z,\hat{z})$ denote the proposal kernels from $z$ to $\hat{z}$ for SMTM and SRWM, respectively, with respect to the surface measure on $\mathbb{S}^d$. The full SMTM kernel, including the Metropolis step, is denoted by $P_M$. We first show that the kernel $P_M(z,\hat{z})$ satisfies the detailed balance condition with respect to $\pi_S$, i.e. $\pi_S(z) P_M(z,\hat{z}) = \pi_S(\hat{z}) P_M(\hat{z},z)$, which implies that $P_M(z,\hat{z})$ is $\pi_S$-reversible and, equivalently, that $\pi$ is an invariant distribution of the Markov chain induced by SMTM.

    $Q_M(z,\hat{z})$ can be viewed as that firstly $N$ candidates $\hat{z}_1,\dots,\hat{z}_N$ are proposed from $Q_S(z,\cdot)$, then $\hat{z}=\hat{z}_j$ is selected with probability $\omega(z,\hat{z}_j)/\sum_{i=1}^{N}\omega(z,\hat{z}_i)$. Since the candidates are i.i.d.\,, we may assume that it is always the $N$-th candidate that is selected. Thus,
    \begin{equation}
        Q_M(z,\hat{z})=NQ_S(z,\hat{z})\int_{\mathbb{S}^d}\dots\int_{\mathbb{S}^d}Q_S(z,\hat{z}_1)\dots Q_S(z,\hat{z}_{N-1})\frac{\omega(z,\hat{z})}{\sum_{i=1}^{N-1}\omega(z,\hat{z}_i)+\omega(z,\hat{z})}\dee\hat{z}_1\dots\dee\hat{z}_{N-1}
    \end{equation}
    Conditional on $\hat{z}_1,\dots,\hat{z}_N$ with $\hat{z}_N=\hat{z}$ and $z_1^\ast,\dots,z_{N-1}^\ast$, the acceptance rate can be written as
    \begin{equation}
        \alpha(z,\hat{z})=1\wedge\frac{\pi_S(\hat{z})\omega(\hat{z},z)/\left(\sum_{i=1}^{N-1}\omega(\hat{z},z_i^\ast)+\omega(\hat{z},z)\right)}{\pi_S(z)\omega(z,\hat{z})/\left(\sum_{i=1}^{N-1}\omega(z,\hat{z}_i)+\omega(z,\hat{z})\right)}.
    \end{equation}
    Thus, we have
    \begin{equation}\label{eq:transition_P_M}
        \begin{aligned}
            P_M(z,\hat{z})&=NQ_S(z,\hat{z})\int_{\mathbb{S}^d}\dots\int_{\mathbb{S}^d}Q_S(z,\hat{z}_1)\dots Q_S(z,\hat{z}_{N-1})Q_S(\hat{z},z_1^\ast)\dots Q_S(\hat{z},z_{N-1}^{\ast})\\
            &\quad\cdot\frac{\omega(z,\hat{z})}{\sum_{i=1}^{N-1}\omega(z,\hat{z}_i)+\omega(z,\hat{z})}\left(1\wedge\frac{\pi_S(\hat{z})\omega(\hat{z},z)/\left(\sum_{i=1}^{N-1}\omega(\hat{z},z_i^\ast)+\omega(\hat{z},z)\right)}{\pi_S(z)\omega(z,\hat{z})/\left(\sum_{i=1}^{N-1}\omega(z,\hat{z}_i)+\omega(z,\hat{z})\right)}\right)\\
            &\quad\cdot\dee\hat{z}_1\dots\dee\hat{z}_{N-1}\dee z_1^\ast\dots\dee z_{N-1}^\ast\\
            &=\pi_S(\hat{z})NQ_S(z,\hat{z})\int_{\mathbb{S}^d}\dots\int_{\mathbb{S}^d}Q_S(z,\hat{z}_1)\dots Q_S(z,\hat{z}_{N-1})Q_S(\hat{z},z_1^\ast)\dots Q_S(\hat{z},z_{N-1}^{\ast})\\
            &\quad\cdot\frac{\omega(z,\hat{z})}{\sum_{i=1}^{N-1}\omega(z,\hat{z}_i)+\omega(z,\hat{z})}\frac{\omega(\hat{z},z)}{\sum_{i=1}^{N-1}\omega(\hat{z},z_i^\ast)+\omega(\hat{z},z)}\\
            &\quad\cdot\left(\frac{\sum_{i=1}^{N-1}\omega(\hat{z},z_i^\ast)+\omega(\hat{z},z)}{\omega(\hat{z},z)\pi_S(\hat{z})}\wedge\frac{\sum_{i=1}^{N-1}\omega(z,\hat{z}_i)+\omega(z,\hat{z})}{\omega(z,\hat{z})\pi_S(z)}\right)\\
            &\quad\cdot\dee\hat{z}_1\dots\dee\hat{z}_{N-1}\dee z_1^\ast\dots\dee z_{N-1}^\ast=:\pi_S(\hat{z})F_M(z,\hat{z}).
        \end{aligned}
    \end{equation}
    Since $Q_S$ and the integral term are both symmetric in $z$ and $\hat{z}$, $F_M$ is also symmetric in $z$ and $\hat{z}$. Therefore, the detailed balance condition holds:
    \begin{equation}
        \pi_S(z)P_M(z,\hat{z})=\pi_S(z)\pi_S(\hat{z})F_M(z,\hat{z})=\pi_S(\hat{z})\pi_S(z)F_M(\hat{z},z)=\pi_S(\hat{z})P_M(\hat{z},z).
    \end{equation}

    Finally, we establish ergodicity of the Markov chain. It suffices to show that the chain is $\pi$-irreducible and aperiodic \citep{meyn2012markov}. Recall the result from \citet{fontaine2022}:
    \begin{proposition}\cite[Proposition 4.3]{fontaine2022}
        Let $P$ be an MTM transition for a target density $\pi$ with connected support $\mathcal{X}$. Suppose that $\pi$ is bounded above on $\mathcal{X}$ and below on any compact subset of $\mathcal{X}$. Suppose that, for each $k=1,\dots,K$, there exist $\delta,\varepsilon>0$ such that the marginal proposal densities are symmetric and satisfy
        \begin{equation}
            q^{(k)}(y|x)>\varepsilon,\quad\forall x,y\in\mathcal{X}:\|y-x\|_2<\delta.
        \end{equation}
        Suppose that, for all fixed $x\in\mathcal{X}$, the weights $\omega^{(k)}(\cdot|x)$ are bounded above and below on any compact set. Then, the kernel $P$ is $\pi$-irreducible and aperiodic.
    \end{proposition}
    Let $\mathcal{X}=\mathbb{S}^d\setminus\mathscr{N}$, which is connected. Under the assumptions on $\pi$, the transformed density $\pi_S$ is bounded above on $\mathcal{X}$ and bounded away from 0 on any compact subset of $\mathcal{X}$. The globally and locally-balanced weight functions inherit the corresponding bounds. Moreover, $Q_S(z,\hat{z})$ depends only on the distance between $z$ and $\hat{z}$, and its radial density is continuous and strictly positive in a neighborhood of 0. Hence $Q_S$ is uniformly locally positive. The argument of \citet[Proposition 4.3]{fontaine2022} therefore shows that $P_M$ is $\pi_S$-irreducible and aperiodic. Equivalently, the SMTM chain in the original coordinates is $\pi$-irreducible and aperiodic.
\end{proof}

\subsection{Proof of \cref{prop:mtm non uniform}}\label{proof:prop:mtm non uniform}
Some definitions are needed for the proof of Proposition \ref{prop:mtm non uniform}. A Markov transition kernel $P$ on $\mathbb{R}^d$ is called to have uniformly tight increment distributions if for every $\epsilon>0$ there exists $K>0$ such that for all $x$ it holds $P(x,B(x,K))\geq 1-\epsilon$, where $B(x,K)$ denotes the open ball with center $x$ and radius $K$ w.r.t.\,Euclidean distance. Moreover, $P$ is said to have geometric drift towards the set $C$ if there exist a function $V\geq1$, finite for at least one $x$, and constant $\lambda<1$ and $b<\infty$ such that $PV(x)\leq\lambda V(x)+b\mathbf{1}_C(x)$. And a set $C$ is called small if there exists $n>0,\delta>0$ and a probability measure $\nu$ such that for $x\in C$ we have $P^n(x,\cdot)\geq\delta\nu(\cdot)$.

\begin{proof}[Proof of \cref{prop:mtm non uniform}]
   The proof below follows a similar argument to that in \citet{jarner2000geometric}. Since RWM has uniformly tight increment distributions, by the MTM proposal structure, we get that MTM also has uniformly tight increment distributions. Then by \citet[Lemma 2.2]{jarner2000geometric}, for such a Markov chain, every small set is bounded.

    Next, by  \citet[Theorem 3.1(ii)]{jarner2000geometric}, if a Markov chain is geometrically ergodic, then for some small set, the geometric drift condition holds. Under both conditions, by \citet[Lemma 3.2]{jarner2000geometric}, for such a Markov kernel $P$, there exist $\rho>0$ and $S>0$ such that 
    \begin{equation}
    \label{Eq:V(x) lower bound}
        V(x)\geq\exp\left(\rho\mathbb{E}_x\left[\tau_{B(0,S)}\right]\right),\quad\forall x: \|x\|\geq S,
    \end{equation}
    where $\tau_A=\min\{n\geq1\mid X(n)\in A\}$ is the first return time of the Markov chain $\{X(n)\}$ to the set $A$.

    Now we prove that for MTM with a given number of candidates $N$, if the single-try proposal $q(\cdot)$ satisfies $\int \|x\|q(\|x\|) \dee x<\infty$, there exist $s>0$ and $c>0$ such that
    \begin{equation}
        V(x)\geq c\exp(s\|x\|).
    \end{equation}
    Choose $S$ used in (\ref{Eq:V(x) lower bound}) and fix $\|x\|\geq S$. Let $\{X(i)\}$ be the Markov chain generated by MTM starting at $x$ and denote the proposed increments as $I^k(i)=X^k(i)-X(i-1)$, where $X^k(i)$ is the $k$-th candidate at the $i$-th iteration. Clearly $I^k(i)$'s are i.i.d.. For $i\geq1$, let
    \begin{equation}
        J^k(i)=\left\langle \mathbf{n}(X(i-1)),I^k(i)\right\rangle\mathbf{1}_{\left\{\left\langle \mathbf{n}(X(i-1)),I^k(i)\right\rangle<0\right\}},
    \end{equation}
    where $\langle\cdot,\cdot\rangle$ denotes the Euclidean inner product,  $\mathbf{n}(v)=v/\|v\|$ denotes the unit vector and if $v=0$, $\mathbf{n}(v)$ can be chosen as any unit vector.

    Since the proposal increment $I^k(i)$ has density $z\mapsto q(\|z\|)$, its distribution is invariant under rotations. Therefore, for any unit vector $u$, the distribution of $\langle u,I^k(i)\rangle$ does not depend on $u$. Taking $u=\mathbf{n}(X(i-1))$, we see that the conditional distribution of $\langle \mathbf{n}(X(i-1)),I^k(i)\rangle$ given $X(i-1)$ is the same for every value of $X(i-1)$. Consequently, the conditional distribution of $J^k(i)$ also does not depend on $X(i-1)$. This shows that $J^k(i)$ is independent of $X(i-1)$. Using this and the fact that $I^k(i)$'s are i.i.d.\,we get that $J^k(i)$'s are also i.i.d.\,random variables. Since $J^k(i)$ is independent with $X(i-1)$, we can choose the unit vector of $X(i-1)$ as $\mathbf{n}(X(i-1))=(0,\dots,0,1)$ and the CDF of $J^k(i)$ can be written as
    \begin{equation}
    \label{Eq:CDF of J}
        \mathbb{P}(J^k(i)\leq t)=\left\{
        \begin{aligned}
            &\int_{-\infty}^t\int_{\mathbb{R}^{d-1}}q(\|x\|) \dee x_1\dots  \dee x_{d-1} \dee x_d & t<0,\\
            &1 & t\geq0.
        \end{aligned}
        \right.
    \end{equation}

    Define a one-sided walk $W(i)$ on $\mathbb{R}$ by
    \begin{equation}
        W(0)=\|X(0)\|=\|x\|,\quad W(i)=W(i-1)+\min_kJ^k(i).
    \end{equation}
    We claim that $W(i)\leq \|X(i)\|$ for any $i$. Indeed, for $i=0$ we have $W(0)=\|X(0)\|$ and for $i>0$ we will show by induction that $W(i)\leq \|X(i)\|$ under the assumption of $W(i-1)\leq \|X(i-1)\|$.
    \begin{itemize}
        \item If $X(i)=X(i-1)$, since $W(i)\leq W(i-1)$ by $J^k(i)\leq0$ and $W_{i-1}\leq\|X(i-1)\|$ by assumption, we get that $W(i)\leq \|X(i-1)\|=\|X(i)\|$.
        \item If $X(i)=X(i-1)+I^l(i)$ for some $l$, then
        \begin{equation}
            \begin{aligned}
                \|X(i)\|&=\left\|X(i-1)+I^l(i)\right\|\\
                &=\left\|X(i-1)+\mathbf{n}(X(i-1))\left\langle \mathbf{n}(X(i-1)),I^l(i)\right\rangle+I^l(i)-\mathbf{n}(X(i-1))\left\langle \mathbf{n}(X(i-1)),I^l(i)\right\rangle\right\|\\
                &\geq \left\|X(i-1)+\mathbf{n}(X(i-1))\left\langle \mathbf{n}(X(i-1)),I^l(i)\right\rangle\right\|\\
                &=\left\|\mathbf{n}(X(i-1))\|X(i-1)\|+\mathbf{n}(X(i-1))\left\langle \mathbf{n}(X(i-1)),I^l(i)\right\rangle\right\|\\
                &=\left|\|X(i-1)\|+\left\langle \mathbf{n}(X(i-1)),I^l(i)\right\rangle\right|
                \geq\|X(i-1)\|+\left\langle \mathbf{n}(X(i-1)),I^l(i)\right\rangle\mathbf{1}_{\left\{\left\langle \mathbf{n}(X(i-1)),I^l(i)\right\rangle<0\right\}}\\
                &=\|X(i-1)\|+J^l(i)
                \geq\|X(i-1)\|+\min_k J^k(i)
                \geq W(i-1)+\min_k J^k(i)=W(i).
            \end{aligned}
        \end{equation}
    \end{itemize}
    Denote $\hat{\tau}_S$ as the first return time to $(-\infty,S]$ of $\{W(i)\}_{i=1}^{\infty}$. Since $W(i)\leq \|X(i)\|$, we have $\hat{\tau}_S\leq\tau_{B(0,S)}$ and $\mathbb{E}_x[\hat{\tau}_S]\leq\mathbb{E}_x[\tau_{B(0,S)}]$. Let $\gamma$ denote the mean decrease of $\{W(i)\}$, then by (\ref{Eq:CDF of J}) we have
    \begin{equation}
        \begin{aligned}
            \gamma&=-\mathbb{E}\left[\min_kJ^k(1)\right]=\mathbb{E}\left[\max_k(-J^k(1))\right]\leq\mathbb{E}\left[\sum_k(-J^k(1))\right]=N\mathbb{E}\left[-J^1(1)\right]\\
            &=-N\int_{-\infty}^0x_d\int_{\mathbb{R}^{d-1}}q(\|x\|) \dee x_1\dots  \dee x_{d-1} \dee x_d\\
            &\leq N\int_{\mathbb{R}^d} \|x\| q(\|x\|)\dee x < \infty.
        \end{aligned}
    \end{equation}
    Then by Wald's formula \citep{wald1945} we get that
    \begin{equation}
        \mathbb{E}_x\left[\hat{\tau}_S\right]\geq\frac{\|x\|-S}{\gamma}.
    \end{equation}
    Now by (\ref{Eq:V(x) lower bound}) and $\mathbb{E}_x[\hat{\tau}_S]\leq\mathbb{E}_x[\tau_{B(0,S)}]$, we can lower bound $V(x)$ by
    \begin{equation}
        V(x)\geq\exp(\rho\frac{\|x\|-S}{\gamma}).
    \end{equation}
    Choosing $s$ and $c$ properly we have proven that
    \begin{equation}
        V(x)\geq c\exp(s\|x\|).
    \end{equation}
    Finally, by \citet[Corollary 3.4]{jarner2000geometric} we finish the proof of Proposition \ref{prop:mtm non uniform}.
\end{proof}

\subsection{Proof of \cref{Thm:uniform ergodic}}\label{proof:Thm:uniform ergodic}
We now present the proof of \cref{Thm:uniform ergodic}. Before doing so, we introduce some definitions and a lemma.
We define the following areas on the sphere $\mathbb{S}^d$: the ``Arctic Circle", $\AC(\epsilon):=\{z\in\mathbb{S}^d:z_{d+1}\geq1-\epsilon\}$ and the almost ``hemisphere" from some point $z$, $\HS(z,\epsilon'):=\{z'\in\mathbb{S}^d:z^Tz'\geq\epsilon'\}$.
\begin{lemma}
\label{Lemma:uniform ergodic}
    Denote $P_M(z,\hat{z})$ as the transition probability for moving from $z$ to $\hat{z}$ for SMTM and $Q_S(z,\hat{z})$ the proposal kernel from $z$ to $\hat{z}$ for SRWM. Then there exist positive constants $\delta_\epsilon$, $\delta'_{\epsilon'}$, $S_{\epsilon,\epsilon',d}$ and $M$ such that
    \begin{itemize}
        \item (globally-balanced case) when the weight function is taken as $\omega(z,\hat{z})=\frac{\pi_S(\hat{z})}{\pi_S(z)}$, 
        \begin{equation}
            P_M(z,\hat{z})\geq Q_S(z,\hat{z})(\delta'_{\epsilon'})^{2N-2}\frac{\pi_S(\hat{z})}{M}\left(1\wedge\frac{(N-1)\delta_\epsilon+\pi_S(\hat{z})}{NM}\right)S^{2N-2}_{\epsilon,\epsilon',d}.
        \end{equation}
        In particular, if $\hat{z}\notin\AC(\epsilon)$, we have
        \begin{equation}
            P_M(z,\hat{z})\geq Q_S(z,\hat{z})(\delta'_{\epsilon'})^{2N-2}\frac{\delta_\epsilon}{M}\left(1\wedge\frac{\delta_\epsilon}{M}\right)S^{2N-2}_{\epsilon,\epsilon',d}.
        \end{equation}
        \item (locally-balanced case) when the weight function is taken as $\omega(z,\hat{z})=\sqrt{\frac{\pi_S(\hat{z})}{\pi_S(z)}}$, 
        \begin{equation}
            P_M(z,\hat{z})\geq Q_S(z,\hat{z})(\delta'_{\epsilon'})^{2N-2}\sqrt{\frac{\pi_S(\hat{z})}{M}}\left(1\wedge\frac{(N-1)\sqrt{\delta_\epsilon\pi_S(\hat{z})}+\pi_S(\hat{z})}{NM}\right)S^{2N-2}_{\epsilon,\epsilon',d}.
        \end{equation}
        In particular, if $\hat{z}\notin\AC(\epsilon)$, we have
        \begin{equation}
            P_M(z,\hat{z})\geq Q_S(z,\hat{z})(\delta'_{\epsilon'})^{2N-2}\sqrt{\frac{\delta_\epsilon}{M}}\left(1\wedge\frac{\delta_\epsilon}{M}\right)S^{2N-2}_{\epsilon,\epsilon',d}.
        \end{equation}
    \end{itemize}
\end{lemma}

\begin{proof}
    We prove for the globally-balanced case. The proof is analogous for the locally-balanced case. By \cref{eq:transition_P_M}, we have
    \begin{equation}
        \begin{aligned}
            \quad P_M(z,\hat{z})&=N\int_{\mathbb{S}^d}\dots\int_{\mathbb{S}^d} Q_S(z,\hat{z})Q_S(z,\hat{z}_1)\dots Q_S(z,\hat{z}_{N-1})\\
            &\quad\cdot\frac{\omega(z,\hat{z})}{\sum_{i=1}^{N-1}\omega(z,\hat{z}_i)+\omega(z,\hat{z})}\left (1\wedge\frac{\pi_S(\hat{z})\omega(\hat{z},z)/(\sum_{i=1}^{N-1}\omega(\hat{z},z_i^{\ast})+\omega(\hat{z},z))}{\pi_S(z)\omega(z,\hat{z})/(\sum_{i=1}^{N-1}\omega(z,\hat{z}_i)+\omega(z,\hat{z}))}\right)\\
            &\quad\cdot Q_S(\hat{z},z_1^{\ast})\dots Q_S(\hat{z},z_{N-1}^{\ast}) \dee\hat{z}_1\dots  \dee\hat{z}_{N-1} \dee z_1^{\ast}\dots  \dee z_{N-1}^{\ast}
        \end{aligned}
    \end{equation}
    Since $\omega(z,\hat{z})=\frac{\pi_S(\hat{z})}{\pi_S(z)}$, it follows that
    \begin{equation}
        \begin{aligned}
            P_M(z,\hat{z})&=N\int_{\mathbb{S}^d}\dots\int_{\mathbb{S}^d} Q_S(z,\hat{z})Q_S(z,\hat{z}_1)\dots Q_S(z,\hat{z}_{N-1})\\
            &\quad\cdot\frac{\pi_S(\hat{z})}{\sum_{i=1}^{N-1}\pi_S(\hat{z}_i)+\pi_S(\hat{z})}\left(1\wedge\frac{\sum_{i=1}^{N-1}\pi_S(\hat{z}_i)+\pi_S(\hat{z})}{\sum_{i=1}^{N-1}\pi_S(z_i^{\ast})+\pi_S(z)}\right)\\
            &\quad\cdot Q_S(\hat{z},z_1^{\ast})\dots Q_S(\hat{z},z_{N-1}^{\ast}) \dee\hat{z}_1\dots  \dee\hat{z}_{N-1} \dee z_1^{\ast}\dots  \dee z_{N-1}^{\ast}
        \end{aligned}
    \end{equation}
    Since $\sup_{x\in\mathbb{R}^d}\pi(x)(R^2+\|x\|^2)^d<\infty$, for any $z\in\mathbb{S}^d$, there exists $M>0$ such that $\pi_S(z)\leq M$. Moreover, if $z\notin \AC(\epsilon)$, then there exist $\delta_\epsilon>0$ such that $\pi_S(z)\geq\delta_\epsilon$. And if $\hat{z}\in\HS(z,\epsilon')$, then there exists $\delta'_{\epsilon'}>0$ such that $Q_S(z,\hat{z})\geq\delta'_{\epsilon'}$.
    Therefore, 
    \begin{equation}
        \begin{aligned}
            P_M(z,\hat{z})&\geq N\int_{\HS(\hat{z},\epsilon')\backslash\AC(\epsilon)}\dots\int_{\HS(\hat{z},\epsilon')\backslash\AC(\epsilon)}\int_{\HS(z,\epsilon')\backslash\AC(\epsilon)}\dots\int_{\HS(z,\epsilon')\backslash\AC(\epsilon)}\\
            &\quad\cdot Q_S(z,\hat{z})Q_S(z,\hat{z}_1)\dots Q_S(z,\hat{z}_{N-1})\\
            &\quad\cdot\frac{\pi_S(\hat{z})}{\sum_{i=1}^{N-1}\pi_S(\hat{z}_i)+\pi_S(\hat{z})}\left(1\wedge\frac{\sum_{i=1}^{N-1}\pi_S(\hat{z}_i)+\pi_S(\hat{z})}{\sum_{i=1}^{N-1}\pi_S(z_i^{\ast})+\pi_S(z)}\right)\\
            &\quad\cdot Q_S(\hat{z},z_1^{\ast})\dots Q_S(\hat{z},z_{N-1}^{\ast}) \dee\hat{z}_1\dots  \dee\hat{z}_{N-1} \dee z_1^{\ast}\dots  \dee z_{N-1}^{\ast}\\
            &\geq NQ_S(z,\hat{z})(\delta'_{\epsilon'})^{N-1}\frac{\pi_S(\hat{z})}{NM}\left(1\wedge\frac{(N-1)\delta_\epsilon+\pi_S(\hat{z})}{NM}\right)(\delta'_{\epsilon'})^{N-1}\\
            &\quad\cdot\int_{\HS(\hat{z},\epsilon')\backslash\AC(\epsilon)}\dots\int_{\HS(\hat{z},\epsilon')\backslash\AC(\epsilon)}\int_{\HS(z,\epsilon')\backslash\AC(\epsilon)}\dots\int_{\HS(z,\epsilon')\backslash\AC(\epsilon)} \dee\hat{z}_1\dots  \dee\hat{z}_{N-1}\\
            &\quad\cdot  \dee z_1^{\ast}\dots  \dee z_{N-1}^{\ast}
        \end{aligned}
    \end{equation}
    Denote the area of $\HS(z_N,\epsilon')\backslash\AC(\epsilon)$ as $S_{\epsilon,\epsilon',d}$, where $z_N$ is the north pole of sphere $\mathbb{S}^d$. Then we have
    \begin{equation}
        P_M(z,\hat{z})\geq Q_S(z,\hat{z})(\delta'_{\epsilon'})^{2N-2}\frac{\pi_S(\hat{z})}{M}\wedge\left(1\wedge\frac{(N-1)\delta_\epsilon+\pi_S(\hat{z})}{NM}\right)S^{2N-2}_{\epsilon,\epsilon',d}.
    \end{equation}
    If $\hat{z}\notin\AC(\epsilon)$, then $\pi_S(\hat{z})\geq\delta_{\epsilon}$ and
    \begin{equation}
        P_M(z,\hat{z})\geq Q_S(z,\hat{z})(\delta'_{\epsilon'})^{2N-2}\frac{\delta_\epsilon}{M}\left(1\wedge\frac{\delta_\epsilon}{M}\right)S^{2N-2}_{\epsilon,\epsilon',d}.
    \end{equation}
\end{proof}

\begin{proof}[Proof of \cref{Thm:uniform ergodic}]
    We prove for the globally-balanced case. The proof is analogous for the locally-balanced case. It is equivalent to showing that the sphere $\mathbb{S}^d$ is a small set. Denote $P_M$ as the transition kernel of SMTM. We will show that, for all measurable sets $A$, there exist a probability measure $\mu$ on $\mathbb{S}^d$ and a constant $\epsilon > 0$ such that
    \begin{equation}
        P_M^3(z,A)\geq\epsilon\mu(A).
    \end{equation}
    This minorization condition implies uniform ergodicity \citep{meyn2012markov}. Consider the 3-step path $z\to z_1\to z_2\to z'$, then we have
    \begin{equation}
        P_M^3(z,A)=\int_{z'\in A}\int_{z_2\in\mathbb{S}^d}\int_{z_1\in\mathbb{S}^d}P_M(z,z_1)P_M(z_1,z_2)P_M(z_2,z') \dee z_1 \dee z_2\dee z'.
    \end{equation}
    By \cref{Lemma:uniform ergodic},
    \begin{equation}
        \begin{aligned}
            &\quad \int_{z_2\in\mathbb{S}^d}\int_{z_1\in\mathbb{S}^d}P_M(z,z_1)P_M(z_1,z_2)P_M(z_2,z') \dee z_1 \dee z_2\\
            &\geq\int_{z_2\in\HS(z',\epsilon')\backslash\AC(\epsilon)}\int_{z_1\in\HS(z,\epsilon')\backslash\AC(\epsilon)}P_M(z,z_1)P_M(z_1,z_2)P_M(z_2,z') \dee z_1 \dee z_2\\
            &\geq\int_{z_2\in\HS(z',\epsilon')\backslash\AC(\epsilon)}\int_{z_1\in\HS(z,\epsilon')\backslash\AC(\epsilon)}Q_S(z,z_1)Q_S(z_1,z_2)\left[(\delta'_{\epsilon'})^{2N-2}\frac{\delta_\epsilon}{M}\left(1\wedge\frac{\delta_\epsilon}{M}\right)S_{\epsilon,\epsilon',d}^{2N-2}\right]^2\\
            &\quad\cdot Q_S(z_2,z')(\delta'_{\epsilon'})^{2N-2}\frac{\pi_S(z')}{M}\left(1\wedge\frac{(N-1)\delta_\epsilon+\pi_S(z')}{NM}\right)S_{\epsilon,\epsilon',d}^{2N-2} \dee z_1 \dee z_2\\
            &\geq (\delta'_{\epsilon'})^{6N-3}S_{\epsilon,\epsilon',d}^{6N-6}\left(\frac{\delta_\epsilon}{M}\right)^2\left(1\wedge\frac{\delta_\epsilon}{M}\right)^2\frac{\pi_S(z')}{M}\left(1\wedge\frac{(N-1)\delta_\epsilon+\pi_S(z')}{NM}\right)\\
            &\quad\cdot \int_{z_2\in\HS(z',\epsilon')\backslash\AC(\epsilon)}\int_{z_1\in\HS(z,\epsilon')\backslash\AC(\epsilon)}\mathbf{1}_{z_1\in\HS(z_2,\epsilon')} \dee z_1 \dee z_2.
        \end{aligned}
    \end{equation}
    Using the same arguments as \citet[Proof of Theorem 2.1]{yang2024}, the last integration term can be bounded by a positive constant $C$:
    \begin{equation}
        \inf_{z,z'\in\mathbb{S}^d}\int_{z_2\in\HS(z',\epsilon')\backslash\AC(\epsilon)}\int_{z_1\in\HS(z,\epsilon')\backslash\AC(\epsilon)}\mathbf{1}_{z_1\in\HS(z_2,\epsilon')} \dee z_1 \dee z_2\geq\frac{1}{2}C>0,
    \end{equation}
    where 
    \begin{equation}
        C:=\inf_{z,z'\in\mathbb{S}^d}\int_{z_2\in\HS(z',0)\backslash \AC(0)}\int_{z_1\in\HS(z,0)\backslash\AC(0)}\mathbf{1}_{z_1\in\HS(z_2,0)}\dee z_1\dee z_2>0.
    \end{equation}
    Therefore,
    \begin{equation}
        \begin{aligned}
            \int_{z_2\in\mathbb{S}^d}\int_{z_1\in\mathbb{S}^d}&P_M(z,z_1)P_M(z_1,z_2)P_M(z_2,z') \dee z_1 \dee z_2\\
            &\geq\frac{1}{2}C(\delta'_{\epsilon'})^{6N-3}S_{\epsilon,\epsilon',d}^{6N-6}\left(\frac{\delta_\epsilon}{M}\right)^2\left(1\wedge\frac{\delta_\epsilon}{M}\right)^2\frac{\pi_S(z')}{M}\left(1\wedge\frac{(N-1)\delta_\epsilon+\pi_S(z')}{NM}\right).
        \end{aligned}
    \end{equation}
    Denoting 
    \begin{equation}
        g(z')=\frac{1}{2}C(\delta'_{\epsilon'})^{6N-3}S_{\epsilon,\epsilon',d}^{6N-6}\left(\frac{\delta_\epsilon}{M}\right)^2\left(1\wedge\frac{\delta_\epsilon}{M}\right)^2\frac{\pi_S(z')}{M}\left(1\wedge\frac{(N-1)\delta_\epsilon+\pi_S(z')}{NM}\right),
    \end{equation}
    then $g(z)>0$ almost surely. It follows that
    \begin{equation}
        P_M^3(z,A)\geq\int_{A}g(z') \dee z'=\int_{\mathbb{S}^d}g(z') \dee z'\frac{\int_{A}g(z') \dee z'}{\int_{\mathbb{S}^d}g(z') \dee z'}=:\epsilon\mu(A),
    \end{equation}
    where $\epsilon=\int_{\mathbb{S}^d}g(z')\dee z'$ and $\mu(A)=\frac{\int_{A}g(z')\dee z'}{\int_{\mathbb{S}^d}g(z')\dee z'}$ is a probability measure.
    
\end{proof}

\section{Proofs in \cref{Sec:optimal scaling}}

\subsection{Proof of \cref{lemma:alpha computation}}\label{proof:lamma:alpha computation}
\begin{proof}

Note that $\alpha_1^j$ is exactly the probability in \cref{eq:SMTM_alpha1} in \cref{alg:SMTM}. In the globally-balanced case, we have
\begin{equation}
    \begin{aligned}
        \alpha_1^j&=1\wedge\frac{\pi_S(\hat{z}_j)\omega(\hat{z}_j,z)/(\sum_{i=1}^{N-1}\omega(\hat{z}_j,z_i^{\ast})+\omega(\hat{z}_j,z))}{\pi_S(z)\omega(z,\hat{z}_j)/(\sum_{i=1}^N\omega(z,\hat{z}_i))}\\
        &=1\wedge\frac{\sum_{i=1}^N\pi_S(\hat{z}_i)}{\sum_{i=1}^{N-1}\pi_S(z_i^{\ast})+\pi_S(z)}\\
        &=1\wedge\frac{\sum_{i=1}^N\frac{\pi_S(\hat{z}_i)}{\pi_S(z)}}{\sum_{i=1}^{N-1}\frac{\pi_S(z_i^{\ast})}{\pi_S(\hat{z}_j)}\frac{\pi_S(\hat{z}_j)}{\pi_S(z)}+1}.
    \end{aligned}
\end{equation}
For $\alpha_2^j$, by the definition of conditional probability, it can be viewed as the product between $\alpha_1^j$ and the probability to choose $\hat{z}_j$, which is $\frac{\omega(z,\hat{z}_j)}{\sum_i\omega(z,\hat{z}_i)}$. Therefore, in the globally-balanced case, we have
\begin{equation}
    \begin{aligned}
        \alpha_2^j&=\frac{\pi_S(\hat{z}_j)}{\sum_{i=1}^N\pi_S(\hat{z}_i)}\cdot\alpha_1^j\\
        &=\frac{\pi_S(\hat{z}_j)}{\sum_{i=1}^N\pi_S(\hat{z}_i)}\left(1\wedge\frac{\sum_{i=1}^N\pi_S(\hat{z}_i)}{\sum_{i=1}^{N-1}\pi_S(z_i^{\ast})+\pi_S(z)}\right)\\
        &=\frac{\pi_S(\hat{z}_j)}{\sum_{i=1}^N\pi_S(\hat{z}_i)}\wedge\frac{\pi_S(\hat{z}_j)}{\sum_{i=1}^{N-1}\pi_S(z_i^{\ast})+\pi_S(z)}\\
        &=\frac{\frac{\pi_S(\hat{z}_j)}{\pi_S(z)}}{\sum_{i=1}^N\frac{\pi_S(\hat{z}_i)}{\pi_S(z)}}\wedge\frac{\frac{\pi_S(\hat{z}_j)}{\pi_S(z)}}{\sum_{i=1}^{N-1}\frac{\pi_S(z_i^{\ast})}{\pi_S(\hat{z}_j)}\frac{\pi_S(\hat{z}_j)}{\pi_S(z)}+1}.
    \end{aligned}
\end{equation}
The proof for the locally-balanced case is analogous.
\end{proof}

\subsection{Proof of \cref{Lemma:lipschitz}}\label{proof:Lemma:lipschitz}
\begin{proof}
    Obviously, for any $j$, $\phi_1^j$ and $\phi_2^j$ are upper bounded by $1$ and lower bounded by $0$, so they are both bounded. It suffices to prove the Lipschitz continuity. For simplicity, we only prove it for $\phi_1^j$ in the globally-balanced case. For the other three functions, the proof follows similarly. Write 
    \begin{equation}
        \zeta^j\left((x_i)_{i=1}^{N},(y_i)_{i=1}^{N-1}\right)=\frac{\sum_{i=1}^{N}e^{x_i}}{\sum_{i=1}^{N-1}e^{x_j}e^{y_i}+1},
    \end{equation}
    so $\phi_1^j$ can be rewritten as $\phi_1^j=1\wedge\zeta^j$, and 
    \begin{equation}
        \begin{aligned}
            &D=\left\{\left((x_i)_{i=1}^{N},(y_i)_{i=1}^{N-1}\right)\in\mathbb{R}^{2N-1}:\zeta^j\left((x_i)_{i=1}^{N},(y_i)_{i=1}^{N-1}\right)\leq1\right\},\\
            &F=\left\{\left((x_i)_{i=1}^{N},(y_i)_{i=1}^{N-1}\right)\in\mathbb{R}^{2N-1}:\zeta^j\left((x_i)_{i=1}^{N},(y_i)_{i=1}^{N-1}\right)\geq1\right\}.
        \end{aligned}
    \end{equation}
    Then $D\cup F=\mathbb{R}^{2N-1}$ and $D\cap F=\{u:\zeta^j(u)=1\}$. By the definition of Lipschitz continuity, it suffices to show that there exists a positive constant $L$ such that for any $((x_i)_{i=1}^{N},(y_i)_{i=1}^{N-1})$ and $((\tilde{x}_i)_{i=1}^{N},(\tilde{y}_i)_{i=1}^{N-1})$ in $\mathbb{R}^{2N-1}$,
    \begin{equation}
        \left|\phi_1^j\left((x_i)_{i=1}^{N},(y_i)_{i=1}^{N-1}\right)-\phi_1^j\left((\tilde{x}_i)_{i=1}^{N},(\tilde{y}_i)_{i=1}^{N-1}\right)\right|\leq L\left\|\left((x_i)_{i=1}^{N},(y_i)_{i=1}^{N-1}\right)-\left((\tilde{x}_i)_{i=1}^{N},(\tilde{y}_i)_{i=1}^{N-1}\right)\right\|.
        \label{eq:lipschitz_def}
    \end{equation}
    If $((x_i)_{i=1}^{N},(y_i)_{i=1}^{N-1})\in F$ and $((\tilde{x}_i)_{i=1}^{N},(\tilde{y}_i)_{i=1}^{N-1})\in F$, then 
    \begin{equation}
        \phi_1^j\left((x_i)_{i=1}^{N},(y_i)_{i=1}^{N-1}\right)=\phi_1^j\left((\tilde{x}_i)_{i=1}^{N},(\tilde{y}_i)_{i=1}^{N-1}\right)=1.
    \end{equation}
    In this case, $\phi_1^j$ is a trivial constant function, and any constant function is Lipschitz continuous for any $L>0$.
    If $((x_i)_{i=1}^{N},(y_i)_{i=1}^{N-1})\in D$ and $((\tilde{x}_i)_{i=1}^{N},(\tilde{y}_i)_{i=1}^{N-1})\in D$, then 
    \begin{equation}
        \phi_1^j\left((x_i)_{i=1}^{N},(y_i)_{i=1}^{N-1}\right)=\zeta^j\left((x_i)_{i=1}^{N},(y_i)_{i=1}^{N-1}\right), 
    \end{equation}
    \begin{equation}
        \phi_1^j\left((\tilde{x}_i)_{i=1}^{N},(\tilde{y}_i)_{i=1}^{N-1}\right)=\zeta^j\left((\tilde{x}_i)_{i=1}^{N},(\tilde{y}_i)_{i=1}^{N-1}\right),
    \end{equation}
    so showing \cref{eq:lipschitz_def} is equivalent to showing that
    \begin{equation}
        \left|\zeta^j\left((x_i)_{i=1}^{N},(y_i)_{i=1}^{N-1}\right)-\zeta^j\left((\tilde{x}_i)_{i=1}^{N},(\tilde{y}_i)_{i=1}^{N-1}\right)\right|\leq L\left\|\left((x_i)_{i=1}^{N},(y_i)_{i=1}^{N-1}\right)-\left((\tilde{x}_i)_{i=1}^{N},(\tilde{y}_i)_{i=1}^{N-1}\right)\right\|.
    \end{equation}
    Taking the partial derivatives, we get that
    \begin{equation}
        \begin{aligned}
            &\frac{\partial\zeta^j}{\partial x_i}=\frac{e^{x_i}}{\sum_{i=1}^{N-1}e^{x_j}e^{y_i}+1}\quad(i\neq j),\\
            &\frac{\partial\zeta^j}{\partial x_j}=\frac{e^{x_j}}{\sum_{i=1}^{N-1}e^{x_j}e^{y_i}+1}-\frac{\sum_{i=1}^{N}e^{x_i}}{\sum_{i=1}^{N-1}e^{x_j}e^{y_i}+1}\cdot\frac{\sum_{i=1}^{N-1}e^{x_j}e^{y_i}}{\sum_{i=1}^{N-1}e^{x_j}e^{y_i}+1},\\
            &\frac{\partial\zeta^j}{\partial y_i}=\frac{-\sum_{i=1}^{N}e^{x_i}}{\sum_{i=1}^{N-1}e^{x_j}e^{y_i}+1}\cdot\frac{e^{x_j}e^{y_i}}{\sum_{i=1}^{N-1}e^{x_j}e^{y_i}+1},
        \end{aligned}
    \end{equation}
    which are all bounded since $\frac{\sum_{i=1}^{N}e^{x_i}}{\sum_{i=1}^{N-1}e^{x_j}e^{y_i}+1}\in (0,1]$. Thus $\zeta^j$ is Lipschitz on $D$.

    If $((x_i)_{i=1}^{N},(y_i)_{i=1}^{N-1})\in D$ and $((\tilde{x}_i)_{i=1}^{N},(\tilde{y}_i)_{i=1}^{N-1})\in F$, then 
    \begin{equation}
        \left|\phi_1^j\left((x_i)_{i=1}^{N},(y_i)_{i=1}^{N-1}\right)-\phi_1^j\left((\tilde{x}_i)_{i=1}^{N},(\tilde{y}_i)_{i=1}^{N-1}\right)\right|=\left|\zeta^j\left((x_i)_{i=1}^{N},(y_i)_{i=1}^{N-1}\right)-1\right|.
    \end{equation}
    Consider the line segment connecting $((x_i)_{i=1}^{N},(y_i)_{i=1}^{N-1})$ and $((\tilde{x}_i)_{i=1}^{N},(\tilde{y}_i)_{i=1}^{N-1})$, parameterized by $\gamma(t)=(1-t)\left((x_i)_{i=1}^{N},(y_i)_{i=1}^{N-1}\right)+t\left((\tilde{x}_i)_{i=1}^{N},(\tilde{y}_i)_{i=1}^{N-1}\right)$ for $t\in[0,1]$. Since $\zeta^j$ is continuous and $\zeta^j((x_i)_{i=1}^{N},(y_i)_{i=1}^{N-1})\leq 1,\zeta^j((\tilde{x}_i)_{i=1}^{N},(\tilde{y}_i)_{i=1}^{N-1})\geq 1$, by the Intermediate Value Theorem, there exists some $t^\ast\in[0,1]$ such that 
    \begin{equation}
        \left((\bar{x}_i)_{i=1}^{N},(\bar{y}_i)_{i=1}^{N-1}\right)=t^\ast\left((x_i)_{i=1}^{N},(y_i)_{i=1}^{N-1}\right)+(1-t^\ast)\left((\tilde{x}_i)_{i=1}^{N},(\tilde{y}_i)_{i=1}^{N-1}\right)\in D\cap F.
    \end{equation}
    Then, we get
    \begin{equation}
        \begin{aligned}
            \left|\zeta^j\left((x_i)_{i=1}^{N},(y_i)_{i=1}^{N-1}\right)-1\right|&=\left|\zeta^j\left((x_i)_{i=1}^{N},(y_i)_{i=1}^{N-1}\right)-\zeta^j\left((\bar{x}_i)_{i=1}^{N},(\bar{y}_i)_{i=1}^{N-1}\right)\right|\\
            &\leq L\left\|\left((x_i)_{i=1}^{N},(y_i)_{i=1}^{N-1}\right)-\left((\bar{x}_i)_{i=1}^{N},(\bar{y}_i)_{i=1}^{N-1}\right)\right\|\\
            &\leq L\left\|\left((x_i)_{i=1}^{N},(y_i)_{i=1}^{N-1}\right)-\left((\tilde{x}_i)_{i=1}^{N},(\tilde{y}_i)_{i=1}^{N-1}\right)\right\|,
        \end{aligned}
    \end{equation}
    where $L$ is the Lipschitz constant, and we have used the fact that $\zeta^j$ is Lipschitz on $D$.

    Thus, we have proved the Lipschitz continuity for $\phi_1^j$.
\end{proof}

\subsection{Proof of \cref{Thm:accept rate}}\label{proof:Thm:accept rate}
The proof of \cref{Thm:accept rate} relies heavily on \citet[Lemma 5.1]{yang2024}, which is given in \cref{Lemma:accept rate}.
\begin{lemma}
\label{Lemma:accept rate}\cite[Lemma 5.1]{yang2024}
    Under the assumptions on $\pi$ in Section \ref{Sec:optimal scaling}, suppose the current state of SRWM $X\sim\pi$ and the parameter of SRWM $R=\sqrt{\lambda d}$ where $\lambda>0$ is a fixed constant. Then, if either $\lambda\neq1$ or $f$ is not the standard Gaussian density, there exists a sequence of sets $\{F_d\}$ such that $\pi(F_d)\to1$ and
    \begin{equation}
        \sup_{X\in F_d}\mathbb{E}_{\hat{X}\mid X}\left[\left|1\wedge\frac{\pi(\hat{X})(R^2+\|\hat{X}\|^2)^d}{\pi(X)(R^2+\|X\|^2)^d}-1\wedge\exp(W_{\hat{X}\mid X})\right|\right]=o\left(d^{-1/4}\log(d)\right),
    \end{equation}
    where $W_{\hat{X}\mid X}\sim\mathcal{N}(\mu,\sigma^2)$,
    \begin{equation}
        \mu=\frac{\ell^2}{2}\left(\frac{4\lambda}{(1+\lambda)^2}-\mathbb{E}[((\log f)')^2]\right),\quad\sigma^2=\ell^2\left(\mathbb{E}[((\log f)')^2]-\frac{4\lambda}{(1+\lambda)^2}\right),
    \end{equation}
    and $\ell$ is a re-parametrization of $h$ satisfying
    \begin{equation}
        \frac{1}{\sqrt{1+h^2(d-1)}}=1-\frac{\ell^2}{2d}\frac{4\lambda}{(1+\lambda)^2}.
    \end{equation}
\end{lemma}
Furthermore, from the proof of \citet[Lemma 5.1]{yang2024} we can also get that
\begin{equation}
\label{Eq:accept convergence}
    \sup_{X\in F_d}\mathbb{E}_{\hat{X}\mid X}\left[\left|\log\frac{\pi(\hat{X})(R^2+\|\hat{X}\|^2)^d}{\pi(X)(R^2+\|X\|^2)^d}\mathbf{1}_{\hat{z}_{d+1}\leq1-\epsilon}-W_{\hat{X}\mid X}\right|\right]=o\left(d^{-1/4}\log(d)\right),
\end{equation}
and
\begin{equation}
\label{Eq:north pole}
    \sup_{X\in F_d}\mathbb{P}(\hat{z}_{d+1}>1-\epsilon)= o(d^{-1/2}),
\end{equation}
where $\hat{z}$ is the spherical proposal from $z=\SP^{-1}(X)$ and $F_d$ is explicitly constructed in \citet[Proof of Lemma 5.1]{yang2024} by
\begin{equation}\label{eq:definition_Fd1}
    F_d:=\bigcap_{k=1}^4\left\{x\in\mathbb{R}^d:|T_k(x)|<d^{-1/2}\log d\right\},
\end{equation}
where
\begin{equation}
    \begin{aligned}
        &T_1(x):=\frac{1}{d}\sum_{i=1}^d\left[(\log f(x_i))'\right]^2-\mathbb{E}_f\left[((\log f)')^2\right],\\
        &T_2(x):=\frac{1}{d}\sum_{i=1}^d\left[(\log f(x_i))''\right]-\mathbb{E}_f\left[(\log f)''\right],\\
        &T_3(x):=\frac{1}{d}\sum_{i=1}^dx_i\left(\log f(x_i)\right)'-\mathbb{E}_f\left[X(\log f)'\right],\\
        &T_4(x):=\frac{1}{d}\sum_{i=1}^dx_i^2-\mathbb{E}_f\left[X^2\right].
    \end{aligned}
\end{equation}

\begin{proof}[Proof of Theorem \ref{Thm:accept rate}]
    We prove the result for $\alpha_1^j$, and the proof for $\alpha_2^j$ is analogous. 

    Let 
    \begin{equation}
        H_d:=\bigcap_{k=1}^4\left\{x\in\mathbb{R}^d:|T_k(x)|<\frac{1}{2}d^{-1/2}\log d\right\}\subset F_d,
    \end{equation}
    then the same moment argument used for $F_d$ in \cite{yang2024} gives $\pi(H_d)\to 1$. We will first bound $\sup_{X\in H_d}\mathbb{P}_{\hat{X}|X}(\hat{X}\notin F_d)$. For $X\in H_d$, if $\hat{X}\notin F_d$, then necessarily,
    \begin{equation}
        \max_{1\leq k\leq 4}|T_k(\hat{X})-T_k(X)|\geq\frac{1}{2}d^{-1/2}\log d.
    \end{equation}
    Since $(\log f)'$ and $(\log f)''$ are Lipschitz continuous, by the definition of $H_d$ and the Cauchy--Schwarz inequality, there exists $C_1<\infty$ such that for any $k$
    \begin{equation}
        |T_k(\hat{X})-T_k(X)|\leq C_1\left(\frac{\|\hat{X}-X\|}{\sqrt{d}}+\frac{\|\hat{X}-\hat{X}\|^2}{d}\right).
    \end{equation}
    Under the scaling $h=O(d^{-1})$ and conditional on $z_{d+1},\hat{z}_{d+1}<1-\varepsilon$, there exists $C_2<\infty$ such that
    \begin{equation}
        \|\hat{X}-X\|\leq C_2\frac{\|U\|}{\sqrt{d}},\quad U\sim\mathcal{N}(0,I_{d+1}).
    \end{equation}
    Consequently, for a small enough constant $c>0$,
    \begin{equation}
        \|\hat{X}-X\|\leq c\log d\Rightarrow\max_{1\leq k\leq 4}|T_k(\hat{X})-T_k(X)|<\frac{1}{2}d^{-1/2}\log d.
    \end{equation}
    Thus,
    \begin{equation}\label{eq:escape_F_d}
        \sup_{X\in H_d}\mathbb{P}_{\hat{X}|X}(\hat{X}\notin F_d)\leq\sup_{x\in H_d}\mathbb{P}_{\hat{X}|X}(\|\hat{X}-X\|>c\log d)+\sup_{x\in H_d}\mathbb{P}_{\hat{X}|X}(\hat{z}_{d+1}>1-\varepsilon)=o(d^{-1/2}).
    \end{equation}
    
    Note that by the definition of $\alpha_1^j$ and $\phi_1^j$
    \begin{equation}
        \alpha_1^j=\phi_1^j\left(\left(\log\frac{\pi_S(\hat{z}_i)}{\pi_S(z)}\right)_{i=1}^{N},\left(\log\frac{\pi_S(z_i^{\ast})}{\pi_S(\hat{z}_j)}\right)_{i=1}^{N-1}\right).
    \end{equation}
    Denote $\hat{z}_{i,d+1}$ as the $(d+1)$-th coordinate of $\hat{z}_i$. For simplicity we denote the truncated log-likelihoods vector as $\mathcal{L}$ and the random variables vector as $\mathcal{R}$, i.e.
    \begin{equation}
        \begin{aligned}
            &\mathcal{L}=\left(\left(\log\frac{\pi_S(\hat{z}_i)}{\pi_S(z)}\mathbf{1}_{\hat{z}_{i,d+1}\leq1-\epsilon}\right)_{i=1}^{N},\left(\log\frac{\pi_S(z_i^{\ast})}{\pi_S(\hat{z}_j)}\mathbf{1}_{z_{i,d+1}^{\ast}\leq1-\epsilon}\right)_{i=1}^{N-1}\right),\\
            &\mathcal{R}=\left(\left(W_i\right)_{i=1}^{N},\left(V_i\right)_{i=1}^{N-1}\right).
        \end{aligned}
    \end{equation}
By the union bound, \cref{Eq:north pole,eq:escape_F_d}, the upper bound and Lipschitz continuity of $\phi_1^j$, we can get that
    \begin{equation}
        \begin{aligned}
            &\quad\sup_{X\in H_d}\mathbb{E}\left[\left|\phi_1^j\left(\left(\log\frac{\pi_S(\hat{z}_i)}{\pi_S(z)}\right)_{i=1}^{N},\left(\log\frac{\pi_S(z_i^{\ast})}{\pi_S(\hat{z}_j)}\right)_{i=1}^{N-1}\right)-\phi_1^j\left(\left(W_i\right)_{i=1}^{N},\left(V_i\right)_{i=1}^{N-1}\right)\right|\middle| X\right]\\
            &\leq\sup_{X\in H_d}\mathbb{E}\left[\left|\phi_1^j(\mathcal{L})-\phi_1^j(\mathcal{R})\right|\left(\mathbf{1}_{\hat{X}_j\in F_d}+\mathbf{1}_{\hat{X}_j\in F_d^c}\right)\middle| X\right]\\
            &\qquad+\sup_{X\in H_d}\mathbb{P}\left(\exists z\in\left\{\hat{z}_1,\dots,\hat{z}_N,z_1^{\ast},\dots,z_{N-1}^{\ast}\right\}\text{ satisfying }z_{d+1}>1-\epsilon\right)\\
            &\leq\sup_{X\in H_d}\mathbb{E}\left[\left|\phi_1^j(\mathcal{L})-\phi_1^j(\mathcal{R})\right|\mathbf{1}_{\hat{X}_j\in F_d}\middle| X\right]+\sup_{X\in H_d}\mathbb{E}\left[\mathbf{1}_{\hat{X}_j\in F_d^c}\middle |X\right]+(2N-1)\cdot o(d^{-1/2})\\
            &\leq L\sup_{X\in F_d}\mathbb{E}\left[\|\mathcal{L}-\mathcal{R}\|\mathbf{1}_{\hat{X}_j\in F_d}\middle| X\right]+ o\left(d^{-1/2}\right)\\
            &\leq L\sup_{X\in F_d}\mathbb{E}\left[\sum_{i=1}^{2N-1}\left|\mathcal{L}_i-\mathcal{R}_i\right|\mathbf{1}_{\hat{X}_j\in F_d}\middle| X\right]+ o\left(d^{-1/2}\right)\\
            &\leq L\sum_{i=1}^{2N-1}\sup_{X\in F_d}\mathbb{E}\left[|\mathcal{L}_i-\mathcal{R}_i|\mathbf{1}_{\hat{X}_j\in F_d}\middle| X\right]+ o\left(d^{-1/2}\right),
        \end{aligned}
    \end{equation}
    where $L$ stands for the Lipschitz constant.
    For any $i\in\{1,\dots,N\}$, by \eqref{Eq:accept convergence} we can directly get that
    \begin{equation}
        \sup_{X\in F_d}\mathbb{E}[|\mathcal{L}_i-\mathcal{R}_i|\mathbf{1}_{\hat{X}_j\in F_d}\mid X]\leq\sup_{X\in F_d}\mathbb{E}[|\mathcal{L}_i-\mathcal{R}_i|\mid X]=o\left(d^{-1/4}\log(d)\right).
    \end{equation}
    For any $i\in\{N+1,\dots,2N-1\}$, by the construction of SMTM, $z_i^{\ast}$'s are chosen with the same SRWM method but starting at $\hat{z}_j$, which implies
    \begin{equation}
        \begin{aligned}
            \sup_{X\in F_d}\mathbb{E}[|\mathcal{L}_i-\mathcal{R}_i|\mathbf{1}_{\hat{X}_j\in F_d}\mid X]&=\sup_{X\in F_d}\mathbb{E}\left[\mathbb{E}\left[|\mathcal{L}_i-\mathcal{R}_i|\mathbf{1}_{\hat{X}_j\in F_d}\middle| \hat{X}_j\right]\middle| X\right]\\ 
            &\leq\sup_{X\in F_d}\mathbb{E}\left[\sup_{\hat{X}_j\in F_d}\mathbb{E}\left[|\mathcal{L}_i-\mathcal{R}_i|\middle | \hat{X}_j\right]\middle| X\right]\\
            &=o\left(d^{-1/4}\log(d)\right).
        \end{aligned}
    \end{equation}
    Hence
    \begin{equation}
        \sup_{X\in H_d}\mathbb{E}\left[\left|\alpha_1^j-\phi_1^j\left((W_i)_{i=1}^{N},(V_i)_{i=1}^{N-1}\right)\right|\middle| X\right]=o\left(d^{-1/4}\log(d)\right).
    \end{equation}
    
    Obviously $(W_i)_{i=1}^{N}$ are identically distributed given the current state $X$ and $(V_i)_{i=1}^{N-1}$ are identically distributed given $\hat{X}_j$. Now we prove that they are conditionally independent. By the construction of $W_i$'s in the proof of \citet[Lemma 5.1]{yang2024}, for any $i\in\{1,\dots,N\}$, $W_i$ can be written as
    \begin{equation}
        W_i=\mu-\sigma\tilde{U}_i,
    \end{equation}
    where $\mu,\sigma$ are defined in \eqref{Eq:normal distribution}, $\tilde{U}_i\mid X\sim\mathcal{N}(0,1)$ stands for the marginal of the normalized $i$-th proposal distribution to the direction $\tilde{v}$ and
    \begin{equation}
        \tilde{v}:=\left(\left(\log f(X_1)\right)',\dots,\left(\log f(X_d)\right)',\sum_{i=1}^d\left(\frac{1+\lambda}{2\lambda}\left(\log f(X_i)\right)'z_i+\frac{1}{R}\right)\right)
    \end{equation}
    depends only on the current state $X$. Since the proposal distributions are conditionally independent, $\tilde{U}_i$'s are also conditionally independent, and so are $W_i$'s. Similarly, we can also get that $V_i$'s are independent when conditioned on $\hat{X}_j$.
\end{proof}

\subsection{Proof of \cref{Thm:esjd}}
\label{proof:thm:esjd}
\begin{proof}
    By the assumptions on $\pi$, it suffices to consider the first coordinate,
    \begin{equation}
        \ESJD_{\SMTM}=N\mathbb{E}_{X\sim\pi}\left[\mathbb{E}\left[(\hat{X}_1-X)^2\alpha_2^1\middle| X\right]\right]=Nd\mathbb{E}_{X\sim\pi}\left[\mathbb{E}\left[(\hat{X}_{1,1}-X_1)^2\alpha_2^1\middle| X\right]\right].
    \end{equation}
    Denote $X_j^{\ast}=\SP(z_j^{\ast})$ for $j\in\{1,\dots,N-1\}$. Define $F_d$ as in \citet[Proof of Theorem 5.1]{yang2024},
    i.e.,
    \begin{equation}
    \begin{aligned}
        F_d:=&\left\{x\in\mathbb{R}^d:\left|\frac{1}{d-1}\sum_{i=2}^d\left[(\log f(x_i))'\right]^2-\mathbb{E}_f\left[((\log f)')^2\right]\right|<d^{-1/8}\right\}\\
        &\cap\left\{x\in\mathbb{R}^d:\left|\frac{1}{d-1}\sum_{i=2}^d\left[(\log f(x_i))''\right]-\mathbb{E}_f\left[(\log f)''\right]\right|<d^{-1/8}\right\}\\
        &\cap\left\{x\in\mathbb{R}^d:\left|\frac{1}{d-1}\sum_{i=2}^dx_i\left(\log f(x_i)\right)'-\mathbb{E}_f\left[X(\log f)'\right]\right|<d^{-1/8}\right\}\\
        &\cap\left\{x\in\mathbb{R}^d:\left|\frac{1}{d}\sum_{i=1}^dx_i^2-\mathbb{E}_f\left[X^2\right]\right|<d^{-1/6}\log d\right\}\\
        &\cap\left\{x\in\mathbb{R}^d:|x_1|<d^{1/5}\right\}.
    \end{aligned}
    \label{eq:definition_Fd2}
\end{equation}
    and $\bar{F}_d=\{x\in\mathbb{R}^{d}:z_{d+1}\leq 1-\epsilon\}$. Define $H_d\subset F_d$ by replacing every threshold in the definition of $F_d$ with one half of its original value. Then, using the same arguments about \cref{eq:escape_F_d}, the Cauchy--Schwarz inequality in \citet[Proof of Theorem 5.1]{yang2024}, we have
    \begin{equation}
        \frac{\ESJD_{\SMTM}}{Nd}\leq\sup_{X\in H_d}\mathbb{E}\left[\left(\hat{X}_{1,1}-X_1\right)^2\phi_2^1\left(\left(\log\frac{\pi_S(\hat{z}_i)}{\pi_S(z)}\right)_{i=1}^N,\left(\log\frac{\pi_S(z_i^{\ast})}{\pi_S(\hat{z}_1)}\right)_{i=1}^{N-1}\right)\mathbf{1}_{\hat{X}_{1:N},X_{1:N-1}^{\ast}\in\bar{F}_d}\middle| X\right]+o(d^{-1}).
    \end{equation}
    Construct the couplings $(\tilde{X}_j)_{j=1}^N$ and $(\bar{X}_j)_{j=1}^{N-1}$ such that
    \begin{equation}
        \begin{aligned}
            &\tilde{X}_{j,i}=\hat{X}_{j,i}\quad\text{for }i=2,\dots,d\\
            &\tilde{X}_{j,1}\stackrel{d}{=}\hat{X}_{j,1}\quad\text{and}\quad\left(\tilde{X}_{j,1}\perp \!\!\! \perp\hat{X}_{j,1}\mid\hat{X}_{j,2:d}\right),
        \end{aligned}
    \end{equation}
    and
    \begin{equation}
        \begin{aligned}
            &\bar{X}_{j,i}=X_{j,i}^{\ast}\quad\text{for }i=2,\dots,d\\
            &\bar{X}_{j,1}\stackrel{d}{=}X_{j,1}^{\ast}\quad\text{and}\quad\left(\bar{X}_{j,1}\perp \!\!\! \perp X_{j,1}^{\ast}\mid X_{j,2:d}^{\ast}\right).
        \end{aligned}
    \end{equation}
    Write $(W_{\tilde{X}_i\mid X})_{i=1}^N$ and $(W_{\bar{X}_i\mid\hat{X}_1})_{i=1}^{N-1}$ as conditionally i.i.d.\,random variables with distribution $\mathcal{N}(\mu,\sigma^2)$ with
    \begin{equation}
        \mu=\frac{\ell^2}{2}\left(1-\mathbb{E}[((\log f)')^2]\right),\quad\sigma^2=\ell^2\left(\mathbb{E}[((\log f)')^2]-1\right),
    \end{equation}
    then by the arguments in \citet[Proof of Theorem 5.1]{yang2024}, the Lipschitz continuity of $\phi_2^1$ and \eqref{Eq:accept convergence}, we get that
    \begin{equation}
        \begin{aligned}
            &\sup_{X\in H_d}\mathbb{E}\left[\left(\hat{X}_{1,1}-X_1\right)^2\phi_2^1\left(\left(\log\frac{\pi_S(\hat{z}_i)}{\pi_S(z)}\right)_{i=1}^N,\left(\log\frac{\pi_S(z_i^{\ast})}{\pi_S(\hat{z}_1)}\right)_{i=1}^{N-1}\right)\mathbf{1}_{\hat{X}_{1:N},X_{1:N-1}^{\ast}\in\bar{F}_d}\middle| X\right]\\
            \to&\sup_{X\in H_d}\mathbb{E}\left[\left(\hat{X}_{1,1}-X_1\right)^2\phi_2^1\left(\mathcal{A}_i)_{i=1}^N,(\mathcal{B}_i)_{i=1}^{N-1}\right)\mathbf{1}_{\hat{X}_{1:N},X_{1:N-1}^{\ast}\in\bar{F}_d}\middle| X\right]\\
            \to&\sup_{X\in H_d}\mathbb{E}\left[\left(\hat{X}_{1,1}-X_1\right)^2\phi_2^1\left(\mathcal{A}_i)_{i=1}^N,(\mathcal{B}_i)_{i=1}^{N-1}\right)\middle| X\right],
        \end{aligned}
    \end{equation}
    where
    \begin{equation}
        \mathcal{A}_i=\log\frac{f(\hat{X}_{i,1})}{f(X_1)}+\frac{\hat{X}_{i,1}^2-\mathbb{E}[\tilde{X}_{i,1}^2]}{2}+W_{\tilde{X}_i\mid X},\quad \mathcal{B}_i=\log\frac{f(X_{i,1}^{\ast})}{f(\hat{X}_{1,1})}+\frac{(X_{i,1}^{\ast})^2-\mathbb{E}[\bar{X}_{i,1}^2]}{2}+W_{\bar{X}_i\mid\hat{X}_1}.
    \end{equation} 
    Now we construct two conditionally i.i.d.\,series $(\tilde{W}_i)_{i=1}^N$ and $(\bar{W})_{i=1}^{N-1}$ satisfying for any $i$
    \begin{equation}
    \label{Eq:independent coupling}
        \tilde{W}_i\stackrel{d}{=}W_{\tilde{X}_i\mid X},\quad\tilde{W}_i\perp \!\!\! \perp\hat{X}_i\quad\text{and}\quad\sup_{X\in F_d}\mathbb{E}\left[\left(\tilde{W}_i-W_{\tilde{X}_i\mid X}\right)^2\right]=o(1),
    \end{equation}
    \begin{equation}
        \bar{W}_i\stackrel{d}{=}W_{\bar{X}_i\mid \hat{X}_1},\quad\bar{W}_i\perp \!\!\! \perp X_i^{\ast}\quad\text{and}\quad\sup_{X\in F_d}\mathbb{E}\left[\left(\bar{W}_i-W_{\bar{X}_i\mid \hat{X}_1}\right)^2\right]=o(1).
    \end{equation}
    According to \citet[S8.10.1]{yang2024}, for any $i\in\{1,\dots,N\}$ we can write 
    \begin{equation}
        W_{\tilde{X}_i\mid X}=\mathcal{F}_1\left(\sum_{j=1}^dV_{i,j}^2,U_{i,2:d+1}\right),\quad\hat{X}_{i,1}=\mathcal{F}_2\left(\sum_{j=1}^dV_{i,j}^2,U_{i,1},U_{i,d+1}\right),
    \end{equation}
    where $\sum_{j=1}^dV_{i,j}^2$ is a chi-squared random variable with degree of freedom $d$ used for normalizing the proposal on the tangent space to the sphere, see \cite[Lemma 3.1]{yang2024}, the sphere movement for different directions $(U_i)_{i=1}^N$ are conditionally independent since the proposals are conditionally independent and $\mathcal{F}$ represents the source of randomness: $W_{\tilde{X}_i\mid X}=\mathcal{F}_1(\sum_{j=1}^dV_{i,j}^2,U_{i,2:d+1})$ denotes the randomness of $W_{\tilde{X}_i\mid X}$ comes from $\sum_{j=1}^dV_{i,j}^2$ and $U_{i,2:d+1}$, and $\hat{X}_{i,1}=\mathcal{F}_2(\sum_{j=1}^dV_{i,j}^2,U_{i,1},U_{i,d+1})$ denotes the randomness of $\hat{X}_{i,1}$ comes from $\sum_{j=1}^dV_{i,j}^2$, $U_{i,1}$ and $U_{i,d+1}$.
    
    Define 
    \begin{equation}
        W_{\tilde{X}_i\mid X}':=\mathcal{F}_1(d,U_{i,2:d+1}).
    \end{equation}
    For $l=2,\dots,d+1$, $U_{i,l}$ can be decomposed by
    \begin{equation}
        U_{i,l}=c_{i,l}^{\|}U_{i,1}+c_{i,l}^{\perp}U_{i,l}^{\perp},
    \end{equation}
    where $c_{i,l}^{\perp}$ is the component independent with $U_{i,1}$ and $c_{i,l}^{\|}$ is the remaining component. Now we replace the $U_{i,1}$ by an independent copy, i.e.
    \begin{equation}
        U_{i,l}'=c_{i,l}^{\|}\tilde{U}_{i,1}+c_{i,l}^{\perp}U_{i,l}^{\perp},
    \end{equation}
    where $\tilde{U}_{i,1}$ is a standard Gaussian random variable independent with $U_{i,1}$ and $(\tilde{U}_{i,1})_{i=1}^N$ are conditionally independent. Since $(U_i)_{i=1}^N$ are conditionally independent, for a fixed $l$, $(c_{i,l}^{\perp}U_{i,l}^{\perp})_{i=1}^N$ are also conditionally independent and so are $(U_{i,l}')_{i=1}^N$. Then we define
    \begin{equation}
        W_{\tilde{X}_i\mid X}'':=\mathcal{F}_1(d,U_{i,2:d+1}').
    \end{equation}
    For $l=2,\dots,d$, $U_{i,l}'$ can be decomposed similarly with respect to $U_{i,d+1}'$, i.e.
    \begin{equation}
        U_{i,l}'=\tilde{c}_{i,l}^{\|}U_{i,d+1}'+\tilde{c}_{i,l}^{\perp}U_{i,l}'^{\perp},
    \end{equation}
    and define
    \begin{equation}
        \tilde{U}_{i,l}=\tilde{c}_{i,l}^{\|}\tilde{U}_{i,d+1}+\tilde{c}_{i,l}^{\perp}U_{i,l}'^{\perp},
    \end{equation}
    where $\tilde{U}_{i,d+1}$ is an independent copy of $U_{i,d+1}'$ and $(\tilde{U}_{i,d+1})_{i=1}^N$ are conditionally independent for any fixed $l$. Using the same argument we get that $(\tilde{U}_{i,l})_{i=1}^N$ are conditionally independent. Finally we define
    \begin{equation}
        \tilde{W}_i:=\mathcal{F}_1(d,\tilde{U}_{i,2:d+1}).
    \end{equation}
    Then $(\tilde{W}_i)_{i=1}^N$ are conditionally independent and satisfy \eqref{Eq:independent coupling} by the same arguments in \citet[S8.10.2-S8.10.4]{yang2024}. Analogously we can construct $(\bar{W}_i)_{i=1}^{N-1}$.
    
    Write $\tilde{\mathcal{A}}_i=\log\frac{f(\hat{X}_{i,1})}{f(X_1)}+\frac{\hat{X}_{i,1}^2-\mathbb{E}[\tilde{X}_{i,1}^2]}{2}+\tilde{W}_i$ and $\tilde{\mathcal{B}}_i=\log\frac{f(X_{i,1}^{\ast})}{f(\hat{X}_{1,1})}+\frac{(X_{i,1}^{\ast})^2-\mathbb{E}[\bar{X}_{i,1}^2]}{2}+\bar{W}_i$, then
    \begin{equation}
        \begin{aligned}
            &\sup_{X\in H_d}\mathbb{E}\left[\left(\hat{X}_{1,1}-X_1\right)^2\phi_2^1\left((\mathcal{A}_i)_{i=1}^N,(\mathcal{B}_i)_{i=1}^{N-1}\right)\middle| X\right]\\
            \to&\sup_{X\in H_d}\mathbb{E}\left[\left(\hat{X}_{1,1}-X_1\right)^2\phi_2^1\left((\tilde{\mathcal{A}}_i)_{i=1}^N,(\tilde{\mathcal{B}}_i)_{i=1}^{N-1}\right)\middle| X\right]\\
            =&\sup_{X\in H_d}\mathbb{E}\left[\left(\hat{X}_{1,1}-X_1\right)^2\mathbb{E}\left[\phi_2^1\left((\tilde{\mathcal{A}}_i)_{i=1}^N,(\tilde{\mathcal{B}}_i)_{i=1}^{N-1}\right)\middle| X,\hat{X}_{1:N},X^{\ast}_{1:N-1}\right]\middle| X\right].
        \end{aligned}
    \end{equation}
    Thus, we have, as $d\to\infty$
    \begin{equation}
        \begin{aligned}
            \frac{\ESJD_{\SMTM}}{Nd}&\to \sup_{X\in H_d}\mathbb{E}\left[\left(\hat{X}_{1,1}-X_1\right)^2\mathbb{E}\left[\phi_2^1\left((\tilde{\mathcal{A}}_i)_{i=1}^N,(\tilde{\mathcal{B}}_i)_{i=1}^{N-1}\right)\middle| X,\hat{X}_{1:N},X^{\ast}_{1:N-1}\right]\middle| X\right]\\
            &\to \mathbb{E}\left[\left(\hat{X}_{1,1}-X_1\right)^2\right]\mathbb{E}\left[\phi_2^1\left((W_i)_{i=1}^N,(V_i)_{i=1}^{N-1}\right)\right],\\
            \ESJD_{\SMTM}&\to N\ell^2\mathbb{E}\left[\phi_2^1\left((W_i)_{i=1}^N,(V_i)_{i=1}^{N-1}\right)\right].
        \end{aligned}
    \end{equation}
\end{proof}

\subsection{Proof of \cref{Thm:maxesjd}}\label{proof:thm:maxesjd}
From \cref{Eq:approx esjd} we get that
\begin{equation}
    \widetilde{\ESJD}(\ell)=N\ell^2\mathbb{E}\left[\phi_2^1((W_i)_{i=1}^N,(V_i)_{i=1}^{N-1})\middle| X\right],
\end{equation}
where $(W_i)_{i=1}^N$ and $(V_i)_{i=1}^{N-1}$ are defined in \cref{Thm:accept rate} and have conditional distribution
\begin{equation}
    \mathcal{N}\left(-\frac{\ell^2}{2}\left(\mathbb{E}\left[((\log f)')^2\right]-\frac{4\lambda}{(1+\lambda)^2}\right),\ell^2\left(\mathbb{E}\left[((\log f)')^2\right]-\frac{4\lambda}{(1+\lambda)^2}\right)\right).
\end{equation}
Also, by \citet[Theorem 2]{bedard2012scaling} we can get the limit of the approximated ESJD of MTM:
\begin{equation}
    \lim_{d\to\infty}\ESJD_{\MTM}(\ell)=N\ell^2\mathbb{E}\left[\phi_2^1((\bar{W}_i)_{i=1}^N,(\bar{V}_i)_{i=1}^{N-1})\middle| X\right],
\end{equation}
where $(\bar{W}_i)_{i=1}^N$ and $(\bar{V}_i)_{i=1}^{N-1}$ are conditionally i.i.d.\,random variables with distribution
\begin{equation}
    \mathcal{N}\left(-\frac{\ell^2}{2}\mathbb{E}\left[((\log f)')^2\right],\ell^2\mathbb{E}\left[((\log f)')^2\right]\right).
\end{equation}
Rescaling the step size of SMTM by 
\begin{equation}
    \ell'^2\mathbb{E}\left[((\log f)')^2\right]=\ell^2\left(\mathbb{E}\left[((\log f)')^2\right]-\frac{4\lambda}{(1+\lambda)^2}\right),
\end{equation}
we have
\begin{equation}
    \widetilde{\ESJD}(\ell)=N\ell'^2\frac{\mathbb{E}\left[((\log f)')^2\right]}{\mathbb{E}\left[((\log f)')^2\right]-\frac{4\lambda}{(1+\lambda)^2}}\mathbb{E}\left[\phi_2^1((\tilde{W}_i)_{i=1}^N,(\tilde{V}_i)_{i=1}^{N-1})\middle| X\right],
\end{equation}
where $(\tilde{W}_i)_{i=1}^N$ and $(\tilde{V}_i)_{i=1}^{N-1}$ are conditionally i.i.d.\,random variables with the same distribution as $(\bar{W}_i)_{i=1}^N$ and $(\bar{V}_i)_{i=1}^{N-1}$. 

Then, by maximizing the approximate ESJD of SMTM, we get
\begin{equation}
    \begin{aligned}
        \max_{\ell}\widetilde{\ESJD}(\ell)&=N\frac{\mathbb{E}\left[((\log f)')^2\right]}{\mathbb{E}\left[((\log f)')^2\right]-\frac{4\lambda}{(1+\lambda)^2}}\max_{\ell'}\ell'^2\mathbb{E}\left[\phi_2^1((\tilde{W}_i)_{i=1}^N,(\tilde{V}_i)_{i=1}^{N-1})\middle| X\right]\\
        &=N\frac{\mathbb{E}\left[((\log f)')^2\right]}{\mathbb{E}\left[((\log f)')^2\right]-\frac{4\lambda}{(1+\lambda)^2}}\max_{\ell}\ell^2\mathbb{E}\left[\phi_2^1((\bar{W}_i)_{i=1}^N,(\bar{V}_i)_{i=1}^{N-1})\middle| X\right]\\
        &=\frac{\mathbb{E}\left[((\log f)')^2\right]}{\mathbb{E}\left[((\log f)')^2\right]-\frac{4\lambda}{(1+\lambda)^2}}\max_{\ell}\lim_{d\to\infty}\ESJD_{\MTM}.
    \end{aligned}
\end{equation}
Therefore,
\begin{equation}
        \frac{\max_\ell\widetilde{\ESJD}}{\max_\ell\lim_{d\to\infty}\ESJD_{\MTM}}=\frac{\mathbb{E}\left[((\log f)')^2\right]}{\mathbb{E}\left[((\log f)')^2\right]-\frac{4\lambda}{(1+\lambda)^2}}=\frac{1}{1-\alpha\cdot\beta\cdot\gamma}.
    \end{equation}

\section{Proofs in \cref{Sec:n to infty}}
\subsection{Proof of \cref{prop:accept rate converge}}\label{proof:prop:accept rate converge}
\begin{proof}
    Firstly, in the globally-balanced case, since $(W_i)_{i=1}^{N}$ and $(V_i)_{i=1}^{N-1}$ are conditionally i.i.d.\,we have
    \begin{equation}
        \begin{aligned}
            \sum_{j=1}^{N}\mathbb{E}\left[\phi_2^j\left((W_i)_{i=1}^N,(V_i)_{i=1}^{N-1}\right)\middle | X\right]&=\sum_{j=1}^{N}\mathbb{E}\left[\frac{e^{W_j}}{\sum_{i=1}^Ne^{W_i}}\wedge\frac{e^{W_j}}{e^{W_j}\sum_{i=1}^{N-1}e^{V_i}+1}\middle | X\right]\\
            &=N\mathbb{E}\left[\frac{e^{W_1}}{e^{W_1}+\sum_{i=2}^Ne^{W_i}}\wedge\frac{e^{W_1}}{e^{W_1}\sum_{i=1}^{N-1}e^{V_i}+1}\middle | X\right]\\
            &=\mathbb{E}\left[\frac{e^{W_1}}{\frac{1}{N}e^{W_1}+\frac{1}{N}\sum_{i=2}^Ne^{W_i}}\wedge\frac{e^{W_1}}{e^{W_1}\frac{1}{N}\sum_{i=1}^{N-1}e^{V_i}+\frac{1}{N}}\middle | X\right].
        \end{aligned}
    \end{equation} 
    The variance of a log-normal random variable with $\mu=-\frac{\sigma^2}{2}$ is $e^{\sigma^2}-1$, so by the strong law of large numbers we have
    \begin{equation}
        \frac{1}{N-1}\sum_{i=2}^Ne^{W_i}\xrightarrow{a.s.} 1,\quad\frac{1}{N-1}\sum_{i=1}^{N-1}e^{V_i}\xrightarrow{a.s.} 1.
    \end{equation}
    Therefore, by the dominated convergence theorem, we have
    \begin{equation}
        \begin{aligned}
            \sum_{j=1}^{N}\mathbb{E}\left[\phi_2^j\left((W_i)_{i=1}^N,(V_i)_{i=1}^{N-1}\right)\middle | X\right]&=\mathbb{E}\left[\mathbb{E}\left[\frac{e^{W_1}}{\frac{1}{N}e^{W_1}+\frac{1}{N}\sum_{i=2}^Ne^{W_i}}\wedge\frac{e^{W_1}}{e^{W_1}\frac{1}{N}\sum_{i=1}^{N-1}e^{V_i}+\frac{1}{N}}\middle | W_1,X\right]\middle | X\right]\\
            &\to\mathbb{E}\left[\frac{e^{W_1}}{\frac{1}{N}e^{W_1}+\frac{N-1}{N}}\wedge\frac{e^{W_1}}{\frac{1}{N}+\frac{N-1}{N}e^{W_1}}\middle | X\right]\\
            &\to\mathbb{E}\left[e^{W_1}\wedge1\middle | X\right].
        \end{aligned}
    \end{equation}
    From \citet[Corollary 5.2]{yang2024}, we know that $\mathbb{E}\left[e^{W_1}\wedge1\middle | X\right]$ is approximately 0.234 using the optimal $\ell$ that maximizes the ESJD.
    
    Observing that the minimum between $\frac{e^{W_1}}{\frac{1}{N}e^{W_1}+\frac{N-1}{N}}$ and $\frac{e^{W_1}}{\frac{1}{N}+\frac{N-1}{N}e^{W_1}}$ always changes when $e^{W_1}\geq 1$ versus $e^{W_1}<1$ if $N\geq 2$, we can rewrite $\frac{e^{W_1}}{\frac{1}{N}e^{W_1}+\frac{N-1}{N}}\wedge\frac{e^{W_1}}{\frac{1}{N}+\frac{N-1}{N}e^{W_1}}$ as 
    \begin{equation}
        \begin{aligned}
            \frac{e^{W_1}}{\frac{1}{N}e^{W_1}+\frac{N-1}{N}}\wedge\frac{e^{W_1}}{\frac{1}{N}+\frac{N-1}{N}e^{W_1}}&=\frac{N}{1+(N-1)e^{-W_1}}\wedge\frac{N}{(N-1)+e^{-W_1}}\\
            &=\frac{1+x}{x+e^{-W_1}}\wedge\frac{1+x}{1+xe^{-W_1}}\\
            &=\frac{1+x}{1+x e^{-W_1}}\mathbf{1}_{e^{-W_1}<1}+\frac{1+x}{e^{-W_1}+x}\mathbf{1}_{e^{-W_1}\geq 1}:=A(x),
        \end{aligned}
    \end{equation}
    where we changed variable $N$ to $x=\frac{1}{N-1}\in (0,1]$.

    Next, we study the monotonicity of $A(x)$. Conditional on $W_1$, we switch the expectation and derivative, then (a) when $e^{-W_1}<1$ the first term increases from $1\to \frac{2}{1+e^{-W_1}}>1$ when $x$ from $0\to1$; (b) similarly the second term increases from $\frac{1}{e^{-W_1}}$ to $\frac{2}{1+e^{-W_1}}$ since $e^{-W_1}\geq1$ when $x$ from $0\to1$. Therefore, for any given $W_1$, the function $A(x)$ is increasing w.r.t.\,$x$, which implies that it is decreasing w.r.t.\,$N$. Overall when $N\to\infty$, we have
    \[
    \mathbb{E}\left[\frac{1}{\frac{1}{N}+\frac{N-1}{N}e^{-W_1}}\wedge \frac{1}{\frac{N-1}{N}+\frac{1}{N}e^{-W_1}}\middle|X\right]\downarrow \mathbb{E}[e^{W_1}\wedge 1\mid X].
    \]
    This shows that the optimal acceptance rate by optimizing $\ell$ in the $N\to\infty$ limit is still roughly $0.234$ for the globally-balanced weight function. 
    
    In the locally-balanced case, similarly, we have
    \begin{equation}
        \phi_2^j((W_i)_{i=1}^N,(V_i)_{i=1}^{N-1})=\frac{e^{W_j}}{\sum_{i=1}^N\sqrt{e^{W_i}e^{W_j}}}\wedge\frac{e^{W_j}}{\sum_{i=1}^{N-1}\sqrt{e^{V_i}e^{W_j}}+1},
    \end{equation}
    and since the variances are finite, by the strong law of large numbers
    \begin{equation}
        \frac{1}{N-1}\sum_{i=2}^{N}\sqrt{e^{W_i-W}}\xrightarrow{a.s.} e^{-\frac{W}{2}-\frac{\sigma^2}{8}},\quad\frac{1}{N-1}\sum_{i=1}^{N-1}\sqrt{e^{V_i-W}}\xrightarrow{a.s.} e^{-\frac{W}{2}-\frac{\sigma^2}{8}}.
    \end{equation}
    Then by the dominated convergence theorem,
    \begin{equation}
        \begin{aligned}
            &\sum_{j=1}^N\mathbb{E}\left[\frac{e^{W_j}}{\sum_{i=1}^N\sqrt{e^{W_i}e^{W_j}}}\wedge\frac{e^{W_j}}{\sum_{i=1}^{N-1}\sqrt{e^{V_i}e^{W_j}}+1}\middle | X\right]\\
            &=\mathbb{E}\left[\frac{1}{\frac{1}{N}\sum_{i=2}^{N}\sqrt{e^{W_i-W_1}}+\frac{1}{N}}\wedge\frac{1}{\frac{1}{N}\sum_{i=1}^{N-1}\sqrt{e^{V_i-W_1}}+\frac{1}{N}e^{-W_1}}\middle | X\right]\\
            &\to\mathbb{E}\left[\frac{1}{\frac{N-1}{N}e^{-\frac{W_1}{2}-\frac{\sigma^2}{8}}+\frac{1}{N}}\wedge\frac{1}{\frac{N-1}{N}e^{-\frac{W_1}{2}-\frac{\sigma^2}{8}}+\frac{1}{N}e^{-W_1}}\middle | X\right]\\
            &\to\mathbb{E}\left[e^{\frac{W_1}{2}+\frac{\sigma^2}{8}}\middle | X\right]=1.
        \end{aligned}
    \end{equation}
    Let $x=\frac{1}{N-1}$. We can rewrite $\frac{1}{\frac{N-1}{N}e^{-\frac{W_1}{2}-\frac{\sigma^2}{8}}+\frac{1}{N}}\wedge\frac{1}{\frac{N-1}{N}e^{-\frac{W_1}{2}-\frac{\sigma^2}{8}}+\frac{1}{N}e^{-W_1}}$ as
    \begin{equation}
        B(x):=\frac{1+x}{e^{-\frac{W_1}{2}-\frac{\sigma^2}{8}}+x}\mathbf{1}_{e^{-W_1}<1}+\frac{1+x}{e^{-\frac{W_1}{2}-\frac{\sigma^2}{8}}+xe^{-W_1}}\mathbf{1}_{e^{-W_1}\geq1}.
    \end{equation}
    Using similar arguments we have: (a) when $e^{-W_1}<1$ the first term decreases from $\frac{1}{e^{-\frac{W_1}{2}-\frac{\sigma^2}{8}}}\to \frac{2}{1+e^{-\frac{W_1}{2}-\frac{\sigma^2}{8}}}>1$ since $e^{-\frac{W_1}{2}-\frac{\sigma^2}{8}}<1$ when $x$ from $0\to1$; (b) when $e^{-W_1}\geq 1$, the second term decreases from $\frac{1}{e^{-\frac{W_1}{2}-\frac{\sigma^2}{8}}}$ to $\frac{2}{e^{-\frac{W_1}{2}-\frac{\sigma^2}{8}}+e^{-W_1}}$, since $W_1\leq0$ and $e^{-W_1}>e^{-\frac{W_1}{2}-\frac{\sigma^2}{8}}$ when $x$ from $0\to1$. Therefore, when $N\to\infty$, we have
    \begin{equation}
        \mathbb{E}\left[\frac{1}{\frac{N-1}{N}e^{-\frac{W_1}{2}-\frac{\sigma^2}{8}}+\frac{1}{N}}\wedge\frac{1}{\frac{N-1}{N}e^{-\frac{W_1}{2}-\frac{\sigma^2}{8}}+\frac{1}{N}e^{-W_1}}\middle | X\right]\uparrow 1.
    \end{equation}
\end{proof}

\subsection{Proof of \cref{thm:infinity weak convergence}}\label{proof:thm:infinity weak convergence}
We will use the following result from \cite{Gagnon2023}:
\begin{theorem}\cite[Theorem 4]{Gagnon2023}
    Let $\{P_N:N\geq1\}$ be a collection of Markov transition kernels and $P$ a Markov transition kernel such that $P_N$ (for any $N$) and $P$ admit $\pi$ as invariant distribution. If the following assumptions hold:
    \begin{itemize}
        \item the Markov transition kernels $\{P_N:N\geq1\}$ satisfy:
        \begin{equation}
            \int|P_Nh(x)-Ph(x)|\pi(x)\dee x\to0
        \end{equation}
        as $N\to\infty$ for all $h\in\mathcal{C}_b$, where
        \begin{equation}
            P_Nh(x):=\int P_N(x,\dee y)h(y),
        \end{equation}
        with an analogous definition for $Ph$, $\mathcal{C}_b$ being the space of bounded continuous functions;
        \item $Ph$ is continuous for any $h\in\mathcal{C}_b$,
    \end{itemize}
    then $\{X_N(m):m\in\mathbb{N}\}$ converges weakly to $\{X(m):m\in\mathbb{N}\}$ provided that $X_N(0)\sim\pi$ and $X(0)\sim\pi$, $\{X_N(m):m\in\mathbb{N}\}$ and $\{X(m):m\in\mathbb{N}\}$ being the Markov chains with transition kernels $P_N$ and $P$, respectively.
\end{theorem}
\begin{proof}[Proof of \cref{thm:infinity weak convergence}]
    Since the stereographic projection is continuous and weak convergence is preserved under continuous functions, it is equivalent to showing that $\{Z_N(m)\}$ converges weakly to $\{Z_{\infty}(m)\}$ provided that $Z_N(0)\sim\pi_S$ and $Z_{\infty}(0)\sim\pi_S$, where $\{Z_N(m)\}$ and $\{Z_{\infty}(m)\}$ are SMTM and the ideal scheme on the sphere, respectively. The Markov kernels of $Z_N$ and $Z_{\infty}$ can be respectively written as
    \begin{equation}
        \begin{aligned}
            P_N(z,A)=&\sum_{j=1}^N\int_{\hat{z}_j\in A}\frac{\omega(z,\hat{z}_j)}{\sum_{i=1}^N\omega(z,\hat{z}_i)}\prod_{i=1}^NQ_S(z,\hat{z}_i)\prod_{i=1}^{N-1}Q_S(\hat{z}_j,z_i^{\ast})\alpha(z,\hat{z}_j) \dee\hat{z}_{1:N} \dee z_{1:N-1}^{\ast}\\
            +&\mathbf{1}_{z\in A}\sum_{j=1}^N\int\frac{\omega(z,\hat{z}_j)}{\sum_{i=1}^N\omega(z,\hat{z}_i)}\prod_{i=1}^NQ_S(z,\hat{z}_i)\prod_{i=1}^{N-1}Q_S(\hat{z}_j,z_i^{\ast})\left(1-\alpha(z,\hat{z}_j)\right) \dee\hat{z}_{1:N} \dee z_{1:N-1}^{\ast},
        \end{aligned}
    \end{equation}
    \begin{equation}
        P_{\infty}(z,A)=\int_{\hat{z}\in A}Q_{\infty}(z,\hat{z})\alpha_{\infty}(z,\hat{z}) \dee\hat{z}+\mathbf{1}_{z\in A}\int Q_{\infty}(z,\hat{z})\left(1-\alpha_{\infty}(z,\hat{z})\right) \dee\hat{z},
    \end{equation}
    where 
    \begin{equation}
        \alpha_\infty(z,\hat{z})=1\wedge\frac{\pi_S(\hat{z})Q_\infty(\hat{z},z)}{\pi_S(z)Q_\infty(z,\hat{z})}.
    \end{equation}
    Since $\hat{z}_1,\dots,\hat{z}_N$ are conditionally i.i.d., the integrals in $P_N$ are the same, so we can further write it as
    \begin{equation}
        \begin{aligned}
            P_N(z,A)=&\int_{\hat{z}_1\in A}\frac{\omega(z,\hat{z}_1)}{\frac{1}{N}\sum_{i=1}^N\omega(z,\hat{z}_i)}\prod_{i=1}^NQ_S(z,\hat{z}_i)\prod_{i=1}^{N-1}Q_S(\hat{z}_1,z_i^{\ast})\alpha(z,\hat{z}_1) \dee\hat{z}_{1:N} \dee z_{1:N-1}^{\ast}\\
            +&\mathbf{1}_{z\in A}\int\frac{\omega(z,\hat{z}_1)}{\frac{1}{N}\sum_{i=1}^N\omega(z,\hat{z}_i)}\prod_{i=1}^NQ_S(z,\hat{z}_i)\prod_{i=1}^{N-1}Q_S(\hat{z}_1,z_i^{\ast})\left(1-\alpha(z,\hat{z}_1)\right) \dee\hat{z}_{1:N} \dee z_{1:N-1}^{\ast}.
        \end{aligned}
    \end{equation}
    By \citet[Theorem 4]{Gagnon2023}, it suffices to show that
    \begin{enumerate}
        \item For any bounded continuous function $h:\mathbb{S}^d\to\mathbb{R}$, as $N\to\infty$,
        \begin{equation}
        \label{Eq:weak convergence}
            \int\left|P_Nh(z)-P_\infty h(z)\right|\pi_S(z)\dee z\to0,
        \end{equation}
        where $P_Nh(z):=\int P_N(z,\hat{z})h(\hat{z}) \dee\hat{z}$ and $P_\infty h(z):=\int P_\infty(z,\hat{z})h(\hat{z}) \dee\hat{z}$.
        \item $P_\infty h$ is continuous for any bounded continuous function $h$.
    \end{enumerate}
    
    We first prove that $\int|P_Nh(z)-P_\infty h(z)|\pi_S(z)\dee z\to0$. By the definition of $P_Nh$ and $P_\infty h$ we can write that
    \begin{equation}
        \begin{aligned}
            P_Nh(z)=&\int\frac{\omega(z,\hat{z}_1)}{\frac{1}{N}\sum_{i=1}^N\omega(z,\hat{z}_i)}\prod_{i=1}^NQ_S(z,\hat{z}_i)\prod_{i=1}^{N-1}Q_S(\hat{z}_1,z_i^{\ast})\alpha(z,\hat{z}_1)h(\hat{z}_1) \dee\hat{z}_{1:N} \dee z_{1:N-1}^{\ast}\\
            +&h(z)\int\frac{\omega(z,\hat{z}_1)}{\frac{1}{N}\sum_{i=1}^N\omega(z,\hat{z}_i)}\prod_{i=1}^NQ_S(z,\hat{z}_i)\prod_{i=1}^{N-1}Q_S(\hat{z}_1,z_i^{\ast})\left(1-\alpha(z,\hat{z}_1)\right) \dee\hat{z}_{1:N} \dee z_{1:N-1}^{\ast},
        \end{aligned}
    \end{equation}
    \begin{equation}
        \begin{aligned}
            P_{\infty}h(z)=\int Q_{\infty}(z,\hat{z})\alpha_{\infty}(z,\hat{z})h(\hat{z}) \dee\hat{z}+h(z)\int Q_{\infty}(z,\hat{z})\left(1-\alpha_{\infty}(z,\hat{z})\right) \dee\hat{z}.
        \end{aligned}
    \end{equation}
    By the triangle inequality, \eqref{Eq:weak convergence} can be achieved as long as we prove that
    \begin{equation}
        \begin{aligned}
            \int\bigg|&\int\frac{\omega(z,\hat{z}_1)}{\frac{1}{N}\sum_{i=1}^N\omega(z,\hat{z}_i)}\prod_{i=1}^NQ_S(z,\hat{z}_i)\prod_{i=1}^{N-1}Q_S(\hat{z}_1,z_i^{\ast})\alpha(z,\hat{z}_1)h(\hat{z}_1) \dee\hat{z}_{1:N} \dee z_{1:N-1}^{\ast}\\
            &-\int Q_{\infty}(z,\hat{z})\alpha_{\infty}(z,\hat{z})h(\hat{z}) \dee\hat{z}\bigg|\pi_S(z)\dee z\to 0,
        \end{aligned}
        \label{eq:convergence_1}
    \end{equation}
    and
    \begin{equation}
        \begin{aligned}
            \int\bigg|&h(z)\int\frac{\omega(z,\hat{z}_1)}{\frac{1}{N}\sum_{i=1}^N\omega(z,\hat{z}_i)}\prod_{i=1}^NQ_S(z,\hat{z}_i)\prod_{i=1}^{N-1}Q_S(\hat{z}_1,z_i^{\ast})\left(1-\alpha(z,\hat{z}_1)\right) \dee\hat{z}_{1:N} \dee z_{1:N-1}^{\ast}\\
            &-h(z)\int Q_{\infty}(z,\hat{z})\left(1-\alpha_{\infty}(z,\hat{z})\right) \dee\hat{z}\bigg|\pi_S(z)\dee z\to 0.
        \end{aligned}
        \label{eq:convergence_2}
    \end{equation}
    For the convergence in \cref{eq:convergence_1}, since
    \begin{equation}
        \begin{aligned}
            &\int Q_{\infty}(z,\hat{z})\alpha_{\infty}(z,\hat{z})h(\hat{z}) \dee\hat{z}=\int\frac{\omega(z,\hat{z}_1)Q_S(z,\hat{z}_1)}{\int\omega(z,\hat{z}_1)Q_S(z,\hat{z}_1) \dee\hat{z}_1}\alpha_\infty(z,\hat{z}_1)h(\hat{z}_1) \dee\hat{z}_1\\
            &=\int\frac{\omega(z,\hat{z}_1)}{\int\omega(z,\hat{z}_1)Q_S(z,\hat{z}_1) \dee\hat{z}_1}\alpha_\infty(z,\hat{z}_1)h(\hat{z}_1)\prod_{i=1}^NQ_S(z,\hat{z}_i)\prod_{i=1}^{N-1}Q_S(\hat{z}_1,z_i^{\ast}) \dee\hat{z}_{1:N} \dee z_{1:N-1}^{\ast},
        \end{aligned}
    \end{equation}
    one can get that
    \begin{equation}
        \begin{aligned}
            \int\bigg|&\int\frac{\omega(z,\hat{z}_1)}{\frac{1}{N}\sum_{i=1}^N\omega(z,\hat{z}_i)}\prod_{i=1}^NQ_S(z,\hat{z}_i)\prod_{i=1}^{N-1}Q_S(\hat{z}_1,z_i^{\ast})\alpha(z,\hat{z}_1)h(\hat{z}_1) \dee\hat{z}_{1:N} \dee z_{1:N-1}^{\ast}\\
            &-\int Q_{\infty}(z,\hat{z})\alpha_{\infty}(z,\hat{z})h(\hat{z}) \dee\hat{z}\bigg|\pi_S(z) \dee z\\
            =&\int\int\bigg|\frac{\omega(z,\hat{z}_1)}{\frac{1}{N}\sum_{i=1}^N\omega(z,\hat{z}_i)}\alpha(z,\hat{z}_1)-\frac{\omega(z,\hat{z}_1)}{\int\omega(z,\hat{z}_1)Q_S(z,\hat{z}_1) \dee\hat{z}_1}\alpha_\infty(z,\hat{z}_1)\bigg|\\
            &|h(\hat{z}_1)|\prod_{i=1}^NQ_S(z,\hat{z}_i)\prod_{i=1}^{N-1}Q_S(\hat{z}_1,z_i^{\ast}) \dee\hat{z}_{1:N} \dee z_{1:N-1}^{\ast}\pi_S(z) \dee z\\
            \leq&M\int\int\bigg|\frac{\omega(z,\hat{z}_1)}{\frac{1}{N}\sum_{i=1}^N\omega(z,\hat{z}_i)}\alpha(z,\hat{z}_1)-\frac{\omega(z,\hat{z}_1)}{\int\omega(z,\hat{z}_1)Q_S(z,\hat{z}_1) \dee\hat{z}_1}\alpha_\infty(z,\hat{z}_1)\bigg|\\
            &\prod_{i=1}^NQ_S(z,\hat{z}_i)\prod_{i=1}^{N-1}Q_S(\hat{z}_1,z_i^{\ast}) \dee\hat{z}_{1:N} \dee z_{1:N-1}^{\ast}\pi_S(z) \dee z\\
            =&M\mathbb{E}\left[\left|\frac{\omega(Z,\hat{Z}_1)}{\frac{1}{N}\sum_{i=1}^N\omega(Z,\hat{Z}_i)}\alpha(Z,\hat{Z}_1)-\frac{\omega(Z,\hat{Z}_1)}{\int\omega(Z,\hat{z}_1)Q_S(Z,\hat{z}_1) \dee\hat{z}_1}\alpha_\infty(Z,\hat{Z}_1)\right|\right],
        \end{aligned}
        \label{eq:expectation_bound}
    \end{equation}
    where $M$ is the upper bound of $|h(z)|$. By the strong law of large numbers, with probability 1,
    \begin{equation}
        \left|\frac{\omega(Z,\hat{Z}_1)}{\frac{1}{N}\sum_{i=1}^N\omega(Z,\hat{Z}_i)}\alpha(Z,\hat{Z}_1)-\frac{\omega(Z,\hat{Z}_1)}{\int\omega(Z,\hat{z}_1)Q_S(Z,\hat{z}_1) \dee\hat{z}_1}\alpha_\infty(Z,\hat{Z}_1)\right|\to0.
    \end{equation}
    To get the convergence of the expectation, we also need to show that the random variable is uniformly integrable. By the Cauchy--Schwarz inequality and \citet[Proposition 4]{Gagnon2023}, we have that
    \begin{equation}
        \begin{aligned}
            &\mathbb{E}\left[\left(\frac{\omega(Z,\hat{Z}_1)}{\frac{1}{N}\sum_{i=1}^N\omega(Z,\hat{Z}_i)}\alpha(Z,\hat{Z}_1)-\frac{\omega(Z,\hat{Z}_1)}{\int\omega(Z,\hat{z}_1)Q_S(Z,\hat{z}_1) \dee\hat{z}_1}\alpha_\infty(Z,\hat{Z}_1)\right)^2\right]\\
            \leq&2\mathbb{E}\left[\left(\frac{\omega(Z,\hat{Z}_1)}{\frac{1}{N}\sum_{i=1}^N\omega(Z,\hat{Z}_i)}\alpha(Z,\hat{Z}_1)\right)^2\right]+2\mathbb{E}\left[\left(\frac{\omega(Z,\hat{Z}_1)}{\int\omega(Z,\hat{z}_1)Q_S(Z,\hat{z}_1) \dee\hat{z}_1}\alpha_\infty(Z,\hat{Z}_1)\right)^2\right]\\
            \leq&2\mathbb{E}\left[\left(\frac{\omega(Z,\hat{Z}_1)}{\frac{1}{N}\sum_{i=1}^N\omega(Z,\hat{Z}_i)}\right)^2\right]+2\mathbb{E}\left[\left(\frac{\omega(Z,\hat{Z}_1)}{\int\omega(Z,\hat{z}_1)Q_S(Z,\hat{z}_1) \dee\hat{z}_1}\right)^2\right]\\
            \leq&2\mathbb{E}\left[\omega(Z,\hat{Z}_1)^4\right]^{1/2}\mathbb{E}\left[\omega(Z,\hat{Z}_1)^{-4}\right]^{1/2}+2\mathbb{E}\left[\omega(Z,\hat{Z}_1)^2\right]\left(\int\omega(Z,\hat{z}_1)Q_S(Z,\hat{z}_1) \dee z_1\right)^{-2},
        \end{aligned}
    \end{equation}
    which is finite by assumption. Therefore, the last expectation term in \cref{eq:expectation_bound} converges to $0$, thereby establishing the convergence in \cref{eq:convergence_1}. By applying the same arguments, we can also demonstrate the convergence in \cref{eq:convergence_2}. Thus, we have proven \cref{Eq:weak convergence}.

    Now we show that $P_\infty h$ is continuous. Let $\textrm{dist}(\cdot,\cdot)$ be the geodesic distance on sphere $\mathbb{S}^d$. Take $z_\epsilon$ satisfying $\textrm{dist}(z,z_\epsilon)\leq\epsilon$ for any $\epsilon>0$. Then, using the explicit formula of $Q_S$ \citep[Lemma A.14]{milinanni2026rapid}, we have in the globally-balanced case 
    \begin{equation}
        \begin{aligned}
            Q_\infty(z_\epsilon,\hat{z})&=\frac{\pi_S(\hat{z})Q_S(z_\epsilon,\hat{z})}{\int \pi_S(\hat{z})Q_S(z_\epsilon,\hat{z}) \dee\hat{z}}\\
            &=\frac{\pi_S(\hat{z})\left(\frac{1}{\cos(\textrm{dist}(z_\epsilon,\hat{z}))}\right)^{d+1}\exp\left(\frac{1}{2h^2}-\frac{1}{2h^2\cos^2(\textrm{dist}(z_\epsilon,\hat{z}))}\right)}{\int \pi_S(\hat{z})\left(\frac{1}{\cos(\textrm{dist}(z_\epsilon,\hat{z}))}\right)^{d+1}\exp\left(\frac{1}{2h^2}-\frac{1}{2h^2\cos^2(\textrm{dist}(z_\epsilon,\hat{z}))}\right) \dee\hat{z}}.
        \end{aligned}
    \end{equation}
    By Fatou's lemma, we have
    \begin{equation}
        \begin{aligned}
            \liminf_{\epsilon\to 0}&\int \pi_S(\hat{z})\left(\frac{1}{\cos(\textrm{dist}(z_\epsilon,\hat{z}))}\right)^{d+1}\exp\left(\frac{1}{2h^2}-\frac{1}{2h^2\cos^2(\textrm{dist}(z_\epsilon,\hat{z}))}\right) \dee\hat{z}\\
            &\geq \int \pi_S(\hat{z})\left(\frac{1}{\cos(\textrm{dist}(z,\hat{z}))}\right)^{d+1}\exp\left(\frac{1}{2h^2}-\frac{1}{2h^2\cos^2(\textrm{dist}(z,\hat{z}))}\right) \dee\hat{z}.
        \end{aligned}
    \end{equation}
    Then, for any $\delta>0$, we can choose $\epsilon$ such that 
    \begin{equation}
        \begin{aligned}
            &\int \pi_S(\hat{z})\left(\frac{1}{\cos(\textrm{dist}(z_\epsilon,\hat{z}))}\right)^{d+1}\exp\left(\frac{1}{2h^2}-\frac{1}{2h^2\cos^2(\textrm{dist}(z_\epsilon,\hat{z}))}\right) \dee\hat{z}\\
            &\geq \int \pi_S(\hat{z})\left(\frac{1}{\cos(\textrm{dist}(z,\hat{z}))}\right)^{d+1}\exp\left(\frac{1}{2h^2}-\frac{1}{2h^2\cos^2(\textrm{dist}(z,\hat{z}))}\right) \dee\hat{z}-\delta.
        \end{aligned}
    \end{equation}
    Now let 
    \begin{equation}
        \delta\leq\frac{1}{2}\int \pi_S(\hat{z})\left(\frac{1}{\cos(\textrm{dist}(z,\hat{z}))}\right)^{d+1}\exp\left(\frac{1}{2h^2}-\frac{1}{2h^2\cos^2(\textrm{dist}(z,\hat{z}))}\right) \dee\hat{z},
    \end{equation}
    we get
    \begin{equation}
        \begin{aligned}
            &\int \pi_S(\hat{z})\left(\frac{1}{\cos(\textrm{dist}(z_\epsilon,\hat{z}))}\right)^{d+1}\exp\left(\frac{1}{2h^2}-\frac{1}{2h^2\cos^2(\textrm{dist}(z_\epsilon,\hat{z}))}\right) \dee\hat{z}\\
            &\geq \frac{1}{2}\int \pi_S(\hat{z})\left(\frac{1}{\cos(\textrm{dist}(z,\hat{z}))}\right)^{d+1}\exp\left(\frac{1}{2h^2}-\frac{1}{2h^2\cos^2(\textrm{dist}(z,\hat{z}))}\right) \dee\hat{z}.
        \end{aligned}
    \end{equation}
    Since the function $f(x)=x^{d+1}\exp(a-ax^2)$ is bounded, there exists a constant $\bar{M}>0$ such that 
    \begin{equation}
        \left(\frac{1}{\cos(\textrm{dist}(z_\epsilon,\hat{z}))}\right)^{d+1}\exp\left(\frac{1}{2h^2}-\frac{1}{2h^2\cos^2(\textrm{dist}(z_\epsilon,\hat{z}))}\right)\leq \bar{M}.
    \end{equation}
    Therefore 
    \begin{equation}
        Q_\infty(z_\epsilon,\hat{z})\leq\frac{\pi_S(\hat{z})\bar{M}}{\frac{1}{2}\int \pi_S(\hat{z})\left(\frac{1}{\cos(\textrm{dist}(z,\hat{z}))}\right)^{d+1}\exp\left(\frac{1}{2h^2}-\frac{1}{2h^2\cos^2(\textrm{dist}(z,\hat{z}))}\right) \dee\hat{z}}:=g(\hat{z}),
    \end{equation}
    where $g(\hat{z})$ is obviously an integrable function with respect to $\hat{z}$. Since both $h$ and $\alpha_\infty$ are bounded, by the dominated convergence theorem and the continuity of $h$, we can prove $P_\infty h$ is continuous by
    \begin{equation}
        \begin{aligned}
            \lim_{\epsilon\to0}P_\infty h(z_\epsilon)&=\lim_{\epsilon\to0}\left(\int Q_{\infty}(z_\epsilon,\hat{z})\alpha_{\infty}(z_\epsilon,\hat{z})h(\hat{z}) \dee\hat{z}+h(z_\epsilon)\int Q_{\infty}(z_\epsilon,\hat{z})\left(1-\alpha_{\infty}(z_\epsilon,\hat{z})\right) \dee\hat{z}\right)\\
            &=\int\lim_{\epsilon\to0}Q_{\infty}(z_\epsilon,\hat{z})\alpha_{\infty}(z_\epsilon,\hat{z})h(\hat{z}) \dee\hat{z}+h(z)\int\lim_{\epsilon\to0}Q_{\infty}(z_\epsilon,\hat{z})\left(1-\alpha_{\infty}(z_\epsilon,\hat{z})\right) \dee\hat{z}\\
            &=\int Q_{\infty}(z,\hat{z})\alpha_{\infty}(z,\hat{z})h(\hat{z}) \dee\hat{z}+h(z)\int Q_{\infty}(z,\hat{z})\left(1-\alpha_{\infty}(z,\hat{z})\right) \dee\hat{z}=P_\infty h(z).
        \end{aligned}
    \end{equation}
    
\end{proof}
\subsection{Proof of \cref{thm: ideal uniform ergodic}}\label{proof:thm: ideal uniform ergodic}
\begin{proof}
    Denote the proposal density of \cref{alg:ideal} on $\mathbb{S}^d$ as $P_\infty(z,\cdot)$. It suffices to prove that $P_\infty(z,\cdot)$ has a lower bound w.r.t.\,$Q_S(z,\cdot)$. We can use the same argument as in the proof of \cref{Thm:uniform ergodic}. Considering the accept/reject step, we have
    \begin{equation}
        P_\infty(z,\hat{z})\geq Q_\infty(z,\hat{z})\alpha_\infty(z,\hat{z}).
    \end{equation}
    In the globally-balanced case, if $\hat{z}\notin\AC(\epsilon)$,
    \begin{equation}
        \begin{aligned}
            Q_\infty(z,\hat{z})&=\frac{\pi_S(\hat{z})Q_S(z,\hat{z})}{\int_{\mathbb{S}^d}\pi_S(y)Q_S(z,y) \dee y}\\
            &\geq\frac{\delta_\epsilon Q_S(z,\hat{z})}{M},
        \end{aligned}
    \end{equation}
    where we used $\pi_S(y)\leq M$. Moreover,
    \begin{equation}
        \begin{aligned}
            \alpha_\infty(z,\hat{z})&=1\wedge\frac{\int_{\mathbb{S}^d}\pi_S(y)Q_S(z,y) \dee y}{\int_{\mathbb{S}^d}\pi_S(y)Q_S(\hat{z},y) \dee y}\\
            &\geq 1\wedge\frac{\int_{\HS(z,\epsilon')\backslash\AC(\epsilon)}\pi_S(y)Q_S(z,y) \dee y}{M}\\
            &\geq 1\wedge\frac{\delta_\epsilon\delta'_{\epsilon'}S_{\epsilon,\epsilon',d}}{M}.
        \end{aligned}
    \end{equation}
    Therefore, 
    \begin{equation}
        P_\infty(z,\hat{z})\geq\frac{\delta_\epsilon Q_S(z,\hat{z})}{M}\left(1\wedge\frac{\delta_\epsilon\delta'_{\epsilon'}S_{\epsilon,\epsilon',d}}{M}\right).
    \end{equation}
    In the locally-balanced case, similarly when $\hat{z}\notin\AC(\epsilon)$, we have
    \begin{equation}
        \begin{aligned}
            Q_\infty(z,\hat{z})&=\frac{\sqrt{\pi_S(\hat{z})}Q_S(z,\hat{z})}{\int_{\mathbb{S}^d}\sqrt{\pi_S(y)}Q_S(z,y) \dee y}\\
            &\geq\frac{\sqrt{\delta_\epsilon}Q_S(z,\hat{z})}{\sqrt{M}},\\
            \alpha_\infty(z,\hat{z})&=1\wedge\sqrt{\frac{\pi_S(\hat{z})}{\pi_S(z)}}\cdot\frac{\int_{\mathbb{S}^d}\sqrt{\pi_S(y)}Q_S(z,y) \dee y}{\int_{\mathbb{S}^d}\sqrt{\pi_S(y)}Q_S(\hat{z},y) \dee y}\\
            &\geq1\wedge\sqrt{\frac{\delta_\epsilon}{M}}\cdot\frac{\sqrt{\delta_\epsilon}\delta'_{\epsilon'}S_{\epsilon,\epsilon',d}}{\sqrt{M}}=1\wedge\frac{\delta_\epsilon\delta'_{\epsilon'}S_{\epsilon,\epsilon',d}}{M}.
        \end{aligned}
    \end{equation}
    Thus we have
    \begin{equation}
        P_\infty(z,\hat{z})\geq\frac{\sqrt{\delta_\epsilon}Q_S(z,\hat{z})}{\sqrt{M}}\left(1\wedge\frac{\delta_\epsilon\delta'_{\epsilon'}S_{\epsilon,\epsilon',d}}{M}\right).
    \end{equation}
\end{proof}

\subsection{Proof of \cref{prop:infinity lower bound}}\label{proof:prop:infinity lower bound}
\begin{proof}
    In the globally-balanced case, $Q_{\infty}(z,\hat{z})=\frac{\pi_S(\hat{z})Q_S(z,\hat{z})}{\int_{\mathbb{S}^d}\pi_S(y)Q_S(z,y) \dee y}$, so the acceptance rate can be written as 
    \begin{equation}
        \alpha_{\infty}(z,\hat{z})=1\wedge\frac{\int_{\mathbb{S}^d}\pi_S(y)Q_S(z,y) \dee y}{\int_{\mathbb{S}^d}\pi_S(y)Q_S(\hat{z},y) \dee y}.
    \end{equation}
    On the one hand, 
    \begin{equation}
        \begin{aligned}
            \int_{\mathbb{S}^d}\pi_S(y)Q_S(z,y) \dee y&\geq\int_{\mathbb{S}^d\backslash\AC(\epsilon)}\pi_S(y)Q_S(z,y) \dee y\\
            &\geq\int_{\HS(z,\epsilon')\backslash\AC(\epsilon)}\delta_{\epsilon}Q_S(z,y) \dee y\\
            &\geq\int_{\HS(z,\epsilon')\backslash\AC(\epsilon)}\delta_{\epsilon}\delta'_{\epsilon'} \dee y\\
            &\geq\delta_{\epsilon}\delta'_{\epsilon'}S_{\epsilon,\epsilon',d}>0.
        \end{aligned}
    \end{equation}
    On the other hand, notice that we have $M=\sup_z\pi_S(z)<\infty$ if we set $\alpha>2d$ or $\beta>0$ in the target distribution  $\pi(x)\propto\|x\|^{-\alpha}$ or $\pi(x)\propto \exp(-\|x\|^{\beta})$, then
    \begin{equation}
        \int_{\mathbb{S}^d}\pi_S(y)Q_S(\hat{z},y) \dee y\leq M\int_{\mathbb{S}^d}Q_S(\hat{z},y) \dee y=M.
    \end{equation}
    Therefore, in the globally-balanced case, we have
    \begin{equation}
        \alpha_{\infty}(z,\hat{z})\geq 1\wedge\frac{\delta_{\epsilon}\delta'_{\epsilon'}S_{\epsilon,\epsilon',d}}{M}:=\xi>0.
    \end{equation}
    In the locally-balanced case, 
    \begin{equation}
        \alpha_{\infty}(z,\hat{z})=1\wedge\sqrt{\frac{\pi_S(\hat{z})}{\pi_S(z)}}\cdot\frac{\int_{\mathbb{S}^d}\sqrt{\pi_S(y)}Q_S(z,y) \dee y}{\int_{\mathbb{S}^d}\sqrt{\pi_S(y)}Q_S(\hat{z},y) \dee y}.
    \end{equation}
    Similarly,
    \begin{equation}
        \int_{\mathbb{S}^d}\sqrt{\pi_S(y)}Q_S(z,y) \dee y\geq\sqrt{\delta_\epsilon}\delta'_{\epsilon'}S_{\epsilon,\epsilon',d}>0.
    \end{equation}
    And
    \begin{equation}
        \int_{\mathbb{S}^d}\sqrt{\pi_S(y)}Q_S(\hat{z},y) \dee y\leq \sqrt{M}\int_{\mathbb{S}^d}Q_S(\hat{z},y) \dee y=\sqrt{M}.
    \end{equation}
    When our target distribution is $\pi(x)\propto\|x\|^{-\alpha}$ or $\pi(x)\propto \exp(-\|x\|^{\beta})$, $\pi(x)$ is only the function of $\|x\|$, then the isotropy of $\pi(x)$ implies that $\pi_S(z)$ is only the function of $z_{d+1}$ and when $1/\|x\|=o(h)$ or $z_{d+1}^2=1-o(h^2)$, according to \citet[Lemma 3.1]{yang2024}, if $h=O(d^{-1/2})$, then we have
    \begin{equation}
        \frac{\hat{z}_{d+1}}{z_{d+1}}=\frac{1}{\sqrt{1+h^2(d-1)}}\left(1-\frac{1}{2}h^2U^2\right)+O_\mathbb{P}(d^{-1/2}),
    \end{equation}
    where $U\sim\mathcal{N}(0,1)$. Then for any $\varepsilon>0$, there exists $C_\varepsilon,D_\varepsilon$ such that when $d>D_\varepsilon$ we have
    \begin{equation}
        \begin{aligned}
            &\mathbb{P}\left(\frac{\hat{z}_{d+1}}{z_{d+1}}-\frac{1}{\sqrt{1+h^2(d-1)}}\left(1-\frac{1}{2}h^2U^2\right)>C_\varepsilon d^{-1/2}\right)<\varepsilon,\\
            &\Rightarrow\mathbb{P}\left(\frac{\hat{z}_{d+1}}{z_{d+1}}>C_\varepsilon d^{-1/2}+\frac{1}{\sqrt{1+h^2(d-1)}}\left(1-\frac{1}{2}h^2U^2\right)\right)<\varepsilon.
        \end{aligned}
    \end{equation}
    Since 
    \begin{equation}
        \lim_{d\to\infty}C_\varepsilon d^{-1/2}+\frac{1}{\sqrt{1+h^2(d-1)}}\left(1-\frac{1}{2}h^2U^2\right)\leq 1,
    \end{equation}
    for any $\delta>0$, there exists $D_\delta$ such that when $d>D_\delta$
    \begin{equation}
        C_\varepsilon d^{-1/2}+\frac{1}{\sqrt{1+h^2(d-1)}}\left(1-\frac{1}{2}h^2U^2\right)<1+\delta.
    \end{equation}
    Now when $d>\mathscr{D}_\varepsilon:=\max\{D_\varepsilon,D_\delta\}$, we have
    \begin{equation}
        \mathbb{P}\left(\frac{\hat{z}_{d+1}}{z_{d+1}}>1+\delta\right)\leq\mathbb{P}\left(\frac{\hat{z}_{d+1}}{z_{d+1}}>C_\varepsilon d^{-1/2}+\frac{1}{\sqrt{1+h^2(d-1)}}\left(1-\frac{1}{2}h^2U^2\right)\right)<\varepsilon.
    \end{equation}
    Let $\delta\to0$ and we get that
    \begin{equation}
        \mathbb{P}\left(\frac{\hat{z}_{d+1}}{z_{d+1}}>1\right)<\varepsilon,
    \end{equation}
    which implies that with probability no smaller than $1-\varepsilon$, $\hat{z}_{d+1}\leq z_{d+1}$ and $\pi_S(\hat{z})/\pi_S(z)\geq 1$. Therefore,
    \begin{equation}
        \alpha_{\infty}(z,\hat{z})\geq 1\wedge\frac{\sqrt{\delta_\epsilon}\delta'_{\epsilon'}S_{\epsilon,\epsilon',d}}{\sqrt{M}}:=\xi>0.
    \end{equation}
\end{proof}

\subsection{Proof of \cref{prop:almost linear gap}}\label{proof:prop:almost linear gap}
Before giving the proof of \cref{prop:almost linear gap}, we first show that the proposal distribution of SRWM in $\mathbb{R}^d$ has finite moments to the order of $d-1$.
\begin{lemma}
\label{lemma:SPS finite moments}
    Denote $\widetilde{Q_S}(x,\cdot)$ as the proposal distribution of SRWM on $\mathbb{R}^d$, then for any $x$, $\widetilde{Q_S}(x,\cdot)$ has finite moments up to order $d-1$.
\end{lemma}
\begin{proof}
     Denote $Q_S(z,\cdot)$ as the proposal distribution of SRWM on $\mathbb{S}^d$. If $x=\SP(z)$ then by Jacobian we have
    \begin{equation}
        \widetilde{Q_S}(x,\hat{x})\propto Q_S(z,\hat{z})(R^2+\|\hat{x}\|^2)^{-d},
    \end{equation}
    where $\hat{x}=\SP(\hat{z})$. According to \citet[Lemma A.14]{milinanni2026rapid}, the explicit formula of $Q_S$  can be written as
    \begin{equation}
        Q_S(z,\hat{z})\propto(1+\|w\|)^{\frac{d+1}{2}}\exp(-\frac{\|w\|^2}{2h^2}).
    \end{equation} 
    To finish the proof, we consider the following three cases:
    \begin{itemize}
        \item  If $z_{d+1}<0$ i.e. $\|x\|<R$, then $\hat{z}$ cannot be the north pole. In this case, $\|\hat{x}\|$ is bounded and $\widetilde{Q_S}(x,\cdot)$ has a bounded support. Any probability distribution with bounded support has all moments finite.
        \item   If $z_{d+1}>0$, i.e. $\|x\|>R$, then $\hat{z}$ reaches the north pole (i.e. $\hat{x}\to\infty$) only when $w,e_{d+1}$ are coplanar, $w^Te_{d+1}>0$ and $\|w\|=\frac{2\|x\|R}{\|x\|^2-R^2}$ (see \cref{fig:geo1}), where $e_{d+1}$ is the unit vector of $(d+1)$-th coordinate. In this case, $Q_S(z,\hat{z})$ is a finite constant and the tail of $\widetilde{Q_S}(x,\hat{x})$ when $\|\hat{x}\|\to\infty$ behaves like $\|\hat{x}\|^{-2d}$, which implies that $\widetilde{Q_S}(x,\cdot)$ has finite moments up to order $d-1$.
        \item   If $z_{d+1}=0$, i.e. $\|x\|=R$, then by \cref{fig:geo2} the relationship between $\|w\|$ and $\|\hat{x}\|$ can be written as $\|w\|=\frac{\|\hat{x}\|^2-R^2}{2\|\hat{x}\|R}$. Therefore, $\widetilde{Q_S}(x,\cdot)$ has the exponential tail:
    \begin{equation}
        \widetilde{Q_S}(x,\hat{x})\propto\left(1+\frac{\|\hat{x}\|}{R}-\frac{R}{2\|\hat{x}\|}\right)^{\frac{d+1}{2}}\left(R^2+\|\hat{x}\|^2\right)^{-d}\exp\left(-\frac{1}{2h^2}\left(\frac{\|\hat{x}\|}{2R}-\frac{R}{2\|\hat{x}\|}\right)^2\right),
    \end{equation}
    and thus has all finite moments.
    \end{itemize} 
    \begin{figure}[h]

    \begin{subfigure}{0.48\textwidth}
    \includegraphics[width=1\linewidth]{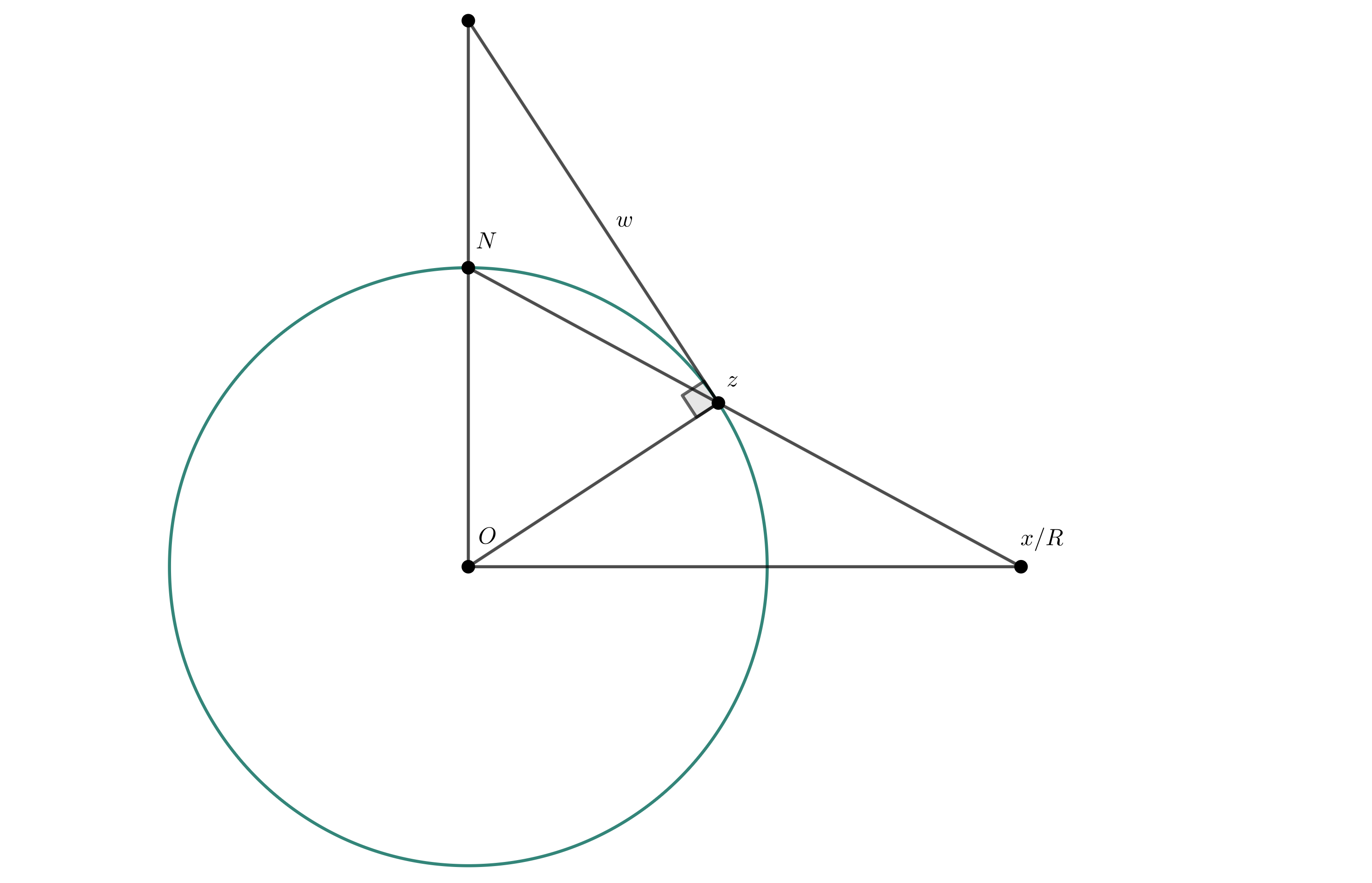} 
    \caption{$z_{d+1}>0$}
    \label{fig:geo1}
    \end{subfigure}
    \begin{subfigure}{0.48\textwidth}
    \includegraphics[width=1\linewidth]{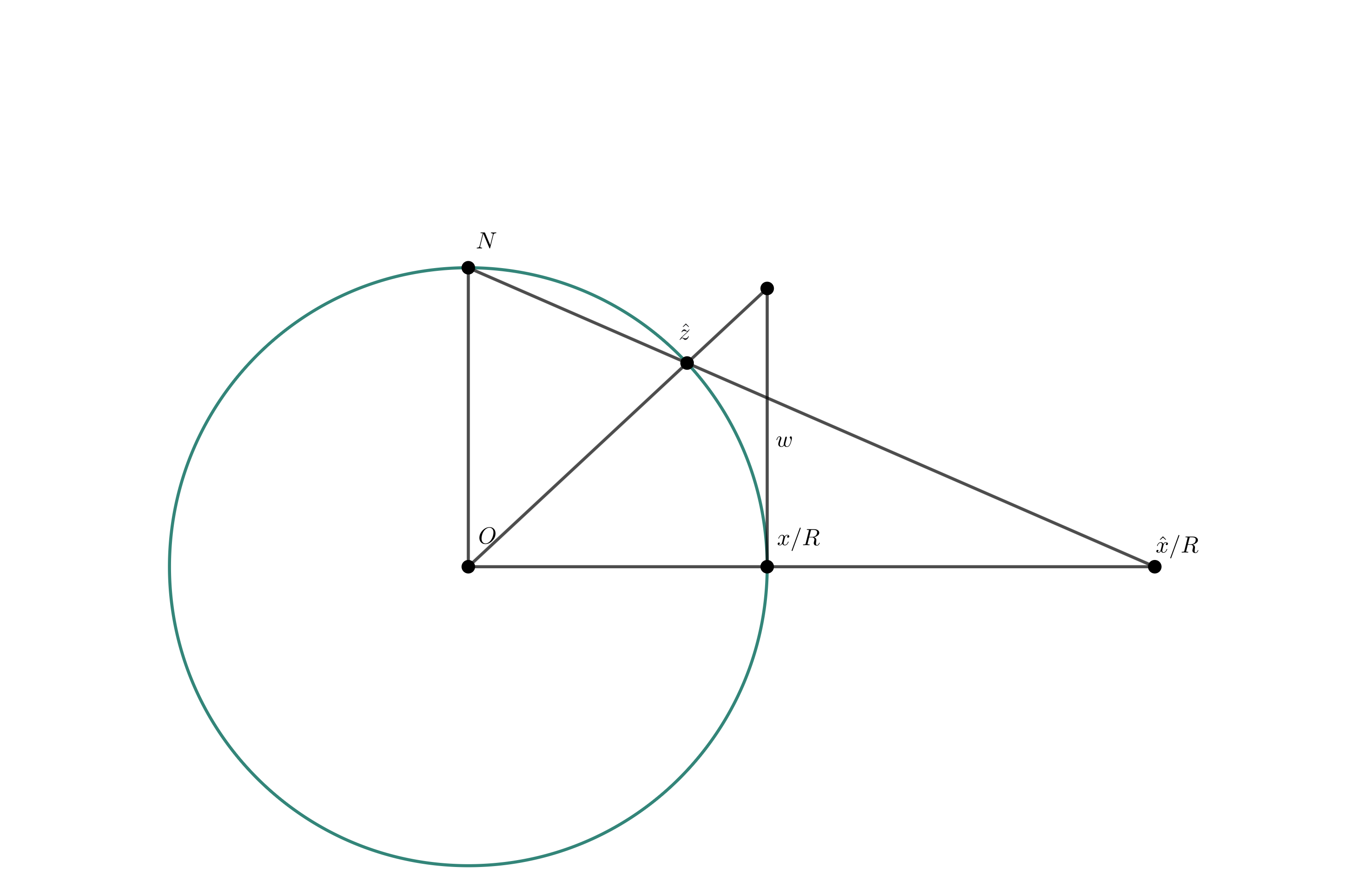}
    \caption{$z_{d+1}=0$}
    \label{fig:geo2}
    \end{subfigure}

    \caption{Illustrations of $x$, $\hat{x}$ and $w$}
    \end{figure}

\end{proof}
\begin{proof}[Proof of Proposition \ref{prop:almost linear gap}]
    By \citet[Theorem 2]{pozza2024fundamental} we have that
    \begin{equation}
        \Gap(P^{(N)})\leq\min\left(2N\mathop{\essinf}\limits_{x \in \mathbb{R}^d}\tilde{P}(x,\mathbb{R}^d\backslash\{x\}),\inf_{\nu\in\mathbb{R}^{d},\|\nu\|=1}\frac{\mathbb{E}[\max_{i=1,\dots,N}(\nu^T(\hat{X}_i-X))^2]}{2\Var(\nu^TX)}\right),
    \end{equation}
    where $X$ follows the Student's $t$ distribution and $\hat{X}_i\mid X\sim\widetilde{Q_S}(X,\cdot)$ are the candidates in the stationary phase. Denote $A_i^2=(\nu^T(\hat{X}_i-X))^2$, then by Jensen's inequality
    \begin{equation}
        (\mathbb{E}[\max_{i=1,\dots,N}A_i^2])^{\frac{p}{2}}\leq\mathbb{E}[\max_{i=1,\dots,N}A_i^{p}]\leq\sum_{i=1}^{N}\mathbb{E}[A_i^{p}].
    \end{equation}
    By \cref{lemma:SPS finite moments}, for any $p<d$
    \begin{equation}
        \begin{aligned}
            \mathbb{E}[A_i^{p}]&=\mathbb{E}[\mathbb{E}[(\nu^T(\hat{X}_i-X))^{p}\mid X]]\\
            &=\sum_{n=0}^{p}\binom{p}{n}\mathbb{E}[\mathbb{E}[(\nu^T\hat{X}_i)^{p-n}(-\nu^TX)^{n}\mid X]]\\
            &=\sum_{n=0}^{p}C(n)\mathbb{E}[(\nu^TX)^{n}]\\
            &\leq M^{\nu}(p)<\infty.
        \end{aligned}
    \end{equation}
    where $C(n)$ are constants. It follows that
    \begin{equation}
        \mathbb{E}[\max_{i=1,\dots,N}(\nu^T(\hat{X}_i-X))^2]\leq (N\mathbb{E}[A_i^{p}])^{\frac{2}{p}}\leq N^{\frac{2}{p}}M^\nu(p)^{\frac{2}{p}}.
    \end{equation}
    Thus, we have proved the proposition.
\end{proof}
\section{Choice of $N$}
\label{app:choose N}

For multiple-try algorithms, the optimal ESJD typically increases with $N$ (as an example, see \cref{table:optimal_parameters}), indicating that it is better to use a large $N$. However, increasing $N$ also increases the computational cost, so the optimal choice of $N$ requires balancing efficiency and computational cost. In practice, the optimal $N$ is problem-dependent, and there is no universal prescription.

In a serial implementation, we adopt the heuristic approach proposed by \cite{bedard2012scaling} and suggest $N=2$ as a reasonable default for SMTM: in \cref{table:optimal_parameters} and \cref{fig:choice_of_N} the improvement from $N=1$ to $N=2$ is the largest improvement, although this observation is problem-dependent and does not constitute a universal optimum. However, as discussed at the end of \cref{Sec:n to infty}, if one considers using parallel computing, SMTM might be more suited to such a context. In this sense, the choice of $N$ is often dictated by the hardware, such as parallel workers, rather than pure statistical theory. In this case, $N$ can be set equal to the number of available CPU cores or GPU threads.

We recommend fixing $N$ before the run and using the same value during warm-up and retained sampling. The experiments in \cref{fig:bayeisanlasso,fig:bayesianstudent} show that larger $N$ can improve robustness and convergence from a poor initialization, but they do not imply that a user must detect the end of burn-in and then change $N$. Keeping $N$ fixed removes this stopping decision and aligns the practical recommendation with the fixed-$N$ setting of the theoretical analysis. Detailed optimization of $N$ for a particular parallel architecture remains outside the scope of this paper.

\begin{table}[]
    \centering
    \begin{subtable}[h]{0.8\textwidth}
        \centering
        \begin{tabular}{llllll}
        $N$          & 1      & 2      & 3      & 4      & 5      \\ \hline
        $\ell^\ast$  & 23.17  & 25.34  & 27.60  & 29.90  & 31.20  \\
        $\ESJD^\ast$ & 135.94 & 221.27 & 284.22 & 338.22 & 378.95 \\
        $(\sum\alpha_2^j)^{\ast}$    & 0.24   & 0.32   & 0.36   & 0.36   & 0.38   \\ \hline
        \end{tabular}
        \caption{LB SMTM}
    \end{subtable}
    \begin{subtable}[h]{0.8\textwidth}
        \centering
        \begin{tabular}{llllll}
        $N$          & 1      & 2      & 3      & 4      & 5      \\ \hline
        $\ell^\ast$  & 23.17  & 26.06  & 28.32  & 28.97  & 31.48  \\
        $\ESJD^\ast$ & 135.94 & 230.93 & 299.92 & 359.95 & 405.31 \\
        $(\sum\alpha_2^j)^{\ast}$    & 0.24   & 0.32   & 0.37   & 0.39   & 0.40   \\ \hline
        \end{tabular}
        \caption{GB SMTM}
    \end{subtable}
    \caption{Optimal parameters with different $N$}
    \label{table:optimal_parameters}
\end{table}
\begin{figure}[htbp]
    \centering
    \includegraphics[width=0.6\linewidth]{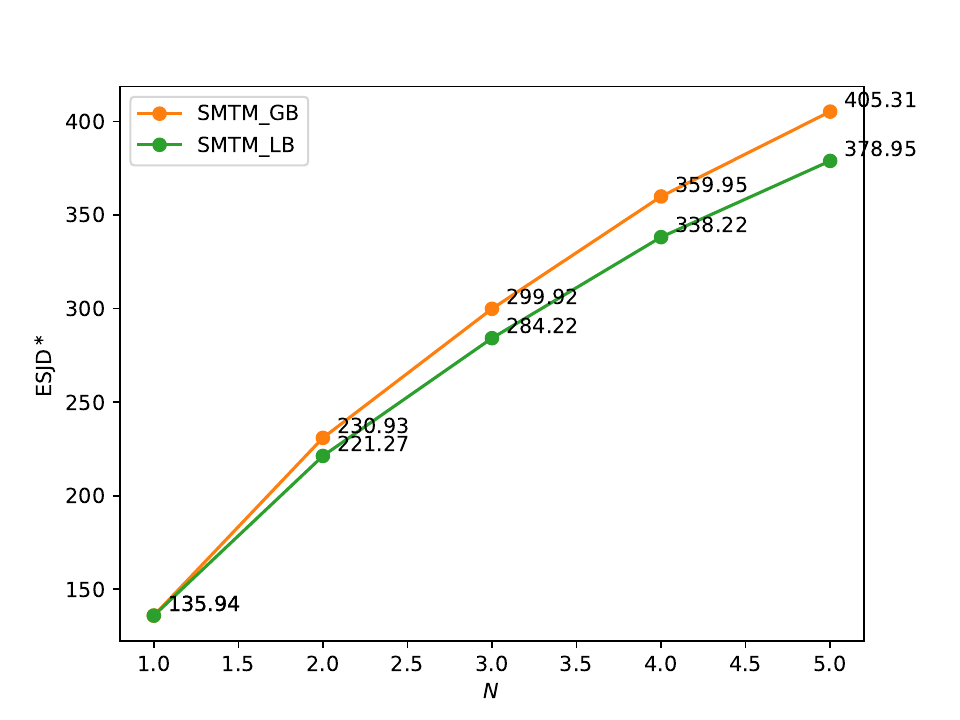}
    \caption{Optimal $\ESJD$ with different $N$}
    \label{fig:choice_of_N}
\end{figure}

\end{appendix}

\bibliographystyle{imsart-nameyear}
\bibliography{bibliography}

\end{document}